\documentclass[journal]{IEEEtran}

\usepackage{mathrsfs}
\usepackage{caption2}
\usepackage{url}
\usepackage{stfloats}
\usepackage{fixltx2e}
\usepackage{amsfonts,amssymb}
\usepackage{eqparbox}
\usepackage{mdwmath}
\usepackage{mdwtab}
\usepackage{array}
\usepackage[cmex10]{amsmath}
\usepackage{multirow}
\usepackage{graphicx}
\usepackage{algorithm}
\usepackage{algpseudocode}

\ifCLASSINFOpdf
\else
\fi
\begin{document}
\title{Hardware Impairments Aware Transceiver for Full-Duplex Massive MIMO Relaying}

\author{Xiaochen~Xia,
        Dongmei~Zhang,
        Kui~Xu,~\IEEEmembership{Member,~IEEE},
        Wenfeng~Ma, and
        Youyun Xu, ~\IEEEmembership{Senior Member,~IEEE}
\thanks{This work is supported by Major Research Plan of National Natural Science Foundation of China (No. 91438115), Special Financial Grant of the China Postdoctoral Science foundation (2015T81079), National Natural Science Foundation of China (No. 61371123, No. 61301165), Jiangsu Province Natural Science Foundation (BK2011002, BK2012055), China Postdoctoral Science Foundation (2014M552612) and Jiangsu Postdoctoral Science Foundation (No.1401178C).}
\thanks{The authors are with the Institute of Communication Engineering, PLA University of Science and Technology (email: tjuxiaochen@gmail.com, zhangdm72@163.com, \{lgdxxukui, lgdxmwf, lgdxxyy\}@126.com).}}


\maketitle

\begin{abstract}
This paper studies the massive MIMO full-duplex relaying (MM-FDR), where multiple source-destination pairs communicate simultaneously with the help of a common full-duplex relay equipped with very large antenna arrays. Different from the traditional MM-FDR protocol, a general model where sources/destinations are allowed to equip with multiple antennas is considered. In contrast to the conventional MIMO system, massive MIMO must be built with low-cost components which are prone to hardware impairments. In this paper, the effect of hardware impairments is taken into consideration, and is modeled using transmit/receive distortion noises. We propose a low complexity hardware impairments aware transceiver scheme (named as HIA scheme) to mitigate the distortion noises by exploiting the statistical knowledge of channels and antenna arrays at sources and destinations. A joint degree of freedom and power optimization algorithm is presented to further optimize the spectral efficiency of HIA based MM-FDR. The results show that the HIA scheme can mitigate the ``ceiling effect" appears in traditional MM-FDR protocol, if the numbers of antennas at sources and destinations can scale with that at the relay.
\end{abstract}


\IEEEpeerreviewmaketitle

\newtheorem{lemma}{Lemma}
\newtheorem{theorem}{Theorem}
\newtheorem{corollary}{Corollary}
\newtheorem{remark}{Remark}
\section{Introduction}

In multi-user MIMO systems, one main challenge is the increased complexity and energy consumption of the signal processing to mitigate the interferences between multiple co-channel users. To achieve energy efficient transmission, the multi-user MIMO system with very large antenna arrays at each base station (known as ``massive MIMO'' system) has been advocated recently \cite{IEEEhowto:Marz2010TWC}. The key result is that, with very large antenna arrays at each base station, both the intracell and intercell interferences can be substantially reduced with simple linear beamforming (BF) processing \cite{IEEEhowto:Marz2010TWC,IEEEhowto:Ngo2013TC}.

On the other hand, full-duplex relaying (FDR) is a promising approach to improve the spectral efficiency (SE) of relaying network while retains the merits of half-duplex relaying (HDR) (e.g., path loss reduction). In FDR, the relay transmits and receives simultaneously at the same frequency and time, but at the cost of a strong echo interference (EI) due to signal leakage between the relay output and input \cite{IEEEhowto:BharSigcomm13}. To mitigate EI, three approaches have been investigated, i.e., 1) passive cancellation, 2) time-domain cancellation \cite{IEEEhowto:BharSigcomm13} and 3) spatial suppression \cite{IEEEhowto:SuraweeraTW14,IEEEhowto:ShangAusCTW14}. The passive cancellation relies on a combination of path loss, cross-polarization and antenna directionality \cite{IEEEhowto:EverettTWC13}. The time-domain cancellation is based on the fact that EI signal is known at full-duplex node. Thus, it can be regenerated and removed in time-domain \cite{IEEEhowto:BharSigcomm13}. In spatial suppression, EI is mitigated with the multiple transmit/receive antennas by approaches such as null-space projection \cite{IEEEhowto:SuraweeraTW14}. Inspired by these works, a number of works have dedicated to the study of FDR protocol on both theory and testbed (See \cite{IEEEhowto:NguyenTSP13,IEEEhowto:SabharwalJSAC14} and the references therein). To achieve spectral and energy efficient transmissions of multiple source-destination pairs, recent works sought to incorporate both HDR \cite{IEEEhowto:CuiTWC14,IEEEhowto:SuraICC13} and FDR \cite{IEEEhowto:NgoICC13}-\cite{IEEEhowto:XXCWCNC14} with massive MIMO.

However, the aforementioned works on massive MIMO are actually based on the assumption that the base stations or relays are equipped with a large number of high-quality transmit/receive radio frequency (RF) chains (which are expensive and power-hungry). In contrast to conventional MIMO system (e.g., at most 8 antennas in LTE system), massive MIMO must be built with low-cost components \cite{IEEEhowto:LarssonCM14} since the deploy cost and energy consumption of circuits will increase dramatically as the number of antennas grows very large. Such low-cost components are prone to hardware imperfections (e.g., phase noise, nonlinear power amplifier, I/Q imbalance, nonlinear low-noise amplifier and ADC impairments), which must be considered in the design of practical massive MIMO system.

This paper focuses on the transceiver design for massive MIMO full-duplex relaying (MM-FDR) with hardware impairments. The effect of hardware impairments is modeled using transmit/receive distortion noises \cite{IEEEhowto:DayJSAC12,IEEEhowto:ZhengTWC13}. There are several challenges in the design of practical transceiver scheme in the considered system. The first is: \emph{how to deal with the EI cancellation without instantaneous EI channel?} EI cancellation is a critical problem in MM-FDR transceiver design which is not only important to reduce EI power, but also useful to reduce distortion noises caused by hardware impairments at the relay (as will be shown in section IV). Different from FDR with small-scale relay antenna arrays \cite{IEEEhowto:BharSigcomm13}-\cite{IEEEhowto:SabharwalJSAC14}, the instantaneous EI channel is usually not easy to obtain in MM-FDR. This is because the learning of EI channel requires training sequence with length not less than the number of relay antennas \cite{IEEEhowto:Ngo2013TC}, which is prohibitive in MM-FDR as the channel coherent time is limited. The lack of EI channel makes traditional EI cancellation techniques (e.g., time-domain cancellation and spatial suppression) difficult to apply. Although passive cancellation does not relies on EI channel estimation, it usually cannot provide satisfactory performance when used alone \cite{IEEEhowto:SabharwalJSAC14}. Another problem is: \emph{How to suppress distortion noises caused by hardware impairments at sources and destinations.} Different from multi-user interference (MUI), the distortion noise caused by transmit imperfection of source can be viewed as an interference signal with the same channel as transmit data. Thus, it cannot be suppressed by relay antenna arrays during coherent combining (Similar problem appears in reception at destinations). This causes performance ceiling on achievable rate as the number of relay antennas grows large, which degrades the gain of massive MIMO significantly.

In this paper, we propose practical transceiver scheme for MM-FDR with hardware impairments considering the above problems. Different from the traditional MM-FDR protocol \cite{IEEEhowto:NgoJSAC14} (where sources and destinations are equipped with single antenna), we consider a general model where sources and destinations are allowed to equip with multiple antennas. The contributions are summarized as follows:
\begin{itemize}
  \item We first examine the limitation of traditional MM-FDR protocol under hardware impairments. In particular, we derive the upper bound on end-to-end achievable rate for traditional MM-FDR protocol with linear processing at the relay. The bound reveals that the achievable rate is limited by the hardware impairments at the sources and destinations, and performance ceiling appears as the number of relay antenna grows large. The result also implies that it is impossible to cancel the ``ceiling effect'' with linear processing, if sources and destinations are only equipped with single antenna.
  \item Based on the upper bound analysis, we propose a hardware impairments aware transceiver scheme (HIA scheme) to mitigate the distortion noises by exploiting the statistical channel knowledge and antennas arrays of sources/destinations. The scheme needs no instantaneous knowledge of EI channel. The asymptotic end-to-end achievable rate of MM-FDR with HIA scheme (HIA-MM-FDR) is derived and the scaling behaviors of MUI, EI and distortion noises are determined.
  \item A joint degree of freedom and power optimization (JDPO) algorithm is presented to further improve the SE of HIA-MM-FDR.
\end{itemize}

The paper is organized as follows. In section II, we review the related work. The system model is described in section III. The upper bound of achievable rate is analyzed in section IV. The HIA scheme is proposed in section V and the JDPO algorithm is presented in section VI. Simulation results are presented in section VII and some conclusions will be drawn at last. \underline{Notations}: $\mathbb E(\cdot)$ and ${\rm var}(\cdot)$ denote the expectation and variance. ${\bf I}_n$ is $n \times n$ identity matrix. $(\cdot)^{*}$, $(\cdot)^{T}$ and $(\cdot)^{H}$ denote conjugate, transpose and conjugate-transpose, respectively. $\rho(\bf A)$, ${\rm Tr} (\bf A)$, $[{\bf A}]_{ij}$, $\lambda_l(\bf A)$ and ${\bf u}_l(\bf A)$ denote the spectral radius, trace, $(i,j)$th element of matrix, $l$th largest eigenvalue and eigenvector with respect to $l$th largest eigenvalue of matrix $\bf A$, respectively. $a$ scales with $b$ means $0 < \mathop {\lim }\limits_{b \to \infty } \frac{a}{b} < \infty $.

\section{Related Work}

The design and performance of using unlimited number of antennas at the base station in cellular system were first considered for independent antennas \cite{IEEEhowto:Marz2010TWC,IEEEhowto:Ngo2013TC}, and soon extended to the scenarios with spatial correlated antennas \cite{IEEEhowto:HoydisJSAC13}-\cite{IEEEhowto:AdhikaryIT13}. The asymptotic SINR for single cell cellular system was analyzed in \cite{IEEEhowto:Ngo2013TC}. It has been shown that, with very large antenna arrays at the base station, a deterministic SINR (also called the ``deterministic equivalent'' of SINR) which depends only on the large-scale fading of channels can be achieved, if the transmit power is scaled by $N$ with perfect CSI and $\sqrt{N}$ with imperfect CSI ($N$ denotes the number of base station antennas). The authors in \cite{IEEEhowto:HoydisJSAC13}-\cite{IEEEhowto:YinJSSP14} have done considerable work to derive the ``deterministic equivalent'' of SINR for massive MIMO system with spatial correlated antennas.

In the field of massive MIMO relaying, the SE and EE of HDR with very large relay antenna arrays were investigated in \cite{IEEEhowto:SuraICC13,IEEEhowto:CuiTWC14}. The MM-FDR with decode-and-forward (DF) relay was first introduced in \cite{IEEEhowto:NgoICC13,IEEEhowto:NgoJSAC14}, and analyzed in end-to-end achievable rate as linear processing is employed. The asymptotic performance of amplify-and-forward based MM-FDR was considered in \cite{IEEEhowto:XXCWCNC14} and the scaling behavior of the infinitely repeating echo interference was determined.

Only a few works considered the effect of hardware impairments on massive MIMO system. A constant envelope signal design has been proposed in \cite{IEEEhowto:MohammedTC13} to facilitate the use of power-efficient RF power amplifiers. The authors in \cite{IEEEhowto:StuderJSAC13} presented a low peak-to-average-power ratio (PAR) precoding solution to enable efficient implementation using non-linear RF components in massive MIMO system. In FDR with small-scale antennas, the optimal precoding under limited ADC dynamic range was studied in \cite{IEEEhowto:DayJSAC12,IEEEhowto:ZhengTWC13}. However, these works cannot provide much insight for the effect of hardware impairments on MM-FDR.

\section{System Model}

Consider the network with $K$ source-destination (S-D) pairs and a single full-duplex relay $R$, where source $S_k$ wishes to communicate with destination $D_k$ ($k\in \{1,\cdots, K\}$) with the help of $R$. The relay adopts the DF policy. It is assumed that the sources and destinations are equipped with $N_S$ and $N_D$ antennas respectively, while the relay is equipped with $N_R+N_T$ antennas ($N_R$ for reception and $N_T$ for transmission). We are interested in the large-$(N_R,N_T)$ regime, i.e., $\min\{N_R,N_T\}\to \infty$. $N_S$, $N_D$ and $K$ can be either fixed or scale with $\min\{N_R,N_T\}$.

Let ${{\bf{H}}_{SR,k}} \in {\mathbb C^{N_R \times N_S}}$ be the channel matrix from $S_k$ to receive antenna array of relay and let ${{\bf{H}}_{RD,k}} \in {\mathbb C^{N_T \times N_D}}$ be the channel matrix from $D_k$ to transmit array of relay. Let ${{\bf{H}}_{EI}} \in {\mathbb C^{N_R \times N_T}}$ denote the EI channel matrix between transmit and receive arrays of relay. The spatial correlation of each MIMO channel is characterized by the Kronecker model \cite{IEEEhowto:KotechaTSP04}. Thus, ${{\bf{H}}_{SR,k}}$ can be expressed as ${{\mathbf{H}}_{SR,k}} = \sqrt {{\beta _{SR,k}}} {\mathbf{C}}_{SR,k}^{1/2}{{\mathbf{X}}_{SR,k}}{\mathbf{\tilde C}}_{SR,k}^{1/2}$, where ${{\beta _{SR,k}}}$ represents the large-scale fading. ${\mathbf{C}}_{SR,k} \!\in\! {\mathbb C^{N_R \times N_R}} $ and ${\mathbf{\tilde C}}_{SR,k} \!\in\! {\mathbb C^{N_S \times N_S}} $ characterize the spatial correlation of received signals across receive array of relay and that of transmitted signals across transmit array of $S_k$. ${{\mathbf{X}}_{SR,k}} \!\in\! {\mathbb C^{N_R \times N_S}}$ consists of the random components of channels whose elements are i.i.d with distribution ${\cal CN}(0,1)$. Based on the same model, the channel matrices ${{\bf{H}}_{RD,k}}$ and ${{\bf{H}}_{EI}}$ can be expressed as ${{\mathbf{H}}_{RD,k}}\! \!=\!\! \sqrt {{\beta _{RD,k}}} {\mathbf{C}}_{RD,k}^{1/2}{{\mathbf{X}}_{RD,k}}{\mathbf{\tilde C}}_{RD,k}^{1/2}$ and ${{\mathbf{H}}_{EI}} \!\! =\!\!\sqrt {{\beta _{EI}}} {\mathbf{C}}_{EI}^{1/2}{{\mathbf{X}}_{EI}}{\mathbf{\tilde C}}_{EI}^{1/2}$. To facilitate the analysis, we assume the following conditions on correlation matrices, i.e., $\forall {\bf C} \in \{ {{{\mathbf{C}}_{SR,k}},{{{\mathbf{\tilde C}}}_{SR,k}},{{\mathbf{C}}_{RD,k}},{{{\mathbf{\tilde C}}}_{RD,k}},{{\mathbf{C}}_{EI}},{{{\mathbf{\tilde C}}}_{EI}}} \}$, where $k\in\{1,\cdots,K\}$,
\begin{itemize}
  \item \textbf{A1}: The spectral radius of ${\bf C}$ is bounded by a constant, i.e., $\rho({\bf C})\le C$.
  \item \textbf{A2}: ${\bf C}$ is a Hermitian and Teoplitz matrix and has unit diagonal elements.
\end{itemize}
The former is a common assumption in the studies of massive MIMO which follows from energy conservation \cite{IEEEhowto:HoydisJSAC13} and \textbf{A2} corresponds to the case of uniform linear {array\footnotemark[1]} (ULA) \cite{IEEEhowto:KotechaTSP04}.
{\footnotetext[1]{We will restrict our analysis to the scenario with ULA at each node. The analysis for arbitrary correlation matrices will be considered as further work.}}

To characterize the effect of hardware imperfections, we adopt the new signal model from \cite{IEEEhowto:DayJSAC12,IEEEhowto:ZhengTWC13}.

\subsubsection{Imperfect Transmit RF Chain}

We model the effect of imperfect transmit RF chain by adding, per transmit antenna, an independent zero-mean Gaussian ``transmit distortion noise'', whose power is proportional to the signal power transmitted at that antenna. The experimental results in \cite{IEEEhowto:SantellaVT98,IEEEhowto:SuzukiJSAC08} have shown that the independent Gaussian distortion noise model closely captures the joint effect of imperfect components in transmit RF chain. Let ${{\mathbf{ x}}_{S,k}}\left[ u \right] \in \mathbb C^{N_S\times 1}$ and ${{\mathbf{ x}}_R}\left[ u \right]\in \mathbb C^{N_T\times 1}$ denote the transmit vectors of source $S_k$ and relay at time instant $u$. Based on the above model, the distorted transmit signals can be expressed as
\begin{equation}\label{2-3}
{{\mathbf{\tilde x}}_{S,k}}\left[ u \right] = {{\mathbf{x}}_{S,k}}\left[ u \right] + {{\mathbf{t}}_{S,k}}\left[ u \right] \hspace{.5em} {\rm and} \hspace{.5em} {{\mathbf{\tilde x}}_R}\left[ u \right] = {{\mathbf{x}}_R}\left[ u \right] + {{\mathbf{t}}_R}\left[ u \right]
\end{equation}
where the distortion noises ${{\mathbf{t}}_{S,k}}\left[ u \right]\sim {\cal CN}\left( {0,{\nu _{S,k}}{\text{diag}}\left( {\mathbb{E}\left[ {{{\mathbf{x}}_{S,k}}\left[ u \right]{\mathbf{x}}_{S,k}^H\left[ u \right]} \right]} \right)} \right)$ and ${{\bf{t}}_R}\left[ u \right]\sim{\cal C}{\cal N}\left( 0,\nu _R \times \right.$ $\left.{\rm{diag}}\left( {{\mathbb{E}}\left[ {{{\bf{x}}_R}\left[ u \right]{{\bf{x}}_R^H}\left[ u \right]} \right]} \right) \right)$. Note that $\nu _{S,k}> 0$ ($\nu _R > 0$) characterizes the level of transmit imperfection. For example, $\nu _{S,k}= 0$ ($\nu _R=0$) corresponds to the conventional assumption of perfect transmit RF chains. The quality of transmit RF chains degrades as $\nu _{S,k}$ ($\nu _R$) increases.

\subsubsection{Imperfect Receive RF Chain}

We model the effect of imperfect receive RF chain by adding, per receive antenna, an independent zero-mean Gaussian ``receive distortion noise'', whose variance is proportional to the signal power received at that antenna. More precisely, let ${{\bf{y}}_R}\left[ u \right]\in \mathbb C^{N_R\times 1}$ and ${{\bf{y}}_{D,k}}\left[ u \right] \in \mathbb C^{N_D\times 1}$ be the undistorted received signals of the relay and destination $D_k$ at time instant $u$, the distorted received signals can be expressed as
\begin{equation}\label{2-4}
    \begin{aligned}
        & {{{\bf{\tilde y}}}_R}\left[ u \right] = {{\bf{y}}_R}\left[ u \right] + {{\bf{r}}_R}\left[ u \right] \hspace{.5em}  {\rm and} \hspace{.5em} {{{\bf{\tilde y}}}_{D,k}}\left[ u \right] = {{\bf{y}}_{D,k}}\left[ u \right] + {{\bf{r}}_{D,k}}\left[ u \right]
    \end{aligned}
\end{equation}
where the distortion noises ${{\bf{r}}_R}[ u ] \sim {\cal C}{\cal N}( {0,{\mu _R}{\rm{diag}}( {{\mathbb{E}}[ {{{\bf{y}}_R}[ u ]{{\bf{y}}_R^H}[ u ]} ]} )} )$ and ${{\mathbf{r}}_{D,k}}[ u ] \sim{\cal CN}( 0,{\mu _{D,k}} {\text{diag}}( {\mathbb{E}[ {{{\mathbf{y}}_{D,k}}[ u ]{\mathbf{y}}_{D,k}^H[ u ]} ]} ) )$. $\mu _R> 0$ ($\mu _{D,k} > 0$) characterizes the level of receive imperfection. The experimental studies in \cite{IEEEhowto:NamgoongTWC05} have shown that the independent Gaussian noise model is a good approximation to the joint effect of imperfect components in receive RF chain.

At time instant $u$, all sources transmit signals ${\bf x}_{S,k}[u]$ ($k\!=\!1,\cdots, K$) to the relay simultaneously. Meanwhile, the relay broadcasts the decoded signals to the destinations. Based on the models in (\ref{2-3}) and (\ref{2-4}), the received signals at the relay and $D_k$ can be expressed as
\begin{equation}\label{2-5}
    \begin{aligned}
        & {{{\mathbf{\tilde y}}}_R}\left[ u \right] = \sum\limits_{k = 1}^K {{\mathbf{H}}_{SR,k}}\left( {{{\mathbf{x}}_{S,k}}\left[ u \right] + {{\mathbf{t}}_{S,k}}\left[ u \right]} \right) \\
        &+ {{\mathbf{H}}_{EI}}\left( {{{\mathbf{x}}_R}\left[ u \right] + {{\mathbf{t}}_R}\left[ u \right]} \right) + {{\mathbf{r}}_R}\left[ u \right] + {{\mathbf{n}}_R}\left[ u \right] \\
        & {{\mathbf{\tilde y}}_{D,k}}\left[ u \right] = {\mathbf{H}}_{RD,k}^H\left( {{{\mathbf{x}}_R}\left[ u \right] + {{\mathbf{t}}_R}\left[ u \right]} \right) + {{\mathbf{r}}_{D,k}}\left[ u \right] + {{\mathbf{n}}_{D,k}}\left[ u \right]
    \end{aligned}
\end{equation}
To keep the complexity low, it is assumed that a singe data steam is transmitted at each source and each node employs only linear processing. In particular, $S_k$ transmits the unit-power symbol $s_k[u]$ using the unitary BF vector ${\bf p}_{S,k}$. Therefore, the transmit vector of $S_k$ can be expressed as ${\bf x}_{S,k}[u] =\sqrt{E_{S,k}}{\bf p}_{S,k} s_k[u]$, where $E_{S,k}$ denotes the transmit power. Relay combines the received signal by multiplying the receive BF matrix ${{\bf{W}}_R} = \left[ {{{\bf{w}}_{R,1}}, \cdots ,{{\bf{w}}_{R,K}}} \right]$, i.e., ${{{\bf{\mathord{\buildrel{\lower3pt\hbox{$\scriptscriptstyle\smile$}}
\over y} }}}_R}\left[ u \right] = {\bf{W}}_R^H{{{\bf{\tilde y}}}_R}\left[ u \right]$. The $k$th element of ${{{\bf{\mathord{\buildrel{\lower3pt\hbox{$\scriptscriptstyle\smile$}}
\over y} }}}_R}\left[ u \right]$
\begin{equation}\label{2-6}
    \begin{aligned}
        & {{\mathord{\buildrel{\lower3pt\hbox{$\scriptscriptstyle\smile$}} \over y} }_{R,k}}\left[ u \right]  = \sqrt{E_{S,k}}{\mathbf{w}}_{R,k}^H{{\mathbf{H}}_{SR,k}}{{\mathbf{p}}_{S,k}}{s_k}\left[ u \right] \\
        & +  \sqrt{E_{S,j}} \sum\limits_{j = 1,j \ne k}^K {{\mathbf{w}}_{R,k}^H{{\mathbf{H}}_{SR,j}}{{\mathbf{p}}_{S,j}}{s_j}\left[ u \right]}  +{\mathbf{w}}_{R,k}^H\sum\limits_{j = 1}^K {{{\mathbf{H}}_{SR,j}}{{\mathbf{t}}_{S,j}}\left[ u \right]} \\
        & + {\bf{w}}_{R,k}^H{{\bf{H}}_{EI}}\left( {{{\bf{x}}_R}\left[ u \right] + {{\bf{t}}_R}\left[ u \right]} \right) + {\bf{w}}_{R,k}^H{{\bf{r}}_R}\left[ u \right] + {\bf{w}}_{R,k}^H{{\bf{n}}_R}\left[ u \right]
    \end{aligned}
\end{equation}
is used to decode the symbol of $S_k$. The relay forwards the decoded symbol using transmit BF matrix ${{\bf{W}}_T} = \left[ {{{\bf{w}}_{T,1}}, \cdots ,{{\bf{w}}_{T,K}}} \right]$. The transmit vector of the relay can be expressed as ${{\mathbf{x}}_R}\left[ u \right]\! = \! {{\mathbf{W}}_T}{\mathbf{\Lambda} _R^{1/2}}{\mathbf{s}}\left[ {u - d} \right]$, where ${\mathbf{s}}\left[ {u - d} \right]\!=\![{{s_1}}\left[ {u - d} \right], \cdots,{{s_K}}\left[ {u - d} \right]]^T$ and ${\mathbf{\Lambda} _R}\!=\!{\rm diag}(E_{R,1},\cdots,E_{R,K})$ is the power allocation matrix at the relay. $d$ denotes the processing delay of the relay. To meet the relay's power constraint, ${{\mathbf{x}}_R}\left[ u \right]$ must satisfy ${\rm{Tr}}\!\left( {{\mathbb{E}}\left[ {{{\bf{x}}_R}\left[ u \right]{\bf{x}}_R^H\left[ u \right]} \right]} \right) \!\! =\!\! \sum_{l = 1}^K {{E_{R,l}}} \!\!\le\! \!E_R^{\rm max}$. $D_k$ uses the unitary BF vector ${\bf p}_{D,k}$ to combine the received signal ${{\mathbf{\tilde y}}_{D,k}}\left[ u \right]$. The combined signal is expressed as
\begin{equation}\label{2-7}
    \begin{aligned}
        & {{\tilde y}_{D,k}}\left[ u \right] =  \sqrt{E_{R,k}}{\mathbf{p}}_{D,k}^H{\mathbf{H}}_{RD,k}^H{{\mathbf{w}}_{T,k}}{s_k}\left[ {u - d} \right] \\
         & + {\mathbf{p}}_{D,k}^H\sum\limits_{j = 1,j \ne k}^K \sqrt{E_{R,j}} {{\mathbf{H}}_{RD,k}^H{{\mathbf{w}}_{T,j}}{s_j}\left[ {u - d} \right]}  \\
        & + {\mathbf{p}}_{D,k}^H{\mathbf{H}}_{RD,k}^H{{\mathbf{t}}_R}\left[ u \right] + {\mathbf{p}}_{D,k}^H{{\mathbf{r}}_{D,k}}\left[ u \right] + {\mathbf{p}}_{D,k}^H{{\mathbf{n}}_{D,k}}\left[ u \right]
    \end{aligned}
\end{equation}

\section{Limitation of MM-FDR with Single Antenna Sources and Destinations}

This section analyzes the upper bound on achievable rate of MM-FDR with {single\footnotemark[2]} antenna sources and destinations. The goal is to reveal the fundamental limitation of previous MM-FDR protocol with single antenna sources and destinations in combating distortion noises. Moreover, the bound provides us important insight on the design of practical transceiver scheme to mitigate the distortion noises, when multiple antennas are available at sources and destinations (See section V).
{\footnotetext[2]{The analysis for general case with multiple antennas at sources/destinations is more informative. However, the derivation is challenging since the BF vectors at sources, destinations and relay must be designed jointly to optimize the achievable rate, and closed-form solution is not available \cite{IEEEhowto:SuraweeraTW14}. The FDR with general number of antennas at each node was studied in \cite{IEEEhowto:ShomoronyIT14}-\cite{IEEEhowto:QianIT11}. However, their analyses assume perfect hardware, and thus give no insight on the effect of hardware impairments. Moreover, the methods in \cite{IEEEhowto:ShomoronyIT14}-\cite{IEEEhowto:QianIT11} cannot be applied in the considered system directly since the transmissions from sources to relay and relay to destinations are coupled due to the presence of distortion noises.}}

\subsection{Upper Bound on Achievable Rate}

As $N_S\!=\!N_D\!=\!1$, we set ${\bf p}_{S,k}\!=\!{\bf p}_{D,k}\!=\!1$, replace ${\bf H}_{i,k}$ ($i\in\{SR,RD\}$) with ${\bf h}_{i,k}$, and let ${\bf H}_{i}=[{\bf h}_{i,1},\cdots,{\bf h}_{i,K}]$ and ${{\mathbf{t}}_S}\left[ u \right] = {\left[ {{\mathbf{t}}_{S,1}^T\left[ u \right], \cdots ,{\mathbf{t}}_{S,K}^T\left[ u \right]} \right]^T}$. We consider a block-fading channel with coherence time $T$ (symbol times), where $\tau$ are used for uplink training of each source/destination, and the remaining $T-2K\tau$ are used for data transmission. The upper bound on achievable rate is obtained by assuming perfect knowledge of ${\bf h}_{SR,k}$, ${\bf h}_{RD,k}$ and ${\bf H}_{EI}$ can be provided with pilot signals, and meanwhile, the MUIs (i.e., the second terms of right-hand side of (\ref{2-6}) and (\ref{2-7})) and EI term ${\mathbf{w}}_{R,k}^H{{\mathbf{H}}_{EI}}{{\mathbf{x}}_R}\left[ u \right]$ can be somehow {cancelled\footnotemark[3]}. This gives us the following upper bound on end-to-end achievable rate
{\footnotetext[3]{With ${\bf H}_{EI}$, the EI signal ${\mathbf{w}}_{R,k}^H{{\mathbf{H}}_{EI}}{{\mathbf{x}}_R}\left[ u \right]$ can be perfectly canceled since ${\bf x}_R[u]$ is known at the relay. For MUIs, this can be realized by some sophisticated interference cancellation methods, e.g., the minimum mean square error with successive interference cancellation (MMSE-SIC).}}
\begin{equation}\label{3-0-1}
\textsf{R}_{k}^{{\text{Upper}}} = \mathop {\max }\limits_{\left\| {{{\mathbf{w}}_{R,k}}} \right\|=\left\| {{{\mathbf{w}}_{D,k}}} \right\| = 1} \min\left\{\textsf{R}_{SR,k}, \textsf{R}_{RD,k}\right\}
\end{equation}
where $\textsf{R}_{SR,k}$ and $\textsf{R}_{RD,k}$ denote the achievable rates of $S_k\to R$ and $R \to D_k$ channels, respectively. According to (\ref{2-6}), (\ref{2-7}) and the assumptions in the above, we have
\begin{equation}\label{3-0-2}
\begin{aligned}
& \textsf{R}_{SR,k} = {\frac{{T - 2K\tau }}{T}{{\log }_2}\left( {1 + \frac{{{E_{S,k}}{\mathbf{w}}_{R,k}^H{{\mathbf{h}}_{SR,k}}{\mathbf{h}}_{SR,k}^H{{\mathbf{w}}_{R,k}}}}{{{\mathbf{w}}_{R,k}^H{\mathbf{Q}}{{\mathbf{w}}_{R,k}}}}} \right)} \\
& \textsf{R}_{RD,k} = \frac{{T - 2K\tau }}{T}  \\
& \times {\log _2}\left( {1 + \frac{E_{R,k}{{\mathbf{w}}_{T,k}^H{{\mathbf{h}}_{RD,k}}{\mathbf{h}}_{RD,k}^H{{\mathbf{w}}_{T,k}}}}{{{\mathbf{h}}_{RD,k}^H{\mathbf{\Theta }}_R^\textsf{T}{{\mathbf{h}}_{RD,k}} + {\mathbb E\left[ {{{\left\| {{{\mathbf{r}}_{D,k}}\left[ u \right]} \right\|}^2}} \right]} + 1}}} \right)
\end{aligned}
\end{equation}
where ${{\mathbf{Q}}}$ is given as ${\mathbf{Q}} = {{\mathbf{H}}_{SR}}{{\bf\Theta } _S^{\textsf{T}}}{\mathbf{H}}_{SR}^H + {{\mathbf{H}}_{EI}}{{\bf\Theta } _R^{\textsf{T}}}{\mathbf{H}}_{EI}^H + {\mathbf{\Theta }}_R^\textsf{R} + {{\mathbf{I}}_{{N_R}}}$ with ${{\bf\Theta } _S^{\textsf{T}}} = \mathbb E\left[{\bf t}_S[u]{\bf t}^H_S[u]\right]$, ${{\bf\Theta } _R^{\textsf{T}}}= \mathbb E\left[{\bf t}_R[u]{\bf t}^H_R[u]\right]$ and ${\mathbf{\Theta }}_R^\textsf{R}= \mathbb E\left[{\bf r}_R[u]{\bf r}^H_S[u]\right]$ denoting the covariance matrices of distortion noises ${\bf t}_S[u]$, ${\bf t}_R[u]$ and ${\bf r}_R[u]$, respectively.

\emph{1) Achievable Rate of $S_k \to R$ Channel:} Since ${\bf w}_{R,k}$ is only related with $\textsf{R}_{SR,k}$, it can be optimized separately to maximize $\textsf{R}_{SR,k}$. From (\ref{3-0-2}), the optimization problem is equivalent to the generalized Rayleigh quotient problem, which can be solved as ${{\mathbf{w}}_{R,k}} = \frac{{{{\bf{Q}}^{ - 1}}{{\bf{h}}_{SR,k}}}}{{{\rm{ }}\left\| {{{\bf{Q}}^{ - 1}}{{\bf{h}}_{SR,k}}} \right\|}}$. The resultant upper bound on $\textsf{R}_{SR,k}$ (denoted by $\textsf{R}_{SR,k}^{{\text{Upper}}}$) can be expressed as
\begin{equation}\label{3-1}
\begin{aligned}
& \textsf{R}_{SR,k}^{{\text{Upper}}} = \frac{{T - 2K\tau }}{T}{\log _2}\left( {1 + {E_{S,k}}{\mathbf{h}}_{SR,k}^H{{\mathbf{Q}}^{ - 1}}{{\mathbf{h}}_{SR,k}}} \right) \\
& = \frac{{T - 2K\tau }}{T}{\log _2}\left( {1 + \frac{{{E_{S,k}}{\mathbf{h}}_{SR}^H{\mathbf{Q}}_k^{ - 1}{{\mathbf{h}}_{SR,k}}}}{{1 + {\nu _{S,k}}{E_{S,k}}{{\mathbf{h}}_{SR,k}}{\mathbf{Q}}_k^{ - 1}{\mathbf{h}}_{SR,k}^H}}} \right)
\end{aligned}
\end{equation}
where $\mathbf{Q}_k$ is defined as ${{\mathbf{Q}}_k} = {\mathbf{Q}} - {\nu _{S,k}}{p_{S,k}}{{\mathbf{h}}_{SR,k}}{\mathbf{h}}_{SR,k}^H$ and is independent of ${\bf h}_{SR,k}$. The second step follows from the matrix inversion lemma \cite{IEEEhowto:Silverstein95}.

\begin{theorem}\label{t1}
Assuming that the assumptions \textbf{A1} and \textbf{A2} hold. As $N_S=N_D=1$, the achievable rate of $S_k \to R$ channel in the large-$(N_R,N_T)$ regime is bounded as
\begin{equation}\label{3-2}
\begin{aligned}
& \textsf{R}_{SR,k} \le \textsf{R}_{SR,k}^{{\text{Upper}}} =  \frac{{T - 2K\tau }}{T} \\
& \times {\log _2}\left( {1 + \frac{{{E_{S,k}}{\beta _{SR,k}}{\text{Tr}}\left( {{{\mathbf{C}}_{SR,k}}{{\mathbf{\Psi }}_k}} \right)}}{{1 + {\nu _{S,k}}{E_{S,k}}{\beta _{SR,k}}{\text{Tr}}\left( {{{\mathbf{C}}_{SR,k}}{{\mathbf{\Psi }}_k}} \right)}}} \right)
\end{aligned}
\end{equation}
where $\rho =\beta_{EI}{\text{Tr}} \left( {{{{\mathbf{\tilde C}}}_{EI}}{\mathbf{\Theta }}_R^\textsf{T}} \right) + 1$, ${{{\mathbf{\Psi }}_k}\left( \rho  \right)}$ is determined by the following fixed-point algorithm with $\delta _{k,l}^{\left( 0 \right)}\left( \rho  \right) = \frac{1}{\rho }$
\begin{equation}\label{3-3}
\begin{array}{ll}
  {{\mathbf{\Psi }}_k}\left( \rho  \right) = {\left( {\frac{1}{{{N_R}}}\sum\limits_{l = 1,l \ne k}^K {\frac{{{\nu _{S,l}}{p_{S,l}}{\beta _{SR,l}}{{\mathbf{C}}_{SR,l}}}}{{1 + {\delta _{k,l}}\left( \rho  \right)}} + {\mathbf{\Theta }}_R^\textsf{R} + \rho {{\mathbf{I}}_{{N_R}}}} } \right)^{ - 1}}\hfill \\
  {\mathbf{\Theta }}_R^\textsf{R }  = {\mu _R}\sum\limits_{l = 1}^K {{\nu _{S,l}}{p_{S,l}}{\beta _{SR,l}}{{\mathbf{I}}_{{N_R}}}}  + {\beta _{EI}}{\text{Tr}}\left( {{{{\mathbf{\tilde C}}}_{EI}}{\mathbf{\Theta }}_R^\textsf{T}} \right){{\mathbf{I}}_{{N_R}}} + {\mu _R}{{\mathbf{I}}_{{N_R}}}\hfill \\
  {\delta _{k,l}}\left( \rho  \right) = \mathop {\lim }\limits_{n \to \infty } \delta _{k,l}^{\left( n \right)}\left( \rho  \right)  \hfill \\
 \delta _{k,l}^{\left( n \right)}\left( \rho  \right) = \frac{{{\nu _{S,l}}{p_{S,l}}{\beta _{SR,l}}}}{{{N_R}}} \\
 \times{\text{Tr}}\left( {{{\mathbf{C}}_{SR,l}}{{\left( {\frac{1}{{{N_R}}}\sum\limits_{j = 1,j \ne k}^K {\frac{{{\nu _{S,j}}{p_{S,j}}{\beta _{SR,j}}{{\mathbf{C}}_{SR,j}}}}{{1 + \delta _{k,l}^{\left( {n - 1} \right)}\left( \rho  \right)}}}  + {\mathbf{\Theta }}_R^\textsf{R} + \rho {{\mathbf{I}}_{{N_R}}}} \right)}^{ - 1}}} \right) \hfill \\
\end{array}
\end{equation}
\end{theorem}

\begin{proof}
See Appendix-B.
\end{proof}

The convergence of the fixed-point algorithm in Theorem \ref{t1} has been proved in \cite{IEEEhowto:WagnerIT12}.

\emph{2) Achievable Rate of $R \to D_k$ Channel:} From (\ref{3-0-2}), the achievable rates of $S_k \to R$ and $R \to D_k$ channels are coupled through ${{\bf\Theta } _R^{\textsf{T}}}= \mathbb E\left[{\bf t}_R[u]{\bf t}^H_R[u]\right]= {\rm diag}\left(\mathbb E\left[{\bf W}_T {\bf \Lambda}_R{\bf W}_T^H\right]\right)$, which makes the design of ${\bf w}_{T,k}$ very challenging. However, with the following theorem we show that this coupling disappears in the large-$(N_R,N_T)$ regime, which allows a simple upper bound for achievable rate of $R \to D_k$ channel.

\begin{theorem}\label{t2}
Assuming that the assumptions \textbf{A1} and \textbf{A2} hold. As $N_S=N_D=1$, the achievable rate of $S_k \to R$ channel is independent of ${\bf w}_{T,k}$. The achievable rate of $R \to D_k$ channel in the large-$(N_R,N_T)$ regime is bounded as (\ref{3-3-1}), which is achieved by the eigen BF ${{\mathbf{w}}_{T,k}} = {{{{\mathbf{h}}_{RD,k}}} \mathord{\left/
 {\vphantom {{{{\mathbf{h}}_{RD,k}}} {\left\| {{{\mathbf{h}}_{RD,k}}} \right\|}}} \right.
 \kern-\nulldelimiterspace} {\left\| {{{\mathbf{h}}_{RD,k}}} \right\|}}$.
\end{theorem}

\begin{proof}
See Appendix-C.
\end{proof}

\newcounter{mytempeqncnt}
\begin{figure*}
\normalsize
\setcounter{mytempeqncnt}{\value{equation}}
\setcounter{equation}{10}
\begin{equation}\label{3-3-1}
\textsf{R}_{RD,k} \le \textsf{R}_{RD,k}^{\rm Upper} \!=\! \frac{{T - 2K\tau }}{T}  {\log _2} \left( 1+ \frac{{{N_T}{\beta _{RD,k}}{E_{R,k}}}}{{{\nu _R}{\beta _{RD,k}}\sum_{l = 1}^K {{E_{R,l}}}  + {\mu _{D,k}}\left( {{N_T}{\beta _{RD,k}}{E_{R,k}} + {\nu _R}{\beta _{RD,k}}\sum_{l = 1}^K {{E_{R,l}}}  + 1} \right) + 1}}\right)
\end{equation}
\setcounter{equation}{\value{mytempeqncnt}}
\hrule
\end{figure*}

With Theorems \ref{t1} and \ref{t2}, the upper bound of end-to-end achievable in the large-$(N_R,N_T)$ regime can be expressed as
\setcounter{equation}{11}
\begin{equation}\label{3-4}
\textsf{R}_{k}^{{\text{Upper}}} = \min\left\{\textsf{R}_{SR,k}^{{\text{Upper}}}, \textsf{R}_{RD,k}^{{\text{Upper}}}\right\}
\end{equation}

\begin{remark}\label{rm1}
The bound in (\ref{3-4}) is derived by assuming linear processing and fixed transmit powers at sources and relay, which is in general not capacity achieving. One exception is when $K \ll \min\{N_R,N_T\}$ and $\min\{N_R,N_T\}\to \infty$. From Theorem \ref{t2}, the design of receive and transmit BF matrices is decoupled as $\min\{N_R,N_T\}\to \infty$. Moreover, the distortion noises are assumed to be circularly symmetric complex Gaussian distributed and independent of the desired signal. Thus, based on the results in \cite{IEEEhowto:TelatarETT99,IEEEhowto:BjCL13}, linear processing along with power control is sufficient to achieve the capacity. When $K \ll \min\{N_R,N_T\}$, the terms $\frac{1}{{{N_R}}}\sum_{l = 1,l \ne k}^K {\frac{{{\nu _{S,l}}{E_{S,l}}{\beta _{SR,l}}{{\mathbf{C}}_{SR,l}}}}{{1 + {\delta _{k,l}}\left( \rho  \right)}}} $ and $\frac{1}{{{N_R}}}\sum_{j = 1,j \ne k}^K {\frac{{{\nu _{S,j}}{E_{S,j}}{\beta _{SR,j}}{{\mathbf{C}}_{SR,j}}}}{{1 + \delta _{k,l}^{\left( {n - 1} \right)}\left( \rho  \right)}}}$ in (\ref{3-3}) vanish. By neglecting the low order term (in $\min\{N_R,N_T\}$), the bound in (\ref{3-4}) reduces to $\frac{{T - 2K\tau }}{T}{\log _2}( {1 + \min \{ {{{{v_{S,k}^{-1}}}},{{{\mu _{D,k}^{-1}}}}} \}} )$, which is independent of $E_{S,k}$ and $E_{R,k}$. Thus, (\ref{3-4}) is the upper bound for capacity of $S_k\to R\to D_k$ channel as $K \ll \min\{N_R,N_T\}$ and $\min\{N_R,N_T\}\to \infty$.
\end{remark}


\subsection{Limitation with Hardware Impairments}

Different from Remark \ref{rm1}, we consider the general case in which $K$ can be either fixed or scales with $\min\{N_R,N_T\}$. With Theorem \ref{t1} and Theorem \ref{t2}, the effect of hardware impairments is still hard to analysis since ${{\mathbf{\Psi }}_k} $ is not in closed-form. To make the analysis tractable, we assume ${{\mathbf{C}}_{SR,l}} = {{\mathbf{C}}_{SR}}$ $\forall l\in\{1,\cdots,K\}$. Using eigenvalue decomposition ${{\mathbf{C}}_{SR}} = {\mathbf{U}}{\Sigma }_{SR}{{\mathbf{U}}^H}$ on $\delta _{k,l}^{\left( n \right)}(\rho)$ in (\ref{3-3}), we have
\begin{equation}\label{3-5}
\begin{aligned}
& \delta _{k,l}^{\left( n \right)}\left( \rho  \right) = \frac{{{\nu _{S,l}}{p_{S,l}}{\beta _{SR,l}}}}{{{N_R}}} {\text{Tr}}\left( {\Sigma _{SR}} \right. \\
& \left. \times{{\left( {\frac{1}{{{N_R}}}\sum\limits_{j = 1,j \ne k}^K {\frac{{{\nu _{S,j}}{p_{S,j}}{\beta _{SR,j}}}}{{1 + \delta _{k,l}^{\left( {n - 1} \right)}\left( \rho  \right)}}} {\Sigma _{SR}} + {\mathbf{\Theta }}_R^\textsf{R} + \rho {{\mathbf{I}}_{{N_R}}}} \right)}^{ - 1}} \right)
\end{aligned}
\end{equation}
By examining (\ref{3-5}) and the expression of ${\mathbf{\Theta }}_R^\textsf{R }$ in (\ref{3-3}), $\delta _{k,l}^{\left( n \right)}(\rho) $ is an ${\cal O}(1)$ term. Note that this is valid for arbitrary $n>0$. Thus, we can conclude that $ {\delta _{k,l}}(\rho)={\cal O}(1)$. Replacing ${\delta _{k,l}}\left( \rho  \right)$ in ${{\bf{\Psi }}_k}$ with the symbol ${\cal O}(1)$ and substituting the result into (\ref{3-2}), $\textsf{R}_{SR,k}$ can be written as (\ref{3-6}).

\begin{figure*}
\normalsize
\setcounter{mytempeqncnt}{\value{equation}}
\setcounter{equation}{13}
\begin{equation}\label{3-6}
\begin{aligned}
\textsf{R}_{SR,k}^{{\text{Upper}}} =\frac{{T - 2K\tau }}{T}{\log _2}\left( {1 + \frac{{{E_{S,k}}{\beta _{SR,k}}\sum\limits_{l = 1}^{{N_R}} {{\lambda _l}\left( {{{\mathbf{C}}_{SR}}} \right){{\left( {\frac{{{\lambda _l}\left( {{{\mathbf{C}}_{SR}}} \right)}}{{{N_R}}}\sum\limits_{l = 1,l \ne k}^K {\frac{{{\nu _{S,l}}{E_{S,l}}{\beta _{SR,l}}}}{{1 + O\left( 1 \right)}} + \frac{{{\text{Tr}}\left( {{\mathbf{\Theta }}_R^\textsf{R}} \right)}}{{{N_T}}} + \rho } } \right)}^{ - 1}}} }}{{1 + {\nu _{S,k}}{p_{S,k}}{\beta _{SR,k}}\sum\limits_{l = 1}^{{N_R}} {{\lambda _l}\left( {{{\mathbf{C}}_{SR}}} \right){{\left( {\frac{{{\lambda _l}\left( {{{\mathbf{C}}_{SR}}} \right)}}{{{N_R}}}\sum\limits_{l = 1,l \ne k}^K {\frac{{{\nu _{S,l}}{E_{S,l}}{\beta _{SR,l}}}}{{1 + O\left( 1 \right)}} + \frac{{{\text{Tr}}\left( {{\mathbf{\Theta }}_R^\textsf{R}} \right)}}{{{N_T}}} + \rho } } \right)}^{ - 1}}} }}} \right)
\end{aligned}
\end{equation}
\setcounter{equation}{\value{mytempeqncnt}}
\hrule
\end{figure*}

For convenience, we further assume $\left\{ {{\beta _{SR,l}},{\beta _{RD,l}},{E_{S,l}},{E_{R,l}},{\nu _{S,l}},{\mu _{D,l}}} \right\} = \left\{ {{\beta _{SR}},{\beta _{RD}},{E_S},{E_R},{\nu _S},{\mu _D}} \right\}$. Note that this assumption does not effect the basis conclusions of the analysis. By neglecting the low order terms (in $\min\{N_R,N_T\}$) in (\ref{3-6}) and (\ref{3-3-1}), the upper bound on end-to-end achievable rate reduces to
\setcounter{equation}{14}
\begin{equation}\label{3-7}
\begin{aligned}
&\textsf{R}_{k}^{{\text{Upper}}} = \frac{{T - 2K\tau }}{T} \\
& \times {\log _2}\left( 1 + \min \left\{ \frac{1}{{{\nu _S} + \frac{K}{{{N_R}}}{\mu _R}{\nu _S} + \frac{K}{{{N_R}}}{\nu _R}\frac{{{E_R}{\beta _{EI}}}}{{{E_S}{\beta _{SR}}}}}}, \right. \right. \\
& \hspace{8em}\left.\left. \frac{1}{{{\mu _{D,k}} + \frac{K}{{{N_T}}}{\nu _R}\left( {1 + {\mu _{D,k}}} \right)}} \right\} \right)
\end{aligned}
\end{equation}

One can make several observations from (\ref{3-7}):

1) As $K\ll \min\{N_R,N_T\}$, the effect of hardware impairments at the relay and EI disappears in the large-$(N_R,N_T)$ regime and the upper bound converges to that in Remark \ref{rm1}. The end-to-end achievable rate is limited by distortion noises caused by hardware impairments at sources and destinations. The explanation is that the transmit distortion can be viewed as an interference signal the same channel as data of $S_k$ from (\ref{2-5}). Thus, it cannot be suppressed after combining by ${\bf w}_{R,k}$. The power of distortion due to receive imperfection of $D_k$ is proportional to that of desired receive signal (after eigen BF) and also cannot be reduced by the transmit BF scheme of relay. Thus, there is a finite ceiling on end-to-end achievable rate as $\min\{N_R,N_T\}\to \infty$. In fact, this poses a major limitation on MM-FDR with single antenna sources and destinations, which degrades the gain of massive array of relay greatly.

2) As $K$ scales with $\min\{N_R,N_T\}$, the effect of hardware impairments at the relay and EI is not negligible. Since the EI channel is typically stronger than desired channels, the $S_k\to R$ channel becomes the bottleneck of end-to-end achievable rate. This indicates that it is of great importance to suppress EI to avoid the bottleneck effect.

\section{Hardware Impairments Aware Transceiver}

In this section, we consider a general model where sources and destinations are equipped with an arbitrary number of antennas. We first present a low complexity hardware impairments aware transceiver scheme (referred to as HIA scheme). Then the achievable rate of MM-FDR with proposed HIA scheme (HIA-MM-FDR) is analyzed in the large-$(N_R,N_T)$ regime and the effect of hardware impairments is discussed.

\subsection{Transceiver Scheme Description}

The optimal transceiver scheme for MM-FDR to optimize achievable rate is hard to find since the transmissions of $S_k \!\to \!R$ and $R \!\to\! D_k$ channels are deeply coupled due to the presence of EI and distortion noises. Even without EI and distortion noises, the problem has been proved NP-hard \cite{IEEEhowto:NguyenTSP13}. As ${\bf H}_{EI}$ is an $N_R\! \times \!N_T$ matrix, the learning of ${\bf H}_{EI}$ requires training sequence with length no less than $\min\{N_R,N_T\}$. This task is prohibitive as the duration of channel coherent time is limited. Thus, we assume that passive EI cancellation \cite{IEEEhowto:EverettTWC13} has been used and instantaneous EI channel is not available. A low complexity HIA scheme is proposed to mitigate EI (after passive cancellation) and distortion noises. As will be shown in simulations, the HIA scheme along with passive EI cancellation provides satisfactory gain compared to HDR.

\emph{1) Transceiver Design at the Relay:} As the variance of EI channel is typically stronger than desired channels, it is of great importance to control the EI power. From section IV-B, EI suppression is also important in reducing distortion noise due to hardware impairments of relay. We consider a two-stage BF scheme at the relay. The receive and transmit BF matrices are expressed as ${\bf W}_R = {\bf P}_R{\bf \underline{W}}_R$ and ${\bf W}_T = {\bf P}_T{\bf \underline{W}}_T$.

\emph{Outer BF matrices:} The outer BF matrices ${\bf P}_R\!  \in \! {\mathbb C}^{N_R\times A_R}$ ($K \! \le \! A_R \! \le\!  N_R$) and ${\bf P}_T \! \in \! {\mathbb C}^{N_T\times A_T}$ ($K \! \le  \! A_T \! \le  \! N_T$) are designed to suppress EI. $A_R$ and $A_T$ are design parameters. Due to the lack of ${\bf H}_{EI}$, ${\bf P}_R$ and ${\bf P}_T$ are allowed to depend only on the statistical knowledge of EI channels. With (\ref{2-6}), we design an average EI power minimization problem, i.e.,
\begin{equation}\label{4-1}
\begin{aligned}
        & \left({\bf \dot{P}}_R,{\bf \dot{P}}_T\right) = \arg \mathop {\min}\limits_{{{\mathbf{P}}_R},{{\mathbf{P}}_T}}\\
         & \mathbb E\bigg[\underbrace{ {{\mathbf{\underline{w}}}_{R,k}^H{\mathbf{P}}_R^H{{\mathbf{H}}_{EI}}\left( {{{\mathbf{P}}_T}{{\mathbf{\underline{W}}}_T}{\Lambda _R}{\mathbf{\underline{W}}}_T^H{\mathbf{P}}_T^H + {\mathbf{\Theta }}_R^\textsf{T}} \right){\mathbf{H}}_{EI}^H{{\mathbf{P}}_R}{{\mathbf{\underline{w}}}_{R,k}}} }_{{\textsf{EI}}_k}\bigg]
 \end{aligned}
\end{equation}
In (\ref{4-1}), the design of ${\bf P}_R$ and ${\bf P}_T$ is coupled with the inner BF matrix, which makes closed-form solution inaccessible. Thus, it is of great interest to find new target function to decouple the problem. For such, using Cauchy-Schwarz inequality twice, an upper bound on average EI power can be obtained as
\begin{equation}\label{4-2}
\begin{aligned}
&\textsf{EI}_k  \le \textsf{EI}_k^{\rm Upper} \\
& = \left\| {{{\mathbf{\underline{w}}}_{R,k}}} \right\|^2{\text{Tr}}\left( {{{\mathbf{\underline{W}}}_T}{\Lambda _R}{\mathbf{\underline{W}}}_T^H} \right){\text{Tr}}\left( {{\mathbf{P}}_R^H{{\mathbf{H}}_{EI}}{{\mathbf{P}}_T}{\mathbf{P}}_T^H{\mathbf{H}}_{EI}^H{{\mathbf{P}}_R}} \right)  \\
& + {\nu _R}\sum\limits_{l = 1}^K {{E_{R,k}}}  \left\| {{{\mathbf{\underline{w}}}_{R,k}}} \right\|^2{\text{Tr}}\left( {{\mathbf{P}}_R^H{{\mathbf{H}}_{EI}}{\mathbf{H}}_{EI}^H{{\mathbf{P}}_R}} \right) \\
\end{aligned}
\end{equation}
Then the problem (\ref{4-1}) can be rewritten approximately as
\begin{equation}\label{4-3}
\begin{aligned}
  & \left( {{{\mathbf{\dot{P}}}_R},{{\mathbf{\dot{P}}}_T}} \right)  = {\text{arg}}\mathop {{{\min}}}\limits_{{{\mathbf{P}}_R},{{\mathbf{P}}_T}} \left\{ {\text{Tr}}\left( {{{\mathbf{\underline{W}}}_T}{\Lambda _R}{\mathbf{\underline{W}}}_T^H} \right) \right.    \\
  &  \times \mathbb E\left[ {{\text{Tr}}\left( {{\mathbf{P}}_R^H{{\mathbf{H}}_{EI}}{{\mathbf{P}}_T}{\mathbf{P}}_T^H{\mathbf{H}}_{EI}^H{{\mathbf{P}}_R}} \right)} \right] \\
  & \bigg. +{  {\nu _R}\sum\limits_{l = 1}^K {{E_{R,k}}} \mathbb E\left[ {{\text{Tr}}\left( {{\mathbf{P}}_R^H{{\mathbf{H}}_{EI}}{\mathbf{H}}_{EI}^H{{\mathbf{P}}_R}} \right)} \right]} \bigg\}  \\
  & = {\text{arg}}\mathop {{{\min}}}\limits_{{{\mathbf{P}}_R},{{\mathbf{P}}_T}} \left\{ {\text{Tr}}\left( {{{\mathbf{\underline{W}}}_T}{\Lambda _R}{\mathbf{\underline{W}}}_T^H} \right){\text{Tr}}\left( {{\mathbf{P}}_T^H{{{\mathbf{\tilde C}}}_{EI}}{{\mathbf{P}}_T}} \right)\right. \\
  & \bigg. \times{\text{Tr}}\left( {{\mathbf{P}}_R^H{{\mathbf{C}}_{EI}}{{\mathbf{P}}_R}} \right)+ {\nu _R}{N_T}\sum\limits_{l = 1}^K {{E_{R,k}}} {\text{Tr}}\left( {{\mathbf{P}}_R^H{{\mathbf{C}}_{EI}}{{\mathbf{P}}_R}} \right) \bigg\}
\end{aligned}
\end{equation}
The multiplicative term $\left\| {{{\mathbf{w}}_{R,k}}} \right\|^2$ has been removed since it has no effect on the optimal solution. The second step follows from a similar derivation of (\ref{A3}). From (\ref{4-3}), the columns of ${\bf \dot{P}}_R$ (${\bf \dot{P}}_T$) are composed of the eigenvectors corresponding to the $A_R$ ($A_T$) smallest eigenvalues of ${{\mathbf{C}}_{EI}}$ (${{\mathbf{\tilde C}}_{EI}}$), i.e.,
\begin{equation}\label{4-3-1}
\begin{aligned}
& {\bf \dot{P}}_R = \left[ {{{\mathbf{u}}_{{N_R} - {A_R} + 1}}\left( {{{\mathbf{C}}_{EI}}} \right),{{\mathbf{u}}_{{N_R} - {A_R} + 2}}\left( {{{\mathbf{C}}_{EI}}} \right), \cdots ,{{\mathbf{u}}_{{N_R}}}\left( {{{\mathbf{C}}_{EI}}} \right)} \right] \\
& {\bf \dot{P}}_T = \left[ {{{\mathbf{u}}_{{N_T} - {A_T} + 1}}\left( {{{{\mathbf{\tilde C}}}_{EI}}} \right),{{\mathbf{u}}_{{N_T} - {A_T} + 2}}\left( {{{{\mathbf{\tilde C}}}_{EI}}} \right), \cdots ,{{\mathbf{u}}_{{N_T}}}\left( {{{{\mathbf{\tilde C}}}_{EI}}} \right)} \right]
\end{aligned}
\end{equation}
In (\ref{4-3-1}), $A_R$ and $A_T$ can be viewed as parameters to balance the allowable EI power and available degree of freedom (DOF) for data transmission, which should be optimized with respect to specific metric. This will be considered in section VI.

\emph{Inner BF Matrix:} The inner BF matrices ${\bf \underline{W}}_R \in \mathbb C^{A_R\times K}$ and ${\bf \underline{W}}_T\in \mathbb C^{A_T\times K}$ are designed to realize the multi-user communication. ${\bf \underline{W}}_R $ and ${\bf \underline{W}}_T $ can be designed with different criteria, e.g., maximizing the desired signal power which corresponding to the eigen BF, or minimizing the MUI which corresponding to the zero-forcing (ZF) scheme \cite{IEEEhowto:NgoJSAC14}. We adopt the latter one since the ZF scheme is known to approach the asymptotic limit (in $\min\{N_R,N_T\}$) of achievable rate faster as the number of relay antennas increases \cite{IEEEhowto:CuiTWC14,IEEEhowto:SuraICC13}. For given BF vectors at sources and destinations, define effective channel vectors of $S_k\to R$ and $R \to D_k$ channels as ${\bf \underline{h}}_{SR,k}={\mathbf{P}}_R^H{{\mathbf{H}}_{SR,k}}{{\mathbf{p}}_{S,k}}$ and ${\bf \underline{h}}_{RD,k}={\mathbf{P}}_T^H{{\mathbf{H}}_{RD,k}}{{\mathbf{p}}_{D,k}}$, and let ${{\mathbf{\underline{H}}}_{SR}} = \left[ {{{\mathbf{\underline{h}}}_{SR,1}}, \cdots ,{{\mathbf{\underline{h}}}_{SR,K}}} \right]$ and ${{\mathbf{\underline{H}}}_{RD}} = \left[ {{{\mathbf{\underline{h}}}_{RD,1}}, \cdots ,{{\mathbf{\underline{h}}}_{RD,K}}} \right]$, the inner BF matrices can be written as
\begin{equation}\label{4-4}
\begin{aligned}
& {{\mathbf{\underline{W}}}_R} = {{\mathbf{\underline{H}}}_{SR}}{\left( {{\mathbf{\underline{H}}}_{SR}^H{{\mathbf{\underline{H}}}_{SR}}} \right)^{ - 1}} \\
& {{\mathbf{\underline{W}}}_T} = {{\mathbf{\underline{H}}}_{RD}}{\left( {{\mathbf{\underline{H}}}_{RD}^H{{\mathbf{\underline{H}}}_{RD}}} \right)^{ - 1}}{\mathbf{\Upsilon }}^{-1/2}
\end{aligned}
\end{equation}
where ${\mathbf{\Upsilon  }}$ is a diagonal normalized matrix with $[{\mathbf{\Upsilon  }}]_{l,l}={{\mathbf{e}}_l^H{{\left( {{\mathbf{\underline{H}}}_{RD}^H{{\mathbf{\underline{H}}}_{RD}}} \right)}^{ - 1}}{{\mathbf{e}}_l}}$. In practice, ${\bf H}_{SR,k}$ and ${\bf H}_{RD,k}$ should be estimated in order to compute (\ref{4-4}). This will be considered in section V-B. Note that (\ref{4-4}) requires $K \le \min\{N_R,N_T\}$ so that $A_R$ and $A_T$ can be selected to ensure the invertibility of ${{\mathbf{\underline{H}}}_{SR}^H{{\mathbf{\underline{H}}}_{SR}}}$ and ${{\mathbf{\underline{H}}}_{RD}^H{{\mathbf{\underline{H}}}_{RD}}}$.

\emph{2) Transceiver Design at the Sources and Destinations:} Similar to \cite{IEEEhowto:NgoJSAC14,IEEEhowto:HoydisJSAC13}, we assume no instantaneous knowledge of channels at sources and destinations. This is reasonable since the amount of feedback could be very huge and unaffordable in MM-FDR. However, it is assumed that the local statistical knowledge of channels (i.e., $({\bf C}_{SR,k},{\bf \tilde C}_{SR,k})$ for $S_k$ and $({\bf C}_{RD,k},{\bf \tilde C}_{RD,k})$ for $D_k$) can be obtained. As observed from (\ref{3-7}), when $N_S=N_D=1$, the achievable rates of $S_k \to R$ channel and $R \to D_k$ channel are limited by the transmit distortion noise at $S_k$ and receive distortion noise at $D_k$. This motivates us to design ${\bf p}_{S,k}$ and ${\bf p}_{D,k}$ to suppress these negative factors with the antenna arrays of $S_k$ and $D_k$.

\emph{Design of ${\bf p}_{S,k}$:} Intuitively, according to (\ref{2-5}) we can design the following problem
\begin{equation}\label{4-5-0}
\begin{aligned}
{{{\bf{\dot p}}}_{S,k}} & = \arg \mathop {\max}\limits_{\|{{\bf{p}}_{S,k}}\|=1} \frac{E_{S,k}{\mathbb E\left[ {{\rm{Tr}}\left( {{{\bf{H}}_{SR,k}}{{\bf{p}}_{S,k}}{\bf{p}}_{S,k}^H{\bf{H}}_{SR,k}^H} \right)} \right]}}{{\mathbb E\left[ {{\rm{Tr}}\left( {{\bf{H}}_{SR,k}{\bf{\Theta }}_{S,k}^\textsf{T}{{\bf{H}}_{SR,k}^H}} \right)} \right]}}\\
& = \arg \mathop {\max}\limits_{\|{{\bf{p}}_{S,k}}\|=1} \frac{{{\bf{p}}_{S,k}^H{{{\bf{\tilde C}}}_{SR,k}}{{\bf{p}}_{S,k}}}}{{{\nu _{S,k}}{\rm{Tr}}\left( {{{{\bf{\tilde C}}}_{SR,k}}{\rm{diag}}\left( {{\bf{p}}_{S,k}^H{{\bf{p}}_{S,k}}} \right)} \right)}} \\
& = \arg \mathop {\max}\limits_{\|{{\bf{p}}_{S,k}}\|=1} \frac{1}{{{\nu _{S,k}}}}{\bf{p}}_{S,k}^H{{{\bf{\tilde C}}}_{SR,k}}{{\bf{p}}_{S,k}}
\end{aligned}
\end{equation}
where ${{\bf{\Theta }}_{S,k}^\textsf{T}}$ is the covariance matrix of ${\bf t}_{S,k}$, i.e., ${{\bf{\Theta }}_{S,k}^\textsf{T}}= \mathbb E [{\bf t}_{S,k}[u]{\bf t}^H_{S,k}[u]]= \nu_{S,k}{\rm diag}\left({\bf p}_{S,k} {\bf p}^H_{S,k}\right)$. The second step follows from a similar derivation with that in (\ref{A3}) and the last step is based on assumption \textbf{A2}. From (\ref{2-5}), the right-hand side of (\ref{4-5-0}) can be interpreted as the average signal to transmit distortion noise ratio at $S_k$. The solution of problem (\ref{4-5-0}) is
\begin{equation}\label{4-6-0}
{{{\bf{\dot p}}}_{S,k}} = {{\bf{u}}_1}\left( {{{{\bf{\tilde C}}}_{SR,k}}} \right)
\end{equation}

\emph{Design of ${\bf p}_{D,k}$:} Similarly, based on (\ref{2-5}), the received BF vector at $D_k$ can be derived by solving the problem ${{\mathbf{\dot p}}_{D,k}} = \arg \mathop {\max }\limits_{\| {{{\mathbf{p}}_{R,k}}}\| = 1} \frac{{{E_{R,k}}\mathbb{E}\left[ {{\text{Tr}}\left( {{{\mathbf{H}}_{RD,k}}{{\mathbf{p}}_{D,k}}{\mathbf{p}}_{D,k}^H{\mathbf{H}}_{RD,k}^H} \right)} \right]}}{{{\mu _{D,k}}{\mathbf{p}}_{D,k}^H{\text{diag}}\left( {\mathbb{E}\left[ {{\mathbf{H}}_{RD,k}^H{{\mathbf{H}}_{RD,k}}} \right]} \right){{\mathbf{p}}_{D,k}}}}$, which results in
\begin{equation}\label{4-7-0}
{{{\bf{\dot p}}}_{D,k}} = {{\bf{u}}_1}\left( {{{{\bf{\tilde C}}}_{RD,k}}} \right)
\end{equation}

\subsection{Reduced Dimension Channel Estimation}

During a training phase, each source/destination transmits pilot sequence sequentially which allows the relay to compute the estimates of channels. With HIA scheme, it is sufficient to estimate the effective channel vectors ${\bf \underline{h}}_{SR,k}$ and ${\bf \underline{h}}_{RD,k}$. This allows a reduced dimension estimation {scheme\footnotemark[4]}. For brevity, we consider the estimation of ${\bf \underline{h}}_{SR,k}$. Let $\phi \in {\mathbb C}^{\tau \times 1}$ and $E_T$ denote the pilot sequence and power of each pilot symbol. $\phi $ is multiplied by ${{\mathbf{{\dot{p}}}}_{S,k}}$ and transmitted by $S_k$. The received pilot matrix at the relay is expressed as
{\footnotetext[4]{Direct estimation of ${\bf H}_{SR,k}$ and ${\bf H}_{RD,k}$ ($k=1,\cdots,K$) requires the total length of pilot sequences no less than $K(N_S+N_D)$. This results in long training phase that degrades the SE greatly as $N_S$ and $N_R$ are large.}}
\begin{equation}\label{4-5}
{{\mathbf{Z}}_{SR,k}} = {{\mathbf{H}}_{SR,k}}\left( {{{\mathbf{{\dot{p}}}}_{S,k}}{\phi ^T} + {{\mathbf{T}}_{SR,k}}} \right) + {{\mathbf{R}}_{SR,k}} + {{\mathbf{N}}_{SR,k}}
\end{equation}
where ${\bf{T}}_{SR,k}\in \mathbb{C}^{N_S \times \tau}$ denotes the transmit distortion noise of $S_k$, whose $l$th column has distribution (based on model (\ref{2-3})) ${\cal CN} ( {0,{\nu _{S,k}}{E_T}{\text{diag}}( {{{\mathbf{{\dot{p}}}}_{S,k}}{\mathbf{{\dot{p}}}}_{S,k}^H} )} )$. ${{\bf{R}}_{SR,k}} \in \mathbb{C}^{N_R \times \tau}$ denotes receive distortion noise at the relay. With model (\ref{2-4}), the $l$th column of ${{\bf{R}}_{SR,k}}$ has distribution $\mathcal{CN}( {0,{\mu _R} {\rm diag}(\mathbb E[ {{{\bf{Z}}_{SR,k}}{{\bf{e}}_l}{{( {{{\bf{Z}}_{SR,k}}{{\bf{e}}_l}} )}^H}} ])} )$. ${{\mathbf{N}}_{SR,k}}$ is the AWGN matrix, whose elements are i.i.d and distributed as ${\cal CN}(0,1)$. By multiplying both side of (\ref{4-5}) with ${\bf {\dot{P}}}_R^H$, we have
\begin{equation}\label{4-6}
{{\mathbf{\tilde{Z}}}_{SR,k}} = {{\mathbf{\underline{h}}}_{SR,k}}{\phi ^T} + {\mathbf{{\dot{P}}}}_R^H{{\mathbf{H}}_{SR,k}}{{\mathbf{T}}_{SR,k}} + {\mathbf{{\dot{P}}}}_R^H{{\mathbf{R}}_{SR,k}} + {\mathbf{{\dot{P}}}}_R^H{{\mathbf{N}}_{SR,k}}
\end{equation}
From (\ref{4-6}), as a merit of HIA scheme, the total length of pilot sequences to obtain the estimates of all effective channels can be as less as $2K$ (when $\tau$ is set to 1).

\begin{theorem}\label{t3}
The LMMES estimator of effective channel ${\bf \underline{h}}_{SR,k}$ can be expressed as
\begin{equation}\label{4-7}
{{{\bf{\hat {\underline{h}}}}}_{SR,k}} = {{\bf{\underline{C}}}_{SR,k}}{{\bf{\Gamma }}_{SR,k}}{{{\bf{\tilde z}}}_{SR,k}}
\end{equation}
where ${{\bf{\underline{C}}}_{SR,k}}$ is the covariance matrix of ${\bf \underline{h}}_{SR,k}$, which can be expressed as  ${{\bf{\underline{C}}}_{SR,k}}={\beta _{SR,k}}{\lambda _1}( {{{{\mathbf{\tilde C}}}_{SR,k}}} ) {\mathbf{{\dot{P}}}}_R^H{{\mathbf{C}}_{SR,k}}{{\mathbf{{\dot{P}}}}_R}$. ${{\bf{\tilde{z}}}_{R,k}}$ is given by ${{\bf{\tilde{z}}}_{SR,k}} = {{\bf{\underline{h}}}_{SR,k}} + \frac{1}{{\tau {E_T}}}{\bf{{\dot{P}}}}_R^H{{\bf{H}}_{SR,k}}{\bf{T}}_{SR,k}{\phi ^*} + \frac{1}{{\tau {E_T}}}{\bf{{\dot{P}}}}_R^H{{\bf{R}}_{SR,k}}{\phi ^*} + \frac{1}{{\tau {E_T}}}{\bf{{\dot{P}}}}_R^H{{\bf{N}}_{SR,k}}{\phi ^*}$ and
\begin{equation}\label{4-8}
\begin{aligned}
{{\mathbf{\Gamma }}_{SR,k}} = \left( {{{\mathbf{\underline{C}}}_{SR,k}} + \frac{{{\nu _{S,k}}}}{{{{\lambda _1}\left( {{{{\mathbf{\tilde C}}}_{SR,k}}} \right)}\tau }}{{\mathbf{\underline{C}}}_{SR,k}} + \frac{1}{{\tau {E_T}}}{{\mathbf{I}}_{{A_R}}}} \right. \\
{\left. { + \frac{{{\mu _R}}}{\tau }\left( {\frac{1}{{{E_T}}} + {\beta _{SR,k}}\left( {{\lambda _1}\left( {{{{\mathbf{\tilde C}}}_{SR,k}}} \right){\text{ + }}{\nu _{S,k}}} \right)} \right){{\mathbf{I}}_{{A_R}}}} \right)^{ - 1}}
\end{aligned}
\end{equation}

The real effective channel ${\bf {\underline{h}}}_{SR}$ can be decomposed as ${\bf {\underline{h}}}_{{SR},k} = {\bf  {\underline{\hat h}}}_{{SR},k} + \Delta {\bf {\underline{h}}}_{{SR},k}$ with $\Delta {\bf {\underline{h}}}_{{SR},k}$ denoting the estimation error. ${\bf \hat  {\underline{h}}}_{SR,k}$ and $\Delta{\bf  {\underline{h}}}_{SR,k}$ are uncorrelated whose covariance matrices can be expressed as ${{{\bf{\underline{\hat{C}}}}_{SR,k}}}={{\bf{\underline{C}}}_{SR,k}}{{\bf{\Gamma }}_{SR,k}}{{\bf{\underline{C}}}_{SR,k}}$ and ${{{\bf{\underline{C}}}_{SR,k}}}-{{{\bf{\underline{\hat{C}}}}_{SR,k}}}$, respectively.
\end{theorem}
\begin{proof}
A sketch of the proof is presented in Appendix-D.
\end{proof}

Different from that with ideal hardware \cite{IEEEhowto:KotechaTSP04}, the LMMSE estimator is not equivalent to MMSE estimator, since the received pilot signal is corrupted by the term ${\mathbf{{\dot{P}}}}_R^H{{\mathbf{H}}_{SR,k}}{{\mathbf{T}}_{SR,k}}$ (due to transmit imperfection of $S_k$), which is not independent with the channel to be estimated. There might exist non-linear estimator that results in smaller MSE. However, the difference should be small since the distortion noises are relatively weak.

The LMMSE estimator of ${\bf \underline{h}}_{RD,k}$ can be obtained using the same approach. For further analysis, we define ${{\bf{\underline{C}}}_{RD,k}}$ and ${\bf \underline{\hat{h}}}_{RD,k}$ as the covariance matrix and LMMSE estimates of ${\bf \underline{h}}_{RD,k}$, respectively. Moreover, we define ${{{\bf{\underline{\hat{C}}}}_{RD,k}}} = {{\bf{\underline{C}}}_{RD,k}}{{\bf{\Gamma }}_{RD,k}}{{\bf{\underline{C}}}_{RD,k}}$ as the covariance matrix of ${\bf \underline{\hat{h}}}_{RD,k}$, where ${{\mathbf{\Gamma }}_{RD,k}}$ is given by
\begin{equation}\label{4-8-1}
\begin{aligned}
{{\mathbf{\Gamma }}_{RD,k}} = \left( {{{\mathbf{\underline{C}}}_{RD,k}} + \frac{{{\nu _{D,k}}}}{{{{\lambda _1}\left( {{{{\mathbf{\tilde C}}}_{RD,k}}} \right)}\tau }}{{\mathbf{\underline{C}}}_{RD,k}} + \frac{1}{{\tau {E_T}}}{{\mathbf{I}}_{{A_T}}}} \right. \\
{\left. { + \frac{{{\mu _R}}}{\tau }\left( {\frac{1}{{{E_T}}} + {\beta _{RD,k}}\left( {{\lambda _1}\left( {{{{\mathbf{\tilde C}}}_{RD,k}}} \right){\text{ + }}{\nu _{D,k}}} \right)} \right){{\mathbf{I}}_{{A_T}}}} \right)^{ - 1}}
\end{aligned}
\end{equation}

\subsection{Achievable Rate Analysis}

This subsection analyzes the achievable rate of HIA-MM-FDR. The achievable rate expressions are derived based on the bounding technique in \cite{IEEEhowto:Hassibi}.

\emph{1) Achievable Rate of $S_k \!\to\! R$ Channel:} By treating $\sqrt{E_{S,k}}{\mathbb E[ {{\bf{w}}_{R,k}^H{{\bf{H}}_{SR,k}}{{\bf{{\dot{p}}}}_{S,k}}} ]}s_k[u]$ as the desired signal at $S_k \! \to\! R$ channel, and approximating the remaining terms in (\ref{2-6}), i.e., $ {{\mathord{\buildrel{\lower3pt\hbox{$\scriptscriptstyle\smile$}} \over y} }_{R,k}}\left[ u \right] \!-\!\sqrt{E_{S,k}}{\mathbb E[ {{\bf{w}}_{R,k}^H{{\bf{H}}_{SR,k}}{{\bf{{\dot{p}}}}_{S,k}}} ]}s_k[u]$, using the worst-case uncorrelated additive Gaussian noise with the same variance, the rate in (\ref{4-11}) is achievable on $S_k\!\! \to\!\! R$ channel, where $\textsf{EI}_k$, ${{\textsf{MUI}}_{R,k}}$, $\textsf{D}_{R,k}^\textsf{T}$ and $\textsf{D}_{R,k}^\textsf{R}$ denote the EI, MUI, effective distortion due to transmit imperfection of sources and effective distortion due to receive imperfection of relay respectively, which are given by
\begin{figure*}
\normalsize
\setcounter{mytempeqncnt}{\value{equation}}
\setcounter{equation}{28}
\begin{equation}\label{4-11}
      \textsf{R}_{SR,k}^{{\rm{HIA}}} = \frac{{T - 2K\tau }}{T} {\log _2}\left( 1 + {\frac{{{E_{S,k}}{{\left| {\mathbb E\left[ {{\bf{w}}_{R,k}^H{{\bf{H}}_{SR,k}}{{\bf{{\dot{p}}}}_{S,k}}} \right]} \right|}^2}}}{{{E_{S,k}}{\rm{var}}\left( {{\bf{w}}_{R,k}^H{{\bf{H}}_{SR,k}}{{\bf{{\dot{p}}}}_{S,k}}} \right){+} \mathbb E\left[ {{{\textsf{EI}}}{{{}}_k}} \right] + \mathbb E\left[ {{\rm{\textsf{MUI}}}{{{}}_{R,k}}} \right] + \mathbb E\left[ {{{\textsf{D}}}_{R,k}^{{\textsf{T}}}} \right] +\mathbb E\left[ {{{\textsf{D}}}_{R,k}^{{\textsf{R}}}} \right] + \mathbb E\left[ {{{\left\| {{\bf{w}}_{R,k}} \right\|}^2}} \right]}}} \right)
\end{equation}
\setcounter{equation}{\value{mytempeqncnt}}
\hrule
\end{figure*}
\setcounter{equation}{29}
\begin{equation}\label{4-12}
\begin{array}{ll}
\textsf{EI}_k = {\bf{w}}_{R,k}^H{{\bf{H}}_{EI}}\left( {{{\mathbf{W}}_T}{\Lambda _R}{\mathbf{W}}_T^H + \mathbb E\left[ {{{\mathbf{t}}_R}\left[ u \right]{\mathbf{t}}_R^H\left[ u \right]} \right]} \right){\bf{H}}_{EI}^H{{\bf{w}}_{R,k}} \\
{{\textsf{MUI}}_{R,k}} = \sum\limits_{j = 1,j \ne k}^K {{E_{S,j}}{{\left| {{\bf{w}}_{R,k}^H{{\bf{H}}_{SR,j}}{{\bf{{\dot{p}}}}_{S,j}}} \right|}^2}} \\
\textsf{D}_{R,k}^\textsf{T} = \sum\limits_{j = 1}^K {{\bf{w}}_{R,k}^H{{\bf{H}}_{SR,j}} \mathbb  E\left[ {{{\bf{t}}_{S,j}}[u]{\bf{t}}_{S,j}^H[u]} \right]{\bf{H}}_{SR,j}^H{{\bf{w}}_{R,k}}} \\
\textsf{D}_{R,k}^\textsf{R} = {\bf{w}}_{R,k}^H \mathbb E\left[ {{{\mathbf{r}}_R}\left[ u \right]{\mathbf{r}}_R^H\left[ u \right]} \right]{{\bf{w}}_{R,k}}
\end{array}
\end{equation}
Wherein, the BF matrices at the relay are given by
\begin{equation}\label{4-12-1}
\begin{aligned}
& {{\mathbf{{W}}}_R} = {\bf {\dot{P}}}_R {{\mathbf{\underline{\hat{H}}}}_{SR}}{\left( {{\mathbf{\underline{\hat{H}}}}_{SR}^H{{\mathbf{\underline{\hat{H}}}}_{SR}}} \right)^{ - 1}} \\
& {{\mathbf{{W}}}_T} = {\bf {\dot{P}}}_T {{\mathbf{\underline{\hat{H}}}}_{RD}}{\left( {{\mathbf{\underline{\hat{H}}}}_{RD}^H{{\mathbf{\underline{\hat{H}}}}_{RD}}} \right)^{ - 1}}{\mathbf{\hat{\Upsilon } }}^{-1/2}
\end{aligned}
\end{equation}
where ${{\mathbf{\underline{\hat{H}}}}_{i}} = \left[ {{{\mathbf{\underline{\hat{h}}}}_{i,1}}, \cdots ,{{\mathbf{\underline{\hat{h}}}}_{i,K}}} \right]$ ($i\in\{SR,RD\}$) and ${\mathbf{\hat\Upsilon  }}$ is a diagonal normalized matrix with $[{\mathbf{\hat\Upsilon  }}]_{l,l}={{\mathbf{e}}_l^H{{\left( {{\mathbf{\underline{\hat{H}}}}_{RD}^H{{\mathbf{\underline{\hat{H}}}}_{RD}}} \right)}^{ - 1}}{{\mathbf{e}}_l}}$. Note that we have replaced the effective channels in inner BF matrices with their estimates obtained in Theorem \ref{t3}.

\begin{theorem}\label{t4}
Assume that the assumptions \textbf{A1} and \textbf{A2} hold, and $A_R$ is selected so that ${\text{Tr}}\left( {{{\mathbf{\underline{C}}}_{SR,k}}} \right) = {\cal O}(N_R)$. {\footnotemark[5]} With HIA scheme, the achievable rate of $S_k \to R$ channel in the large-$(N_R,N_T)$ regime is given by $(\ref{4-11})$, where the powers of desired signal and AWGN are ${{{E_{S,k}}{{\left| {\mathbb E\left[ {{\bf{w}}_{R,k}^H{{\bf{H}}_{SR,k}}{{\bf{{\dot{p}}}}_{S,k}}} \right]} \right|}^2}}} = E_{S,k}$ and ${\mathbb E\left[ {{{\left\| {{{\bf{w}}_{R,k}}} \right\|}^2}} \right]}={( {{\rm{Tr}}( {{{{\bf{\hat {\underline{C}}}}}_{SR,k}}} )} )^{ - 1}} $, and
{\footnotetext[5]{By increasing $A_R$, this can always be achieved since ${\text{Tr}}\left( {{{\mathbf{\underline{C}}}_{SR,k}}} \right)$ approaches $N_R$ as $A_R$ increases from assumption \textbf{A2}.}}
\begin{equation}\label{4-13}
\begin{aligned}
& {E_{S,k}}{\rm{var}}\left( {{\bf{w}}_{R,k}^H{{\bf{H}}_{SR,k}}{{\bf{{\dot{p}}}}_{S,k}}} \right) = {E_{S,k}}{\left( {{\rm{Tr}}\left( {{{{\bf{\hat {\underline{C}}}}}_{SR,k}}} \right)} \right)^{ - 2}}\delta _{k}^{SR} \\
&  \mathbb E \left[\textsf{EI}_k\right] = \frac{{\beta _{EI}}{\rm{Tr}}\left( {{{{\bf{\tilde C}}}_{EI}}{{\bf \Omega} _R}} \right){\rm{Tr}}\left( {{{{\bf{\hat {\underline{C}}}}}_{SR,k}}{\bf{{\dot{P}}}}_R^H{{\bf{C}}_{EI}}{{\bf{{\dot{P}}}}_R}} \right)}{\left( {{\rm{Tr}}\left( {{{{\bf{\hat {\underline{C}}}}}_{SR,k}}} \right)} \right)^2} \\
& \mathbb E[{{\textsf{MUI}}_{R,k}}] = \sum_{j = 1,j \ne k}^K \frac{{E_{S,j}}{\rm{Tr}}\left( {\left( {{{\bf{\underline{C}}}_{SR,j}} - {{{\bf{\hat {\underline{C}}}}}_{SR,j}}} \right){{{\bf{\hat {\underline{C}}}}}_{SR,k}}} \right)}{\left( {{\rm{Tr}}\left( {{{{\bf{\hat {\underline{C}}}}}_{SR,k}}} \right)} \right)^{ 2}}\\
& \mathbb E[\textsf{D}_{R,k}^\textsf{T}] = \frac{{\nu _{S,k}}{E_{S,k}} g_{kk} + \sum\limits_{j = 1,j \ne k}^K {{\nu _{S,j}}{E_{S,j}} g_{kj}{{}}} }{\left( {{\rm{Tr}}\left( {{{{\bf{\hat {\underline{C}}}}}_{SR,k}}} \right)} \right)^{  2}}\\
& \mathbb E \left[\textsf{D}_{R,k}^\textsf{R} \right] = \frac{\mu _R}{ {{\rm{Tr}}\left( {{{{\bf{\hat {\underline{C}}}}}_{SR,k}}} \right)} } \\
& \times {{\sum\limits_{j = 1}^K {{E_{S,j}}{\beta _{SR,j}}{\rm{Tr}}\left( {{{{\bf{\tilde C}}}_{SR,j}}{{\bf{\Omega }}_{S,j}}} \right)}  + {\beta _{EI}}{\rm{Tr}}\left( {{{{\bf{\tilde C}}}_{EI}}{{\bf{\Omega }}_R}} \right) + 1}}
\end{aligned}
\end{equation}
where $\delta _{k}^{SR}$ is defined as $\delta _{k}^{SR} = \mathbb E[ {{{| {\Delta {\bf{\underline{h}}}_{SR,k}^H{{{\bf{\hat {\underline{h}}}}}_{SR,k}}} |}^2}} ]$ and is given by Lemma \ref{lm1} in Appendix-A. $g_{kk}$ and $g_{kj}$ are given by (\ref{D10}) and (\ref{D12}) in the Appendix-E. ${{\bf \Omega}_{S,k}}$ and ${\bf\Omega}_{R}$ are expressed as
\begin{equation}\label{4-14}
\begin{aligned}
 & {{\bf{\Omega }}_{S,k}} = {{\bf{{\dot{p}}}}_{S,k}}{\bf{{\dot{p}}}}_{S,k}^H + {\nu _{S,k}}{\rm{diag}}\left( {{{\bf{{\dot{p}}}}_{S,k}}{\bf{{\dot{p}}}}_{S,k}^H} \right)\\
 & {{\bf{\Omega }}_R} = \sum\limits_{l = 1}^K {{E_{R,l}} \frac{{{\bf{{\dot{P}}}}_T}{{{\bf{\hat {\underline{C}}}}}_{RD,l}}{\bf{{\dot{P}}}}_T^H + {\nu _R}{\rm{diag}}\left( {{{\bf{{\dot{P}}}}_T}{{{\bf{\hat {\underline{C}}}}}_{RD,l}}{\bf{{\dot{P}}}}_T^H} \right)}{{{{\rm{Tr}}\left( {{{{\bf{\hat {\underline{C}}}}}_{RD,l}}} \right)}}}}
\end{aligned}
\end{equation}
The scaling behaviors of above factors are shown in Table \ref{tab1}.
\end{theorem}
\begin{proof}
See Appendix-E.
\end{proof}

\begin{table*}[!t]
\caption{Scaling behaviors of different factors in Theorem \ref{t4}.}
\label{tab1}
\centering
\includegraphics[width=12cm]{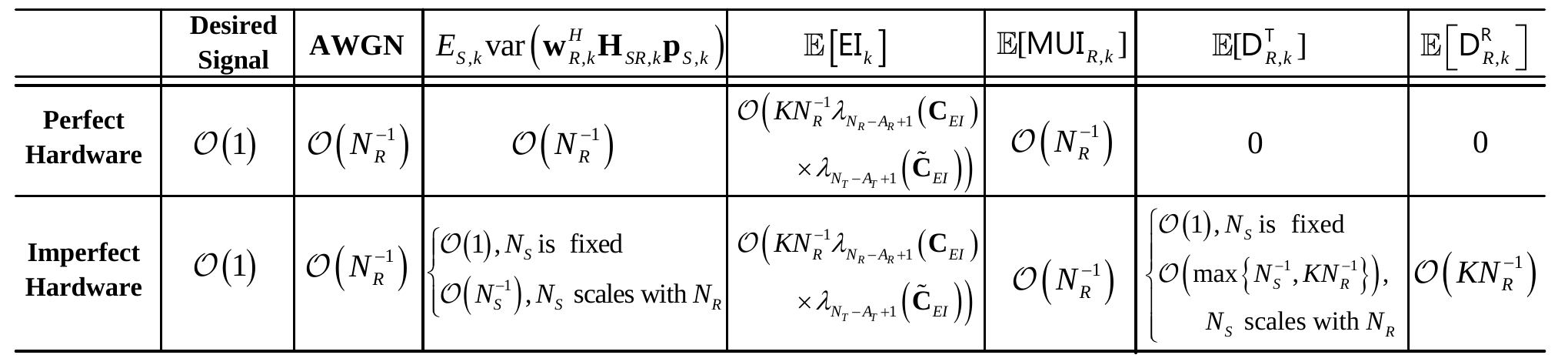}
\end{table*}
\begin{table*}[!t]
\caption{Scaling behaviors of different factors in Theorem \ref{t5}.}
\label{tab2}
\centering
\begin{tabular}{lll}
\includegraphics[width=11.5cm]{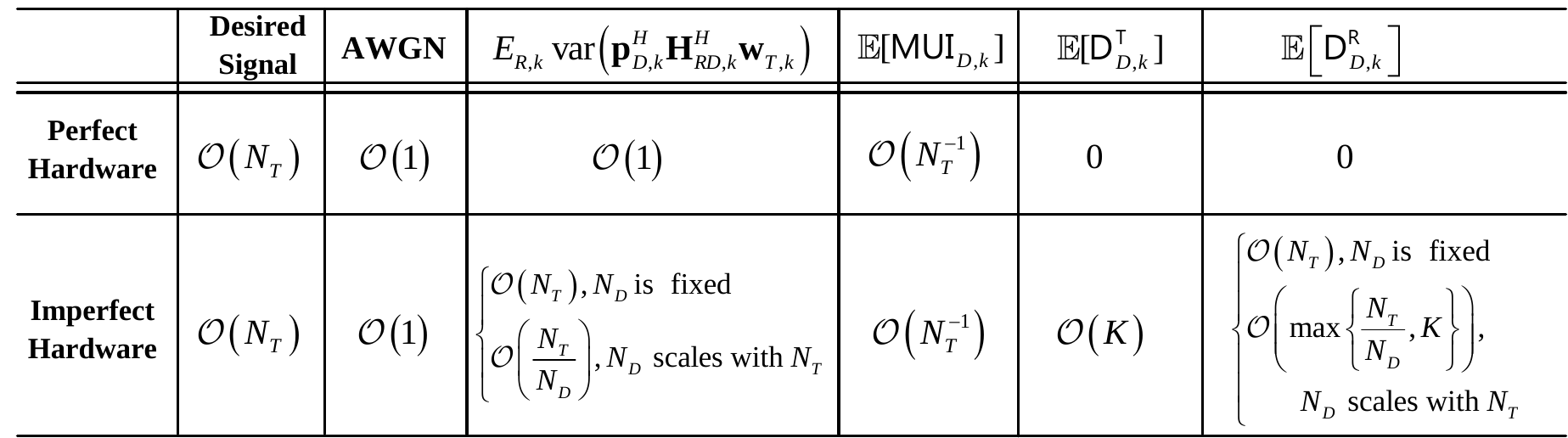}
\end{tabular}
\end{table*}

Several important results on the effect of EI, MUI and hardware impairments on achievable rate of $S_k \to R$ channel can be obtained from Theorem \ref{t4} and Table \ref{tab1}.

\emph{Effect of EI:} With the proposed HIA scheme, the scaling behavior of EI is high correlated with ${\lambda _{{N_R} - {A_R} + 1}}( {{{\mathbf{C}}_{EI}}} ) $ and ${\lambda _{{N_T} - {A_T} + 1}}( {{{{\mathbf{\tilde C}}}_{EI}}} )$. With the assumption \textbf{A1}, we have ${\lambda _{{N_R} - {A_R} + 1}}( {{{\mathbf{C}}_{EI}}} ) \le {\cal O}(1)$ and ${\lambda _{{N_T} - {A_T} + 1}}( {{{{\mathbf{\tilde C}}}_{EI}}} ) \le {\cal O}(1)$. In fact, under specific spatial correlation model (e.g., the physical channel model with a fixed number of angular bins \cite{IEEEhowto:HoydisJSAC13}), it is possible to make ${\lambda _{{N_R} - {A_R} + 1}}\left( {{{\mathbf{C}}_{EI}}} \right) <{\cal O}(1)$ and ${\lambda _{{N_T} - {A_T} + 1}}( {{{{\mathbf{\tilde C}}}_{EI}}} )<{\cal O}(1)$ by selecting proper $A_R$ and $A_T$. In this case, the power of EI can decrease faster than ${\cal O}(KN_R^{-1})$ in the large-$(N_R,N_T)$ regime. Note that this is in contrast with the result in \cite{IEEEhowto:NgoJSAC14}, which shown that the power of EI decreases exactly with ${\cal O}(KN_R^{-1})$.

\emph{Effect of MUI:} The scaling behavior of MUI confirms well with the classic result in MM-FDR with perfect hardware \cite{IEEEhowto:NgoJSAC14}, i.e., the MUI diminishes as $N_R\to \infty$ if $K \ll N_R$.

\emph{Effect of Transmit Imperfection of Sources:} (i) When $N_S$ is fixed, an interesting observation is that the term ${E_{S,k}}{\rm{var}} ( {{\bf{w}}_{R,k}^H{{\bf{H}}_{SR,k}}{{\bf{{\dot{p}}}}_{S,k}}}  )$, which is widely convinced $\mathcal{O} (N_R^{-1} )$ terms with perfect hardware, scales as $\mathcal{O} (1 )$. As a result, the achievable rate of $S_k \to R$ channel is limited by the joint effect of ${E_{S,k}}{\rm{var}} ( {{\bf{w}}_{R,k}^H{{\bf{H}}_{SR,k}}{{\bf{{\dot{p}}}}_{S,k}}}  )$ and the effective distortion due to the transmit imperfection of sources, i.e., $\textsf{D}_{R,k}^\textsf{T}$. In fact, the variation of scaling behavior of ${E_{S,k}}{\rm{var}} ( {{\bf{w}}_{R,k}^H{{\bf{H}}_{SR,k}}{{\bf{{\dot{p}}}}_{S,k}}}  )$ can also be viewed as a ``negative effect'' of transmit imperfection of sources. To see this, inserting the expression of $\delta _{k}^{SR}$ (given by Lemma \ref{lm1} in Appendix-A) into (\ref{4-13}) and letting $N_R\to \infty$, we can obtain $ {E_{S,k}}{\rm{var}} ( {{\bf{w}}_{R,k}^H{{\bf{H}}_{SR,k}}{{\bf{{\dot{p}}}}_{S,k}}}  ) = \frac{{{\nu _{S,k}}}}{\tau }{E_{R,k}}{\text{Tr}} ( {{{{\mathbf{\hat {\underline{C}}}}}_{SR,k}}}  ){\mathbf{{\dot{p}}}}_{S,k}^H{\text{diag}} ( {{\mathbf{{\dot{p}}}}_{S,k}^H{{\mathbf{{\dot{p}}}}_{S,k}}} ){{\mathbf{{\dot{p}}}}_{S,k}} + {\cal O} ( {N_R^{ - 1}}  )$, which demonstrates that ${E_{S,k}}{\rm{var}} ( {{\bf{w}}_{R,k}^H{{\bf{H}}_{SR,k}}{{\bf{{\dot{p}}}}_{S,k}}}  ) $ scales as $ {\cal O} ( 1  )$ as long as ${\nu _{S,k}} > 0$. (ii) When $N_S$ scales with $N_R$ and $K\ll N_R$, it is seen that the proposed HIA scheme reduces the effect of transmit imperfection at sources by a factor of $\frac{1}{N_S}$. The reason is that the transmit distortion is pre-suppressed by the transmit BF vector ${\bf {\dot{p}}}_{S,k}$.

\emph{Effect of Receive Imperfection of the Relay:} By comparing the result in Table \ref{tab1} and $(\ref{3-7})$, it is seen that the HIA scheme does not change the scaling behavior of distortion noise caused by receive imperfection of the relay. However, the receive distortion noise power is still reduced, since EI (which is a main contributor of receive distortion noise) is suppressed.

\emph{2) Achievable Rate of $R \to D_k$ Channel:} According to (\ref{2-7}) and the bounding technique in \cite{IEEEhowto:Hassibi}, the achievable rate of $R\to D_k$ channel can be expressed as
\begin{equation}\label{4-15}
\begin{aligned}
& \textsf{R}_{RD,k}^{\rm HIA}= \frac{{T - 2K\tau }}{T} \\
& \times{\log _2} \begin{array}{ll}\left( {1 + \frac{{{E_{R,k}}{{\left| {\mathbb E\left[ {{\bf{{\dot{p}}}}_{D,k}^H{\bf{H}}_{RD,k}^H{{\bf{w}}_{T,k}}} \right]} \right|}^2}}}{{{E_{R,k}}{\mathop{\rm var}} \left( {{\bf{{\dot{p}}}}_{D,k}^H{\bf{H}}_{RD,k}^H{{\bf{w}}_{T,k}}} \right) + \mathbb E\left[ {{\textsf{MUI}}_{D,k}} \right] + \mathbb E\left[ {\textsf{D}_{D,k}^\textsf{T}} \right] + \textsf{D}_{D,k}^\textsf{R}}+1}} \right)\end{array}
 \end{aligned}
\end{equation}
where ${{\textsf{MUI}}_{D,k}}$, $\textsf{D}_{D,k}^\textsf{T}$ and $\textsf{D}_{D,k}^\textsf{R}$ denote the MUI, effective distortion due to transmit imperfection of relay and effective distortion due to receive imperfection of destinations, respectively, which are given by
\begin{equation}\label{4-16}
\begin{array}{l}
{{\textsf{MUI}}_{D,k}} = \sum\limits_{j = 1,j \ne k}^K {{E_{R,j}}{{\left| {{\bf{{\dot{p}}}}_{D,k}^H{\bf{H}}_{RD,k}^H{{\bf{w}}_{T,j}}} \right|}^2}} \\
\textsf{D}_{D,k}^\textsf{T} = {\bf{{\dot{p}}}}_{D,k}^H{\bf{H}}_{RD,k}^H{\bf{\Theta }}_{R}^\textsf{T}{{\bf{H}}_{RD,k}}{{\bf{{\dot{p}}}}_{D,k}}\\
\textsf{D}_{D,k}^\textsf{R} = {\bf{{\dot{p}}}}_{D,k}^H{\bf{\Theta }}_{D,k}^\textsf{R}{{\bf{{\dot{p}}}}_{D,k}}
\end{array}
\end{equation}

\begin{theorem}\label{t5}
Assume that assumptions \textbf{A1} and \textbf{A2} hold, and $A_T$ is selected so that ${\text{Tr}}\left( {{{\mathbf{\underline{C}}}_{RD,k}}} \right) = {\cal O}(N_T)$. With HIA scheme, the achievable rate of $R \to D_k$ channel in the large-$(N_R,N_T)$ regime is given by $(\ref{4-15})$, where the power of desired signal is ${E_{R,k}}{\left| {\mathbb E[ {{\bf{{\dot{p}}}}_{D,k}^H{\bf{H}}_{RD,k}^H{{\bf{w}}_{T,k}}} ]} \right|^2} = {E_{R,k}}{ {{\rm{Tr}}( {{\bf{\hat {\underline{C}}}}_{RD,k}^{}} )} }$, and
\begin{equation}\label{4-17}
\begin{aligned}
& {E_{R,k}}\operatorname{var} \left( {{\mathbf{{\dot{p}}}}_{D,k}^H{\mathbf{H}}_{RD,k}^H{{\mathbf{w}}_{T,k}}} \right) = {\left( {{\text{Tr}}\left( {{\mathbf{\hat {\underline{C}}}}_{RD,k}^{}} \right)} \right)^{ - 1}}{E_{R,k}}\delta _{k}^{RD} \\
& \mathbb E\left[{{\textsf{MUI}}_{D,k}}\right] = \sum\limits_{j = 1,j \ne k}^K \frac{{E_{R,j}}{\rm{Tr}}\left( {\left( {{{\bf{\underline{C}}}_{RD,k}} - {{{\bf{\hat {\underline{C}}}}}_{RD,k}}} \right){{{\bf{\hat {\underline{C}}}}}_{RD,j}}} \right)}{{ {{\rm{Tr}}\left( {{\bf{\hat {\underline{C}}}}_{RD,j}^{}} \right)} }} \\
& \mathbb E\left[ \textsf{D}_{D,k}^\textsf{T} \right] = {\nu _R}\sum\limits_{l = 1}^K \frac{{E_{R,l}}{\rm{Tr}}\left( {{{\bf{C}}_{RD,k}}{\rm{diag}}\left( {{{\bf{{\dot{P}}}}_T}{{{\bf{\hat {\underline{C}}}}}_{RD,l}}{\bf{{\dot{P}}}}_T^H} \right)} \right)}{\left({\lambda _1}\left( {{{{\bf{\tilde C}}}_{RD,k}}} \right)\right)^{-1}{ {{\rm{Tr}}\left( {{{{\bf{\hat {\underline{C}}}}}_{RD,l}}} \right)} }} \\
& \textsf{D}_{D,k}^\textsf{R}  = \frac{{{\mu _{D,k}}{E_{R,k}}}}{{{\text{Tr}}\left( {{{{\mathbf{\hat {\underline{C}}}}}_{RD,k}}} \right)}}\left( {\frac{{{\mathbf{{\dot{p}}}}_{D,k}^H{\text{diag}}\left( {{{\mathbf{{\dot{p}}}}_{D,k}}{\mathbf{{\dot{p}}}}_{D,k}^H} \right){{\mathbf{{\dot{p}}}}_{D,k}}}}{{{{\left( {{\text{Tr}}\left( {{{{\mathbf{\hat {\underline{C}}}}}_{RD,k}}} \right)} \right)}^{ - 2}}}}} \right. \\
&  +  \frac{  \frac{{{\mu _R}}}{\tau }\left( {\frac{1}{{{E_T}}} { + \frac{1}{{{\mu _R}{E_T}}}}+ {\beta _{RD,k}}\left( {{\lambda _1}\left( {{{{\mathbf{\tilde C}}}_{RD,k}}} \right) + {\nu _{D,k}}} \right)} \right)    }{\left({{\text{Tr}}\left( {{\mathbf{\underline{C}}}_{RD,k}^3{\mathbf{\Gamma }}_{RD,k}^2} \right)}\right)^{-1}{{\lambda _1}\left( {{{{\mathbf{\tilde C}}}_{RD,k}}} \right)}} \\
& \left. { + \sum\limits_{j = 1,j \ne k}^K {\frac{{{E_{R,j}}{\beta _{RD,j}}}}{{{E_{R,k}}}}{\text{Tr}}\left( {{{\mathbf{C}}_{RD,k}}{{\mathbf{{\dot{P}}}}_T}{{{\mathbf{\hat {\underline{C}}}}}_{RD,j}}{\mathbf{{\dot{P}}}}_T^H} \right)}  + \frac{{{\text{Tr}}\left( {{{{\mathbf{\hat {\underline{C}}}}}_{RD,k}}} \right)}}{{{E_{R,k}}}}} \right)
\end{aligned}
\end{equation}
$\delta _{k}^{RD}$ is defined as $\delta _{k}^{RD} = \mathbb E[ {{{| {\Delta {\bf{\underline{h}}}_{RD,k}^H{{{\bf{\hat {\underline{h}}}}}_{RD,k}}} |}^2}} ]$ and is given by Lemma \ref{lm1} in Appendix-A. The scaling behaviors of above factors are shown in Table \ref{tab2}.
\end{theorem}
\begin{proof}
The proof is similar with that for Theorem \ref{t4}, and thus it is neglected.
\end{proof}

From Theorem \ref{t5} and Table \ref{tab2}, one can observe a strong similarity between the effects of hardware impairments on $R\to D_k$ channel and that on $S_k \to R$ channel. In particular, the effective distortion due to the transmit imperfection of relay $\textsf{D}_{D,k}^\textsf{T}$ is suppressed by the large transmit array of the relay, and it diminishes as $N_T\to \infty$ and $K \ll N_T$. Moreover, when $N_D$ is fixed, the achievable rate of $R\to D_k$ channel is limited by the term ${E_{R,k}}\operatorname{var} \left( {{\mathbf{{\dot{p}}}}_{D,k}^H{\mathbf{H}}_{RD,k}^H{{\mathbf{w}}_{T,k}}} \right)$ and distortion noise caused by receive imperfection of $D_k$. However, as $N_D$ scales with $N_T$ and $K \ll N_T$, these factors are suppressed by $\frac{1}{N_D}$. This is because that applying ${\bf {\dot{p}}}_{D,k}$ can be viewed as a post-suppression on receive distortion noise.

With HIA scheme, the following end-to-end rate for $S_k \to R \to D_k$ channel is achievable in the large-$(N_R,N_T)$ regime
\begin{equation}\label{4-18}
    \textsf{R}_k^{\rm HIA} = \min\left(\textsf{R}_{SR,k}^{\rm HIA},\textsf{R}_{RD,k}^{\rm HIA}\right)
\end{equation}
Different from that for upper bound in the last section, $\textsf{R}_k^{\rm HIA}$ cannot be expressed in a simple form like (\ref{3-7}). However, based on Table \ref{tab1} and Table \ref{tab2}, one can simply deduce that $\textsf{R}_k^{\rm Lower}$ will converge to a finite ceiling if $N_S$ and $N_D$ are fixed. However, with HIA scheme, the end-to-end achievable rate grows without bound as $\min(N_R,N_T) \to \infty$ if $N_S$ and $N_D$ scale with $\min(N_R,N_T)$ and $K\ll \min(N_R,N_T)$. This is the same to the situation with ideal hardware \cite{IEEEhowto:NgoJSAC14}.

\subsection{Discussion on Hardware Design}

Based on the achievable rate expressions, we discuss the hardware design of sources, destinations and relay in this subsection. Although the achievable rate expressions hold for arbitrary $K\le \min\{N_R,N_T\}$, we will restrict our analysis to $K\ll \min\{N_R,N_T\}$. This is a fundamental condition to ensure the benefit of massive MIMO, i.e., the MUI is suppressed by a large surplus of degrees of freedom \cite{IEEEhowto:LarssonCM14}. In MM-FDR with hardware impairments, it is also essential to control the number of S-D pairs (by, e.g., some user scheduling scheme) to suppress the EI and distortions.

\subsubsection{Hardware Design of Sources and Destinations}

As $N_S$ and $N_D$ are fixed and small (due to e.g., size limitation), from Theorem \ref{t4} and \ref{t5}, it is of great importance to reduce $\nu_{S,k}$, $\nu_{D,k}$ and $\mu_{D,k}$ in order to alleviate the ``ceiling effect'' on achievable rate. This indicates that it is beneficial to use high-quality hardware at sources and destinations in this case. Note that, different from the relay side, the use of high-quality hardware at sources and destinations is affordable, since the increased cost scales with $N_S$ and $N_D$, instead of $\min\{N_R,N_T\}$.

When $N_S$ and $N_D$ scale with $\min\{N_R,N_T\}$, from Theorem \ref{t4} and \ref{t5}, the distortion noise powers due to hardware impairments at sources and destinations are lower order terms when compared with desired signal powers. This means that the sources and destinations can decrease the hardware quality to some extent without degrading the performance. Later an example will be given to show how to achieve this.

\subsubsection{Hardware Design of Relay}

Similar to that for sources and destinations when $N_S$ and $N_D$ scale with $\min\{N_R,N_T\}$, the distortion noises caused by hardware impairments at the relay are lower order terms when compared with desired signal powers, which makes it possible to decrease the hardware quality without hurting the performance greatly.

\emph{Example:} To get a clear insight, we consider the achievable rate of $S\to R$ {channel\footnotemark[6]} given by Theorem \ref{t4}. We assume that there is only one S-D pair and let $\nu_{R}=\mu_{R}$. When $N_S$ scales with $N_R$, the achievable rate of $S_k\to R$ channel reduces to
{\footnotetext[6]{We consider the achievable rate of $S_k \!\to \! R$ channel to simplify the analysis. Similar result can be obtained based on the achievable rate of $R\!\to \!D$ channel in Theorem \ref{t5}. Therefore, the analysis is also valid for end-to-end achievable rate.}}
\begin{equation}\label{4-19}
\begin{aligned}
 & \textsf{R}_{SR,k}^{\rm HIA} = \frac{{T - 2K\tau }}{T}\\
  &\times {\log _2} \left( {1 + \frac{1}{{{c_1}{\nu _{S,1}}N_R^{ - 1} + {c_2}{\mu _R}N_R^{ - 1} + {\cal O}\left( {N_R^{ - 1}} \right)}}} \right)
  \end{aligned}
\end{equation}
where $c_1$ and $c_2$ are positive constants independent of $\nu_{S,1}$, $\mu_{R}$ and $\{N_S,N_R\}$. In (\ref{4-19}), we have replaced $N_S$ with the product of $N_R$ and some constant. Assume we wish to achieve rate ${R}_T$. According to (\ref{4-19}), we have (as $N_R \to \infty$) ${c_1}{\nu _{S,1}} + {c_2}{\mu _R} = ( {{{( {{2^{\frac{T}{{T - 2K\tau }}{R_T}}} - 1} )}^{ - 1}} - {\cal O}( {N_R^{ - 1}} )} ){N_R} $ $\approx {( {{2^{\frac{T}{{T - 2K\tau }}{R_T}}} - 1} )^{ - 1}}{N_R}$. The result is encouraging since it implies that we can increase ${\nu _{S,1}}$ and ${\mu _R}$ linearly as $N_R$ increase (with $N_S$ scaling with $N_R$) without degrading the achievable rate. As a result, inexpensive MM-FDR is possible.

\section{Joint DOF and Power Optimization}

In this section, we propose a low-complexity JDPO algorithm to maximize the SE (defined as the sum of all destinations' achievable rates) of HIA-MM-FDR, subject to the maximum power constrains. The achievable rate expressions obtained in section V-C are utilized in the proposed algorithm. The algorithm needs only statistical knowledge of channels. Therefore, it can be computed offline at a central node (e.g., relay) and then broadcasts to the other nodes.

Let $E_{S,k}^{\max }$ and $E_R^{\max}$ be the maximum transmit power constraints at $S_k$ and relay, respectively, the SE optimization problem can be formulated as follows
\begin{equation}\label{5-2}
\begin{aligned}
\mathop {\max }\limits_{{E_{S,1}}, \cdots ,{E_{S,K}};{E_{R,1}}, \cdots ,{E_{R,K}};A_R,A_T} \textsf{SE} = \sum\limits_{k = 1}^K {\min \left\{ {{\textsf{R}_{SR,k}^{\rm HIA}},{\textsf{R}_{RD,k}^{\rm HIA}}} \right\}} \\
  s.t.\left\{ {\begin{array}{*{20}{l}}
  {0 \le {E_{S,k}} \le E_{S,k}^{\max },k = 1, \cdots ,K} \\
  {\sum\limits_{k = 1}^K {{E_{R,k}}}  \le E_R^{\max }} \\
  A_R \in \mathcal{A}_R, A_T \in {\cal A}_T
\end{array}} \right.
\end{aligned}
\end{equation}
where $\mathcal{A}_R = \{K,K+1,\cdots,N_R\}$ and  $\mathcal{A}_T = \{K,K+1,\cdots,N_T\}$.

According to Theorem 4 and Theorem 5, we rewrite ${\textsf{R}_{SR,k}^{\rm HIA}}$ and ${\textsf{R}_{RD,k}^{\rm HIA}}$ as ${\textsf{R}_{SR,k}^{\rm HIA}} = \frac{{T - 2K \tau }}{T}{\log _2}\left( {1 + {\gamma _{SR,k}}} \right)$ and ${\textsf{R}_{RD,k}^{\rm HIA}} = \frac{{T - 2K \tau }}{T}{\log _2}\left( {1 + {\gamma _{RD,k}}} \right)$, where ${{\gamma _{SR,k}}}$ and ${{\gamma _{RD,k}}}$ denote the effective received SINRs of $S_k\to R$ and $R \to D_k$ channels respectively, which can be expressed as
\begin{equation}\label{5-3}
    \begin{aligned}
        & {\gamma _{SR,k}} = \frac{{{E_{S,k}}}}{{\sum_{j = 1}^K {{a^R_{k,j}}{E_{S,j}}}  + \sum_{j = 1}^K {{b^R_{k,j}}{E_{R,j}}}  + \frac{ \left( {{\mu _R} + 1} \right)}{{{\text{Tr}}\left( {{{{\mathbf{\hat {\underline{C}}}}}_{SR,k}}} \right)}} }}\\
        & {\gamma _{RD,k}} = \frac{{{\text{Tr}}\left( {{\mathbf{\hat {{\underline{C}}}}}_{RD,k}^{}} \right){E_{R,k}}}}{{\sum_{j = 1}^K {b_{k,j}^D{E_{R,j}}}  + {\mu _{D,k}} + 1}}
    \end{aligned}
\end{equation}
where
\begin{equation*}
\begin{aligned}
   & a_{kj}^R = \frac{\delta _{kj}^{SR} + {\nu _{S,k}}{g_{kj}}}{\left( {{\text{Tr}}\left( {{{{\mathbf{\hat {\underline{C}}}}}_{SR,k}}} \right)} \right)^{ 2}} + {\mu _R}{\beta _{SR,j}}\frac{{{\text{Tr}}\left( {{{{\mathbf{\tilde C}}}_{SR,j}}{{\mathbf{\Omega }}_{S,j}}} \right)}}{{{\text{Tr}}\left( {{{{\mathbf{\hat {\underline{C}}}}}_{SR,k}}} \right)}}\\
   &  b_{kj}^R = \frac{{{\beta _{EI}}{\text{Tr}}\left( {{{{\mathbf{\tilde C}}}_{EI}}\left( {{{\mathbf{{\dot{P}}}}_T}{{{\mathbf{\hat {\underline{C}}}}}_{RD,j}}{\mathbf{{\dot{P}}}}_T^H + {\nu _R}{\text{diag}}\left( {{{\mathbf{{\dot{P}}}}_T}{{{\mathbf{\hat {\underline{C}}}}}_{RD,j}}{\mathbf{{\dot{P}}}}_T^H} \right)} \right)} \right)}}{{{\text{Tr}}\left( {{{{\mathbf{\hat {\underline{C}}}}}_{SR,k}}} \right){\text{Tr}}\left( {{{{\mathbf{\hat {\underline{C}}}}}_{RD,j}}} \right)}} \\
   & \times \left( {\frac{{{\text{Tr}}\left( {{{{\mathbf{\hat {\underline{C}}}}}_{SR,k}}{\mathbf{{\dot{P}}}}_R^H{{\mathbf{C}}_{EI}}{{\mathbf{{\dot{P}}}}_R}} \right)}}{{{\text{Tr}}\left( {{{{\mathbf{\hat {\underline{C}}}}}_{SR,k}}} \right)}} + {\mu _R}} \right)\\
\end{aligned}
\end{equation*}
\begin{equation*}
\begin{aligned}
   & b_{kj}^D = \\
   & \left\{ \begin{aligned}
   & \frac{{\delta _{kk}^{RD}}}{{{\text{Tr}}\left( {{\mathbf{\hat {\underline{C}}}}_{RD,k}^{}} \right)}} + \frac{{{\nu _R}{\text{Tr}}\left( {{{\mathbf{C}}_{RD,k}}{\text{diag}}\left( {{{\mathbf{{\dot{P}}}}_T}{{{\mathbf{\hat  {\underline{C}}}}}_{RD,k}}{\mathbf{{\dot{P}}}}_T^H} \right)} \right)}}{{{{\left( {{\lambda _1}\left( {{{{\mathbf{\tilde C}}}_{RD,k}}} \right)} \right)}^{ - 1}}{\text{Tr}}\left( {{{{\mathbf{\hat {\underline{C}}}}}_{RD,k}}} \right)}}\\
   & + \frac{{{\mu _{D,k}}{\mathbf{{\dot{p}}}}_{D,k}^H{\text{diag}}\left( {{{\mathbf{{\dot{p}}}}_{D,k}}{\mathbf{{\dot{p}}}}_{D,k}^H} \right){{\mathbf{{\dot{p}}}}_{D,k}}}}{{{{\left( {{\text{Tr}}\left( {{{{\mathbf{\hat  {\underline{C}}}}}_{RD,k}}} \right)} \right)}^{ - 1}}}} +\frac{{{\mu _{D,k}}{\mu _R}{\text{Tr}}\left( {{\mathbf{\underline{C}}}_{RD,k}^3{\mathbf{\Gamma }}_{RD,k}^2} \right)}}{{\tau {\lambda _1}\left( {{{{\mathbf{\tilde C}}}_{RD,k}}} \right){\text{Tr}}\left( {{{{\mathbf{\hat {\underline{C}}}}}_{RD,k}}} \right)}} \\
   & \times \left( {\frac{1}{{{E_T}}} + {\beta _{RD,k}}\left( {{\lambda _1}\left( {{{{\mathbf{\tilde C}}}_{RD,k}}} \right) + {\nu _{D,k}}} \right) + \frac{1}{{{E_T}{\mu _R}}}} \right),j = k \\
   & \frac{{\delta _{kj}^{RD}}}{{{\text{Tr}}\left( {{\mathbf{\hat {\underline{C}}}}_{RD,j}^{}} \right)}} + {\nu _R}\frac{{{\text{Tr}}\left( {{{\mathbf{C}}_{RD,k}}{\text{diag}}\left( {{{\mathbf{{\dot{P}}}}_T}{{{\mathbf{\hat {\underline{C}}}}}_{RD,j}}{\mathbf{{\dot{P}}}}_T^H} \right)} \right)}}{{{{\left( {{\lambda _1}\left( {{{{\mathbf{\tilde C}}}_{RD,k}}} \right)} \right)}^{ - 1}}{\text{Tr}}\left( {{{{\mathbf{\hat {\underline{C}}}}}_{RD,j}}} \right)}} \\
   & + \frac{{{\mu _{D,k}}{\beta _{RD,j}}{\text{Tr}}\left( {{{\mathbf{C}}_{RD,k}}{{\mathbf{{\dot{P}}}}_T}{{{\mathbf{\hat {\underline{C}}}}}_{RD,j}}{\mathbf{{\dot{P}}}}_T^H} \right)}}{{{\text{Tr}}\left( {{{{\mathbf{\hat {\underline{C}}}}}_{RD,k}}} \right)}},j \ne k
\end{aligned}  \right.
\end{aligned}
\end{equation*}

With the above results and the formula $\frac{{T - 2K\tau }}{T}\sum_{k = 1}^K$ $ {{{\log }_2}\left( {1 + {\gamma _k}} \right)}  = \frac{{T - 2K\tau }}{T}{\log _2}\left( {\prod_{k = 1}^K {\left( {1 + {\gamma _k}} \right)} } \right)$ ($\gamma_k$ is defined as $\gamma_k =\min\{\gamma_{SR,k},\gamma_{RD,k}\}$), (\ref{5-2}) can be rewritten as
\begin{equation}\label{5-4}
    \begin{array}{l}
        \mathop {\min }\limits_{{E_{S,1}}, \cdots ,{E_{S,K}};{E_{R,1}}, \cdots ,{E_{R,K}};{A_R},{A_T}} \prod\limits_{k = 1}^K {{{\left( {1 + {\gamma _k}} \right)}^{ - 1}}}  \\
        s.t.{\rm{ }}\left\{ \begin{array}{l}
        {C_1}:\frac{{{E_{S,k}}}}{{\sum\limits_{j = 1}^K {{a^R_{k,j}}{E_{S,j}}}  + \sum\limits_{j = 1}^K {{b^R_{k,j}}{E_{R,j}}}  +  \left( {{\mu _R} + 1} \right){\left( {{\text{Tr}}\left( {{{{\mathbf{\hat {\underline{C}}}}}_{SR,k}}} \right)} \right)^{ - 1}} }} \ge {\gamma _k}, \\
        k = 1,2, \cdots ,K\\
        {C_2}:\frac{{{\text{Tr}}\left( {{\mathbf{\hat {{\underline{C}}}}}_{RD,k}^{}} \right){E_{R,k}}}}{{\sum_{j = 1}^K {b_{k,j}^D{E_{R,j}}}  + {\mu _{D,k}} + 1}} \ge {\gamma _k},k = 1,2, \cdots ,K\\
        {C_3}:{E_{S,k}} \le E_{S,k}^{\max },k = 1,2, \cdots ,K\\
        {C_4}:\sum\limits_{j = 1}^K {{E_{R,j}}}  \le {E_R} \\
        {C_5}:  A_R \in \mathcal{A}_R, A_T \in {\cal A}_T
        \end{array} \right.
    \end{array}
\end{equation}
By using some algebra manipulations on the inequality constraints, $C_1$, $C_2$ can be rewritten as
\begin{equation}\label{5-5}
    \begin{array}{l}
        {C_1}:\sum\limits_{j = 1}^K {a_{k,j}^RE_{S,k}^{ - 1}{E_{S,j}}{\gamma _k}}  + \sum\limits_{j = 1}^K {b_{k,j}^RE_{S,k}^{ - 1}{E_{R,j}}{\gamma _k}}  \\
        + \left( {{\mu _R} + 1} \right){\left( {{\text{Tr}}\left( {{{{\mathbf{\hat C}}}_{SR,k}}} \right)} \right)^{ - 1}}E_{S,k}^{ - 1}{\gamma _k} \le 1\\
        {C_2}:\sum\limits_{j = 1}^K {b_{k,j}^DE_{R,k}^{ - 1}{E_{R,j}}{\gamma _k}}  + \left( {{\mu _{D,k}} + 1} \right)E_{R,k}^{ - 1}{\gamma _k} \le {\text{Tr}}\left( {{\mathbf{\hat {\underline{C}}}}_{RD,k}^{}} \right)
    \end{array}
\end{equation}

Problem (\ref{5-4}) is a combinatorial optimization problem which is in general NP hard. To solve (\ref{5-4}), we propose a JDPO algorithm to find suboptimal solution. Our strategy is as follows. First, using the similar approach as that in \cite{IEEEhowto:WeeraddanaVT11}, we show that the power control problem with fixed $A_R$ and $A_T$ can be approximated as a geometric programming (GP) and solved efficiently using the convex optimization tools \cite{IEEEhowto:BoydOE2007}. Then we present a heuristic approach based on sequential optimization to solve $A_R$ and $A_T$.

\emph{Power Control with Fixed $A_R$ and $A_T$:} From (\ref{5-4}) and (\ref{5-5}), the inequality constraints are posynomial functions \cite{IEEEhowto:BoydOE2007}. As shown in \cite{IEEEhowto:WeeraddanaVT11}, for any ${\gamma _{k}}>0$, $1 +  {\gamma _{k}}$ can be approximated using a monomial function ${\theta _k}\gamma _{k}^{{{\omega _k}}}$ near a point ${{\hat\gamma} _{k}}$, where ${\omega_k} = {\hat\gamma _{k}}/\left( {1 + {\hat\gamma _{k}}} \right)$ and ${\theta_k} = \hat\gamma _{k}^{ - {\omega_k}}\left( {1 + {\hat\gamma _{k}}} \right)$. By using this, the target function of (\ref{5-4}) can be approximated as a monomial function, i.e., $\prod_{k = 1}^K {\theta _k^{ - 1}\gamma _k^{ - {\omega _k}}}$. As a result, the solution for power control problem (with fixed $A_R$ and $A_T$) can be obtained by solving several GPs. The algorithm is summarized in Algorithm \ref{alg1}.

\begin{algorithm}
  \caption{Solve the power control problem with fixed $A_R$ and $A_T$ by GP.}
  \begin{algorithmic}[1]
    \State {\textbf{Initialization:} Let ${{\hat\gamma}^{(i)} _{k}}$ denote the solution of ${{\gamma}_k}$ after the $i$th iteration. Compute the initial value ${{\hat\gamma}^{(0)} _{ k}}$ using (\ref{5-3}). Set a tolerance $\varepsilon$ and the maximum iteration times $L_{GP}$}.
    \State {\textbf{For the $(i+1)$th iteration:}}
        \begin{itemize}
          \item Compute ${\omega_k} = {{\hat\gamma}^{(i)} _{k}}/\left( {1 + {{\hat\gamma}^{(i)} _{k}}} \right)$ and ${\theta_k} = {\left( {\hat \gamma _{k}^{(i)}} \right)^{ - {b_k}}}\left( {1 + \hat \gamma _{ k}^{(i)}} \right)$.
          \item Solve the following GP problem
                    \begin{array}{l}
                        \mathop {\min }\limits_{{E_{S,1}}, \cdots ,{E_{S,K}};{E_{R,1}}, \cdots ,{E_{R,K}};} \prod\limits_{k = 1}^K {\theta _k^{ - 1}\gamma _k^{ - {\omega _k}}} \hspace{.5em} s.t.  \hspace{.5em}  C_1 \sim C_5
                    \end{array}
        \end{itemize}
    \State If $\mathop {\max }\limits_{k\! =\! 1, \cdots ,K} \left| {{\hat\gamma} _{ k}^{(i+1)}\!\! -\!\! {{\hat\gamma}^{(i)} _{ k}}} \right| \!<\! \varepsilon $ or $i\!+\!1\!=\!L_{GP}$, stop. Otherwise, set $i\! =\! i\!+\! 1$ and go back to step 2.
  \end{algorithmic}
  \label{alg1}
\end{algorithm}

\emph{Optimization of $A_R$ and $A_T$:} An obvious approach to obtain $A_R$ and $A_T$ is via 2-D search. However, the complexity is very high since  the power control problem described in the above should be computed for each $(A_R,A_T)$. In this work, we present a heuristic approach to find suboptimal solution. In this approach, $A_R$ ($A_T$) is searched over a subset ${\cal \tilde{A}}_R$ (${\cal \tilde{A}}_T$), whose elements are sampled from ${\cal {A}}_R$ (${\cal {A}}_T$). For example, ${\cal \tilde{A}}_R$ can be selected as $\{K, 2K, \cdots, \left\lfloor {\frac{{{N_R}}}{K}} \right\rfloor K\}$. The motivation of this idea is that the target function does not change dramatically as $A_R$ and $A_T$ increase/decrease by a small step as observed from the numerical results. Moreover, instead of 2-D search, ${\cal {A}}_R$ and ${\cal {A}}_T$ are optimized sequentially over the sampled subsets ${\cal \tilde{A}}_R$ and ${\cal \tilde{A}}_T$ via 1-D search. The process repeats several times (much less than $\min\{|{\cal \tilde{A}}_R|,|{\cal \tilde{A}}_T|\}$) to improve the solution.

A summary of the JDPO algorithm is presented in Algorithm \ref{alg2}. The algorithm converges to a local optimum since the target function is improved in each iteration. The complexity of Algorithm \ref{alg2} is upper bounded as $L( {| {{{\tilde {\cal A}}_R}} | + | {{{\tilde {\cal A}}_T}} |} ){L_{GP}}{{\cal C}_{GP}}$, where ${{\cal C}_{GP}}$ is the complexity to solve GP when $A_R=N_R$ and $A_T=N_T$. Usually, GP is solved using inner point method with polynomial time. The exact expression of ${{\cal C}_{GP}}$ is quite difficult and related with the structure of the problem. Some insights on ${{\cal C}_{GP}}$ can be found in [37, Sec. 11.5].

\begin{algorithm}[t]
  \caption{JDPO algorithm for SE optimization.}
  \begin{algorithmic}[1]
    \State \textbf{Initialization:} Set an initial $A_T$, i.e., $A_T^{(0)}$, and a maximum repeat times $L$. Select the subsets ${\cal \tilde{A}}_R$ and ${\cal \tilde{A}}_T$.
    \ForAll {$l=1,\cdots,L$}
    \State {Set $A_T=A_T^{(l-1)}$ and compute $\textsf{SE}$ using Algorithm \ref{alg1} for all $A_R \in {\cal \tilde{A}}_R$.}
    \State {Update $A_R^{(l)} = \arg \mathop {\max }\limits_{{A_R}} \textsf{SE}$}.
    \State {Set $A_R=A_R^{(l)}$ and compute $\textsf{SE}$ using Algorithm \ref{alg1} for all $A_T \in {\cal \tilde{A}}_T$.}
    \State Update $A_T^{(l)} = \arg \mathop {\max }\limits_{{A_T}} \textsf{SE}$.
    \EndFor
  \end{algorithmic}
  \label{alg2}
\end{algorithm}

\begin{remark}\label{rm2}
\emph{(JDPO for EE Optimization)} Instead of the SE optimization, we can also formulate an EE (defined as the SE divided by total transmit power \cite{IEEEhowto:Ngo2013TC}) optimization problem subject to a target SE ${R_T}$
\begin{equation}\label{5-6}
\begin{array}{ll}
    \mathop {\max }\limits_{{E_{S,1}}, \cdots ,{E_{S,K}};{E_{R,1}}, \cdots ,{E_{R,K}};A_R,A_T} \frac{{\sum\limits_{k = 1}^K {\min \left\{ {\textsf{R}_{SR,k}^{{\rm{HIA}}},\textsf{R}_{RD,k}^{{\rm{HIA}}}} \right\}} }}{{\sum\limits_{k = 1}^K {{E_{S,k}}}  + \sum\limits_{k = 1}^K {{E_{R,k}}} }} \\
    \hspace{3em} s.t.\left\{ \begin{array}{l}
    {\sum\limits_{k = 1}^K {\min \left\{ {{\textsf{R}_{SR,k}^{\rm HIA}},{\textsf{R}_{RD,k}^{\rm HIA}}} \right\}} = {R_T}} \\
    0 \le {E_{S,k}} \le E_{S,k}^{\max },k = 1, \cdots ,K\\
    \sum\limits_{k = 1}^K {{E_{R,k}}}  \le {E_R^{\max}}
    \end{array} \right.
    \end{array}
\end{equation}
Using (\ref{5-3}), the above problem is equivalent to $\mathop {\max }\limits_{{E_{S,1}}, \cdots ,{E_{S,K}};{E_{R,1}}, \cdots ,{E_{R,K}};A_R,A_T} {\sum_{k = 1}^K {{E_{S,k}}}  + \sum_{k = 1}^K {{E_{R,k}}} \hspace{.3em} s.t.  } $ ${C_1\sim C_5}$ and $ \prod_{k = 1}^K {\left( {1 + {\gamma _k}} \right)}  = {2^{\frac{T}{{T - 2K\tau }}{R_T}}} $. With the technique in \cite{IEEEhowto:WeeraddanaVT11}, the equality constraint can be converted to a monomial, and the power control problem (with fixed $A_R$ and $A_T$) becomes a GP. Thus, (\ref{5-6}) can be solving by using a similar JDPO algorithm as that in Algorithm \ref{alg2}.
\end{remark}

\section{Simulation Results and Discussion}
This section presents the simulation results to verify the analyses in the previous sections. Throughout this section, we set $\nu_{S,k}=\nu_{S}$, $\nu_{D,k}=\nu_{D}$, $\mu_{D,k}=\mu_{D}$ for convenience. The correlation matrices of desired channels are generated with the exponential correlation model \cite{IEEEhowto:LoykaCL01}
\begin{equation*}\label{6-1}
 {\left[ {{{\bf{C}}_{i,k}}{\text{ or }}{{{\bf{\tilde C}}}_{i,k}}} \right]_{l,j}} = \left\{ {\begin{array}{*{20}{l}}
{r_{i,k}^{j - l},l \le j}\\
{{{\left( {r_{i,k}^{l - j}} \right)}^*},l > j}
\end{array}} \right.,i \in \left\{ {SR,RD} \right\}
\end{equation*}
The model approximates the property of ULA, where the correlation between adjacent antennas is $|r_{i,k}|\in [0,1]$ and the phase of $r_{i,k}$ describes the angle of arrival/departure as seen from the array. \cite{IEEEhowto:HammarwallTSP09} shows how to map some of the parameters of ULA to this model. The correlation matrices of EI channel ${{\bf{C}}_{EI}}$ and ${{{\bf{\tilde C}}}_{EI}}$ are generated similarly with a parameter $r_{EI}$. For convenience, we let $|r_{i,k}|=r_0$, $\forall k\in\{1,\cdots, K\}$. The phases of $r_{i,k}$ and $r_{EI}$ are uniformly selected from $[0,\pi]$.

We assume that 25$\sim$35dB EI cancellation can be provided by passive EI suppression techniques (more than 40dB cancellation has been reported by using such techniques for infrastructure node \cite{IEEEhowto:EverettTWC13}). The variances of EI channel (after passive cancellation) and desired channels are selected as $\beta_{EI}/\beta_{i,k} \in[0,25]$dB. With the path loss model in \cite{IEEEhowto:EverettTWC13}, the above range corresponds to the setup with the distances from sources/destinations to relay varying from 250m to 500m and 10m segregation between relay transmit and receive arrays.

\subsection{Impact of Hardware Impairments}

This subsection considers the effect of hardware impairments on SE of MM-FDR. The channel coherent time is set to $T=300$ and the length of pilot sequence is $\tau=2$.

\begin{figure}[t]
\centering
\includegraphics[width=8.5cm]{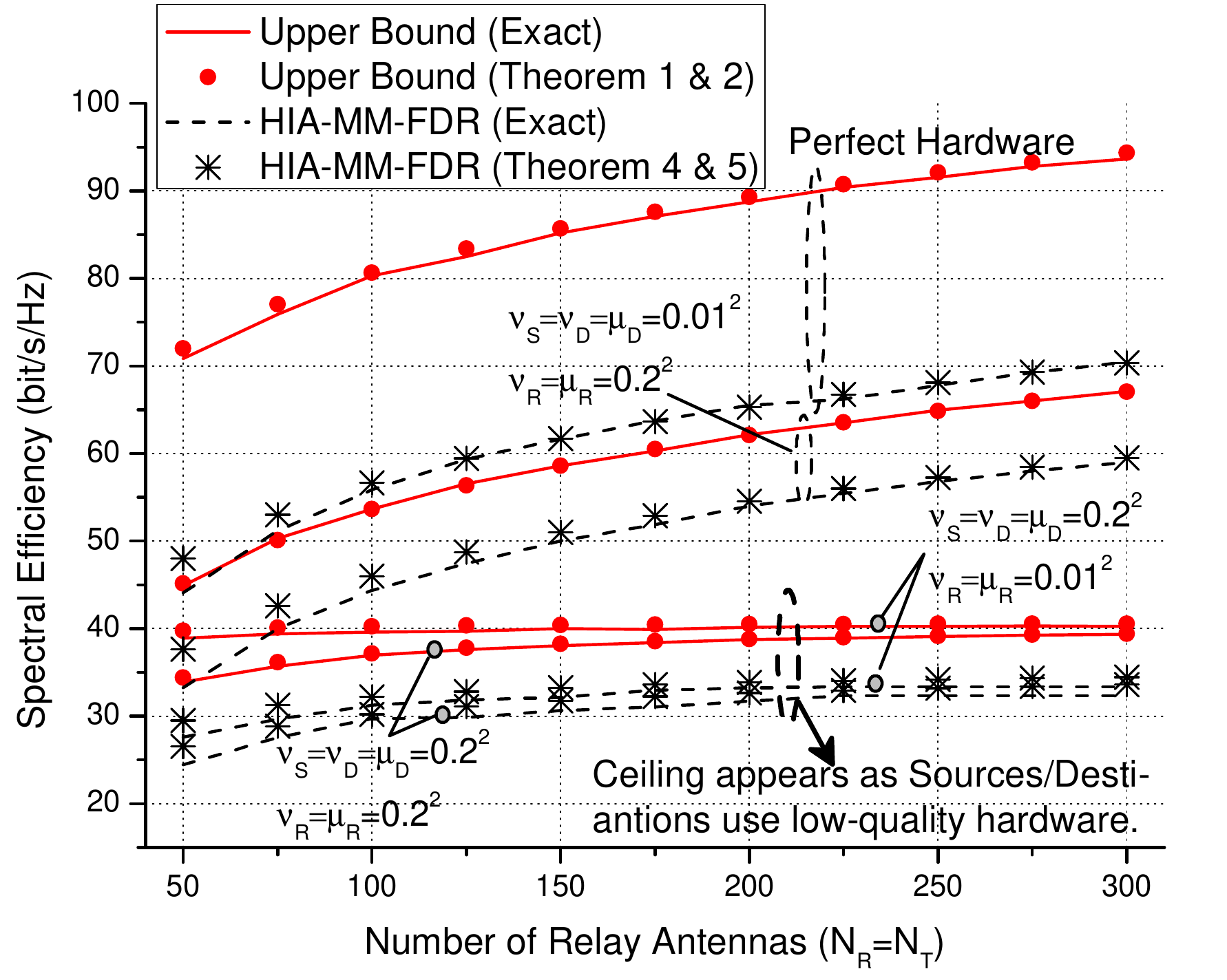}
\caption{SE of traditional MM-FDR with single antenna at sources/destinations, where $N_S=N_D=1$, $K=10$, $E_{S,k}=E_{R,k}=8$dB, $E_T=10$dB, $\beta_{SR,k}=\beta_{RD,k}=\beta_{EI} = 1$, $r_0=0.2$, $|r_{EI}|=0.8$, $A_R=\max\{K,\left\lfloor {\frac{2}{3}{N_R}} \right\rfloor \}$, $A_T=\max\{K,\left\lfloor {\frac{2}{3}{N_T}} \right\rfloor \}$.} \label{fig1a}
\end{figure}

\begin{figure}[t]
\centering
\includegraphics[width=8.5cm]{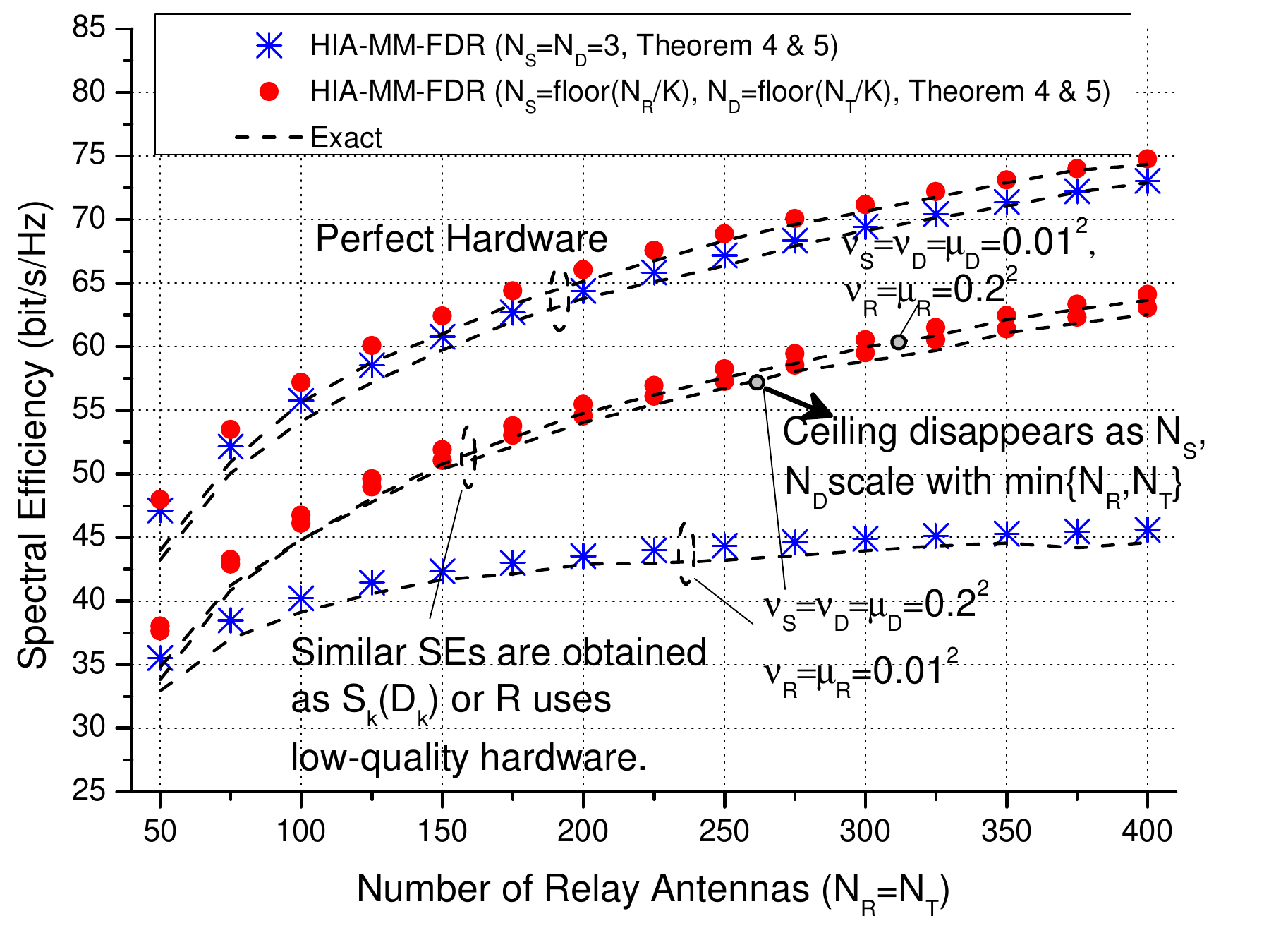}
\caption{SE of MM-FDR with multiple antennas at sources/destinations. The setup is the same with Fig \ref{fig1b}.} \label{fig1b}
\end{figure}

Fig. \ref{fig1a} shows the SE of MM-FDR with single antenna at sources/destinations with different levels of hardware impairments. The SEs based on transceiver scheme in section IV (which achieves the upper bound when $N_S=N_D=1$) and HIA scheme (with ${\bf p}_{S,k}={\bf p}_{D,k}=1$) are simulated. From Fig. \ref{fig1a}, the SE is more sensitive to hardware impairments at sources and destinations. When $\nu_S=\nu_D=\mu_D=0.2^2$, the SE approaches to a finite ceiling quickly as the number of relay antennas increases. Similar results can be observed when sources and destinations are equipped with multiple but fixed number of antennas in Fig. \ref{fig1b}. However, the result changes when $N_S$ and $N_D$ scale with $\min\{N_R,N_T\}$ and HIA scheme is used. From Fig. \ref{fig1b}, as ${N_S} = \left\lfloor {\frac{{{N_R}}}{K}} \right\rfloor $ and ${N_D} = \left\lfloor {\frac{{{N_T}}}{K}} \right\rfloor $, similar SEs are achieved as sources/destinations or relay employ low-quality hardware, and no performance ceiling appears. This demonstrates the validness of HIA scheme. At last, it is seen that the asymptotic results in Theorem \ref{t1} \& \ref{t2} and Theorem \ref{t4} \& \ref{t5} match well with exact results.

\begin{figure}[t]
\centering
\includegraphics[width=8.5cm]{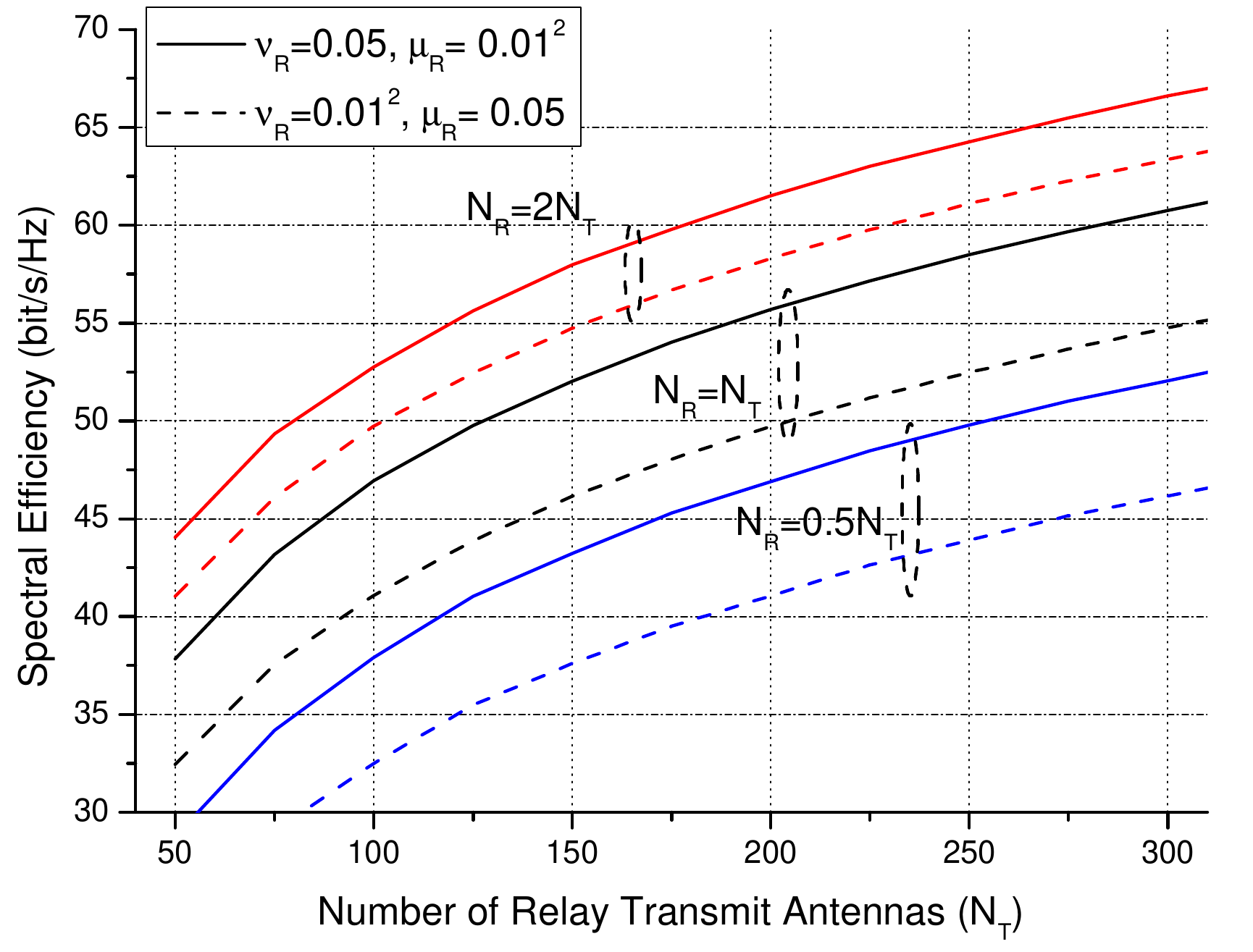}
\caption{Comparison between transmit/receive imperfections at the relay, where $\nu_S=\nu_D=\mu_D=0.01^2$, $K=10$, $E_{S,k}=E_{R,k}=5$dB, $E_T=10$dB, $\beta_{SR,k}=\beta_{RD,k}=1$, $\beta_{EI}=5$dB, $r_0=0.2$, $|r_{EI}|=0.8$, $A_R=\max\{K,\left\lfloor {\frac{2}{3}{N_R}} \right\rfloor \}$, $A_T=\max\{K,\left\lfloor {\frac{2}{3}{N_T}} \right\rfloor \}$, ${N_S} = \left\lfloor {\frac{{{N_R}}}{K}} \right\rfloor $, ${N_D} = \left\lfloor {\frac{{{N_T}}}{K}} \right\rfloor $.} \label{fig2}
\end{figure}

As the performance is affected by impairments of both transmit and receive RF chains at the relay, we compare the effect of transmit and receive imperfections on SE in Fig. \ref{fig2}. We assume that the sources and destinations use high-quality hardware ($\nu_S=\nu_R=\mu_R=0.01^2$). From the figure, the effect of receive imperfection is more detrimental. The reason is that, with EI at the relay, the power of distortion noise caused by receive imperfection is much stronger than that due to transmit imperfection. However, the performance difference decreases as $N_R$ increases from $0.5N_T$ to $N_R=2N_T$, since the power of effective receive distortion scales as ${\cal O}(KN_R^{-1})$ (see Table \ref{tab1}). This implies that use relatively higher-quality hardware or more antennas at the receive side of relay is beneficial. Using a similar setup, one can obtain a parallel conclusion for destination, i.e., the receive imperfection is more harmful than transmit imperfection. This is because that the transmit imperfection of destination only induces larger channel estimator errors, which is a part of the received signal at the destination. Thus, based on model (\ref{2-4}), the distortion noise due to receive imperfection will be more detrimental.

\subsection{Impact of Number of S-D Pairs and Channel Coherent Time}

\begin{figure}[t]
\centering
\includegraphics[width=8.5cm]{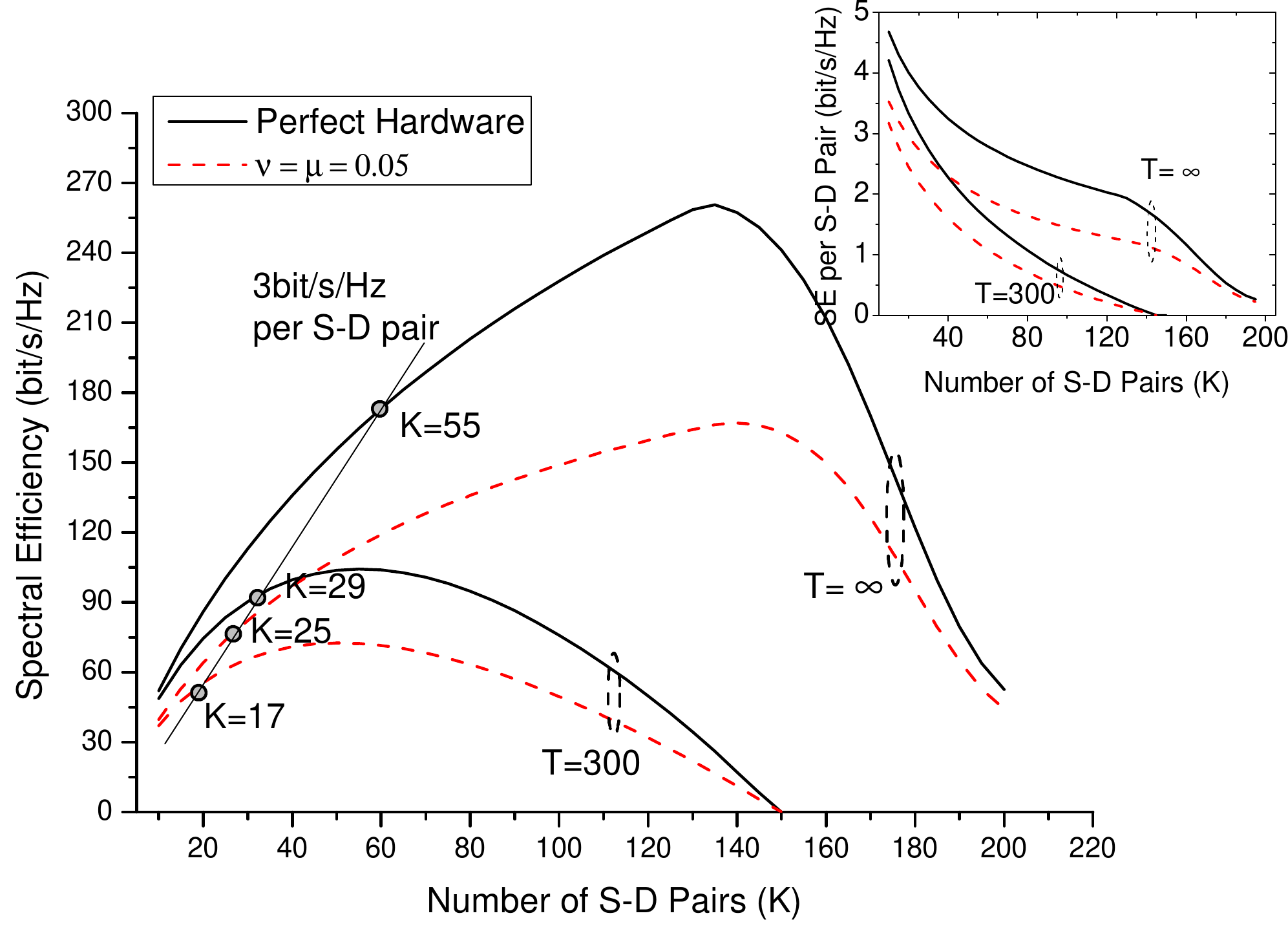}
\caption{SE v.s. number of S-D pair $K$, where $\nu_S=\nu_D=\nu_R = \nu$, $\mu_D=\mu_R=\mu$, $E_{S,k}=E_{R,k}=5$dB, $E_T=10$dB, $\beta_{SR,k}=\beta_{RD,k}=1$, $\beta_{EI}=5$dB, $r_0=0.4$, $|r_{EI}|=0.7$, $N_R=N_T=200$, $N_S=N_D=10$, $A_R=\max\{K,\left\lfloor {\frac{2}{3}{N_R}} \right\rfloor \}$, $A_T=\max\{K,\left\lfloor {\frac{2}{3}{N_T}} \right\rfloor \}$.} \label{fig3}
\end{figure}

Fig. \ref{fig4} shows the SE as a function of the number of S-D pairs $K$, where the length of pilot sequence of each source or destination is set to 1. When $T=300$, the figure reveals that the SE will not increase without bound as $K$ increases since SE is ultimately limited by the channel coherent time. In fact, as the number of S-D pair approaches to $T/2$, the SE converges to zero since all time is allocated to pilot phase. Meanwhile, the SE per S-D pair is a strictly decreasing function of $K$. Similar result can be observed as the channel coherent time is very long ($T\to \infty$). This is because that the MUI, EI and distortion noises become limiting factors when $K$ is large as seen from the scaling behaviors in Table \ref{tab1} and Table \ref{tab2}. This implies that the number of S-D pairs should be limited in order to ensure a SE guarantee for each S-D pair, and this number decreases as the quality of hardware degrades.

\subsection{Comparison with Relevant Schemes}

In this subsection, the SE of HIA-MM-FDR is compared with the massive MIMO HDR (MM-HDR) \cite{IEEEhowto:SuraICC13} and MM-FDR with ZF-based transceiver (ZF-MM-FDR) \cite{IEEEhowto:NgoJSAC14}. The channel coherent time is set to $T=300$ and the length of pilot sequence is $\tau=2$. The power of pilot symbol is $E_T=10$dB. For HIA-MM-FDR, the maximum power constraints at sources and relay are set to $E_{S_k}^{\max} =5$dB and $E_{R}^{\max} =K E_{S_k}^{\max}$. Without JDPO, we set $E_{S,k} =E_{S,k}^{\max}$, $E_{R,k}=E_{R}^{\max}/K$dB. Moreover, $A_R$ and $A_T$ are set to $A_R=\max\{K,\left\lfloor {\frac{2}{3}{N_R}} \right\rfloor \}$ and $A_T=\max\{K,\left\lfloor {\frac{2}{3}{N_R}} \right\rfloor \}$. When JDPO is applied, $E_{S,k}$, $E_{R,k}$, $A_R$ and $A_T$ are determined by Algorithm \ref{alg2}. The maximum repeat time $L$ of JDPO is set to 3 and the elements of subset ${\cal \tilde A}_R$ are picked uniformly from ${\cal A}_R$ with step $\max\{10, \left\lfloor {\frac{N_R}{K}{}} \right\rfloor \}$. ${\cal \tilde A}_T$ is obtained with the similar approach. The above parameters are chosen so that the performance loss due to the suboptimal search approach to obtain $A_R$ and $A_T$ is negligible.

\begin{figure}[t]
\centering
\includegraphics[width=8.5cm]{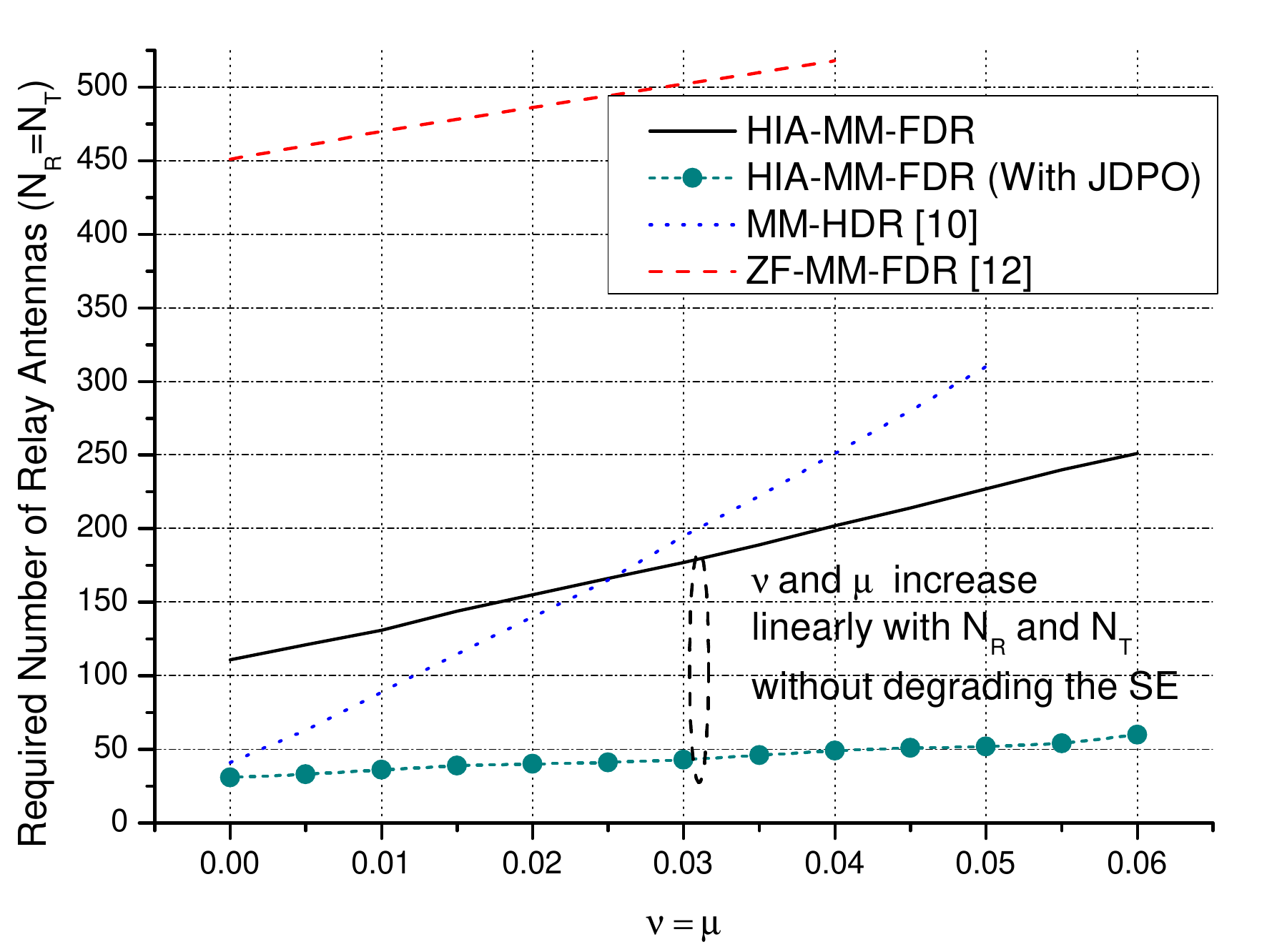}
\caption{Number of relay antennas required to achieve 3bit/s/Hz SE per S-D pair with different hardware qualities, where $K=10$, $\beta_{SR,k}=\beta_{RD,k}=1$, $\beta_{EI}=5$dB, ${N_S} = \left\lfloor {\frac{{{N_R}}}{K}} \right\rfloor $, ${N_D} = \left\lfloor {\frac{{{N_T}}}{K}} \right\rfloor $, $r_0=0.4$, $|r_{EI}|=0.7$.} \label{fig4}
\end{figure}

Fig. \ref{fig4} simulates the number of relay antennas required to achieve 3bit/s/Hz SE per S-D pair (which ideally can support the 64-QAM transmission with 1/2 channel code) with different levels of hardware quality, where we set $\nu_S=\nu_D=\nu_R = \nu$ and $\mu_D=\mu_R=\mu$. The figure reveals a tradeoff between the number of antennas and hardware quality, that is (as discussed in section V-D), by increasing $N_R$ and $N_T$, we can reduce $\nu$ and $\mu$ linearly without degrading the SE. Meanwhile, it is seen that the proposed HIA-MM-FDR reduces the required number of relay antennas significantly when compared with ZF-MM-FDR. Moreover, the HIA-MM-FDR outperforms the MM-HDR for large $\nu$ and $\mu$. This is because that the distortion noises become the main limiting factor in this case when compared with EI.

\begin{figure}[t]
\centering
\includegraphics[width=8.5cm]{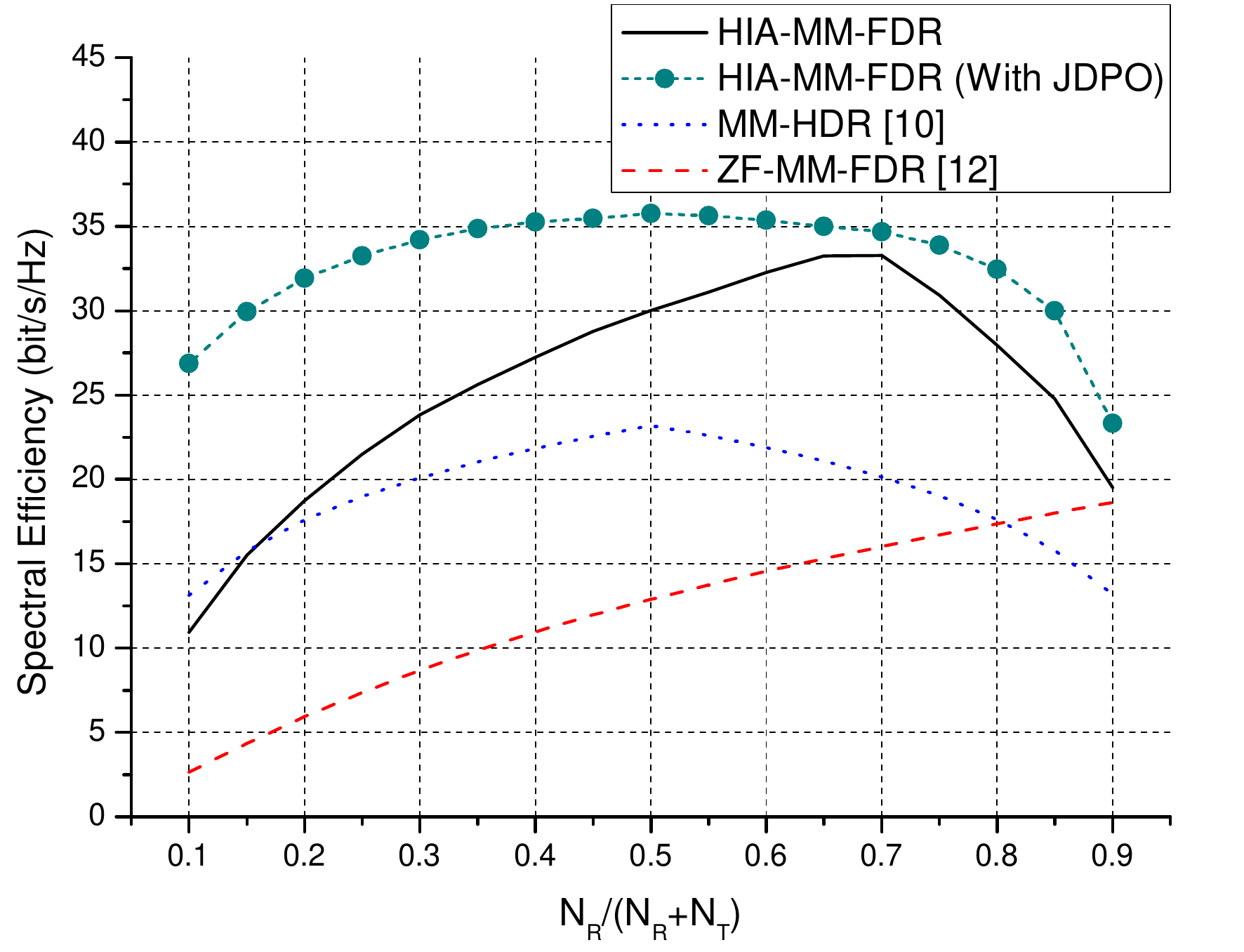}
\caption{SE comparison with asymmetric numbers of transmit and receive relay antennas, where $K=10$, ${N_S} = \left\lfloor {\frac{{{N_R}}}{K}} \right\rfloor $, ${N_D} = \left\lfloor {\frac{{{N_T}}}{K}} \right\rfloor $, $\beta_{SR,k}=\beta_{RD,k}=1$, $\beta_{EI}=5$dB, $r_0=0.4$, $|r_{EI}|=0.7$, $\nu_S=\nu_D=\nu_R =0.05$, $\mu_D=\mu_R =0.05$.} \label{fig5}
\end{figure}

The effect of asymmetric numbers of transmit and receive antennas at the relay is shown in Fig. \ref{fig5}, where we set $N_R+N_T=200$. It is shown that allocating more antennas to the receive side of relay is beneficial for HIA-MM-FDR (without JDPO) and ZF-MM-FDR. The reason is that, in the considered setup, the limiting factors of system performance are EI and distortion noise due to receive imperfection at the relay. From Theorem \ref{t4}, these factors can be suppressed by receive antenna array of relay. However, when JDPO is applied, it is optimal to set $N_R=N_T$. This is because that the JDPO algorithm in fact tries to balance the achievable rates of each hop by power control and adjusting $A_R$ and $A_T$. This makes it unnecessary to allocate more antennas to receive side.

\begin{figure}[t]
\centering
\includegraphics[width=8.5cm]{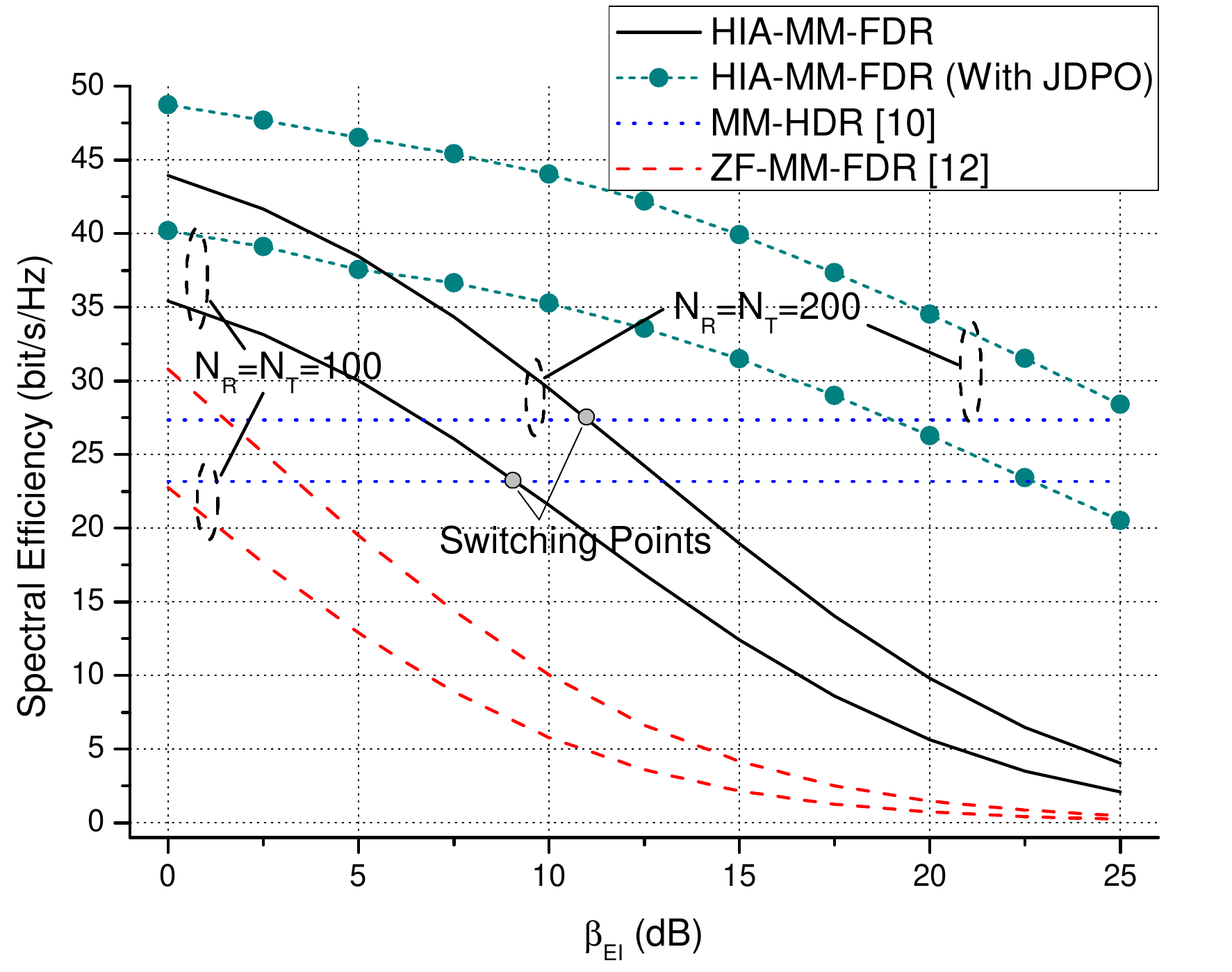}
\caption{SE v.s. variance of echo interference channel $\beta_{EI}$, where $K=10$, ${N_S} = \left\lfloor {\frac{{{N_R}}}{K}} \right\rfloor $, ${N_D} = \left\lfloor {\frac{{{N_T}}}{K}} \right\rfloor $, $\beta_{SR,k}=\beta_{RD,k}=1$, $r_0=0.4$, $|r_{EI}|=0.7$, $\nu_S=\nu_D=\nu_R =0.05$, $\mu_D=\mu_R =0.05$.} \label{fig6}
\end{figure}

Fig. \ref{fig6} shows the SE of HIA-MM-FDR as a function of variance for EI channel $\beta_{EI}$. As expected, there exists a switching point between HIA-MM-FDR and MM-HDR as $\beta_{EI}$ increases. By increasing $N_R$ and $N_T$, the constraint on $\beta_{EI}$ for HIA-MM-FDR to achieve a performance gain relaxes, which indicates that HIA-MM-FDR becomes more attractive when the number of relay antennas is large. Moreover, Fig. \ref{fig6} demonstrates that the proposed JDPO algorithm can reduce the constraint on $\beta_{EI}$ significantly. In particular, as $N=200$, a SE gain of 7.5bit/s/Hz can be achieved by HIA-MM-FDR when compared to MM-HDR as $\beta_{EI}$ is $20$dB.

\begin{figure}[t]
\centering
\includegraphics[width=8.5cm]{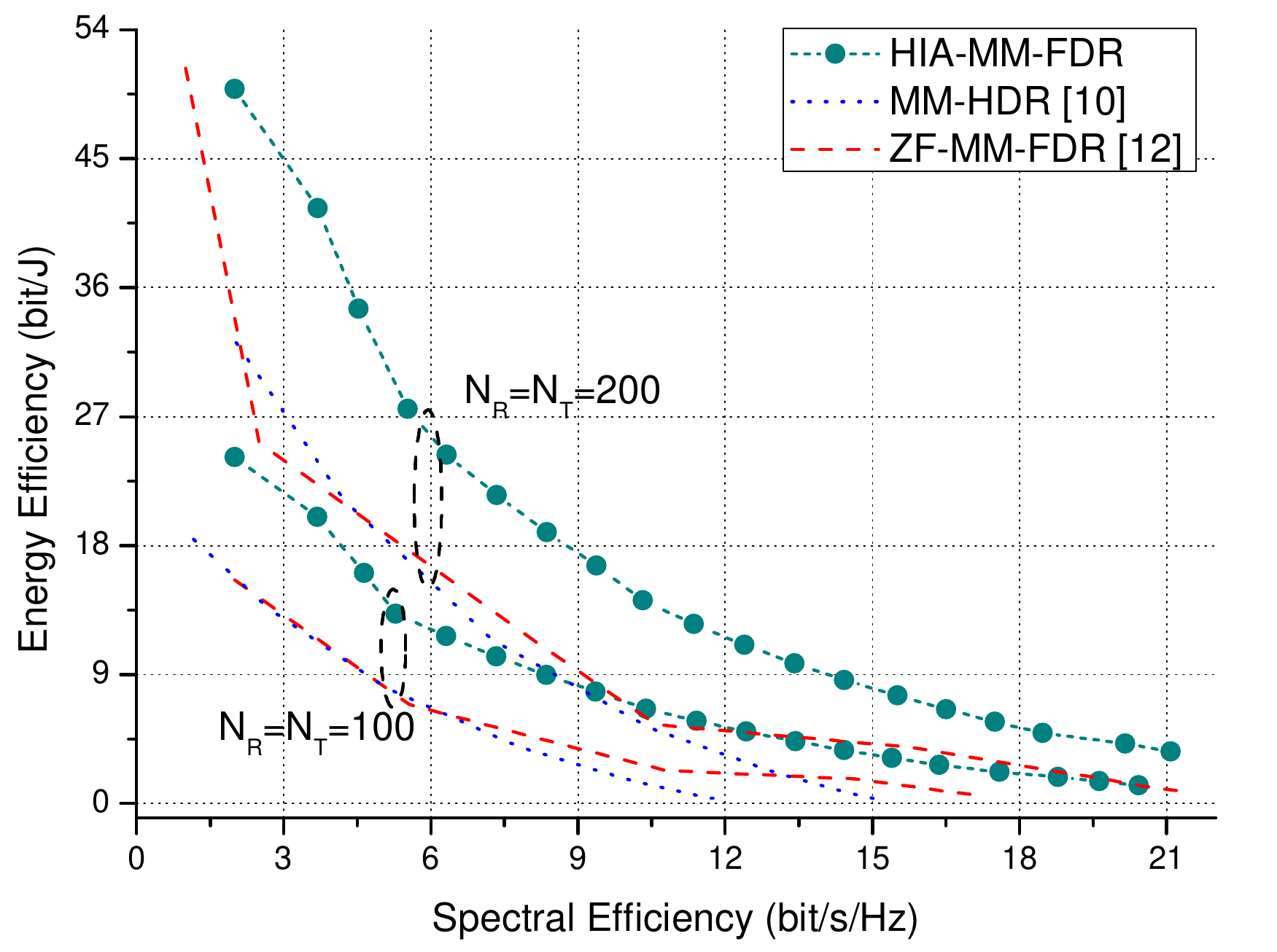}
\caption{Comparison on spectral-energy efficiency tradeoff, where $K=10$, $\beta_{EI}=10$dB, ${N_S} = \left\lfloor {\frac{{{N_R}}}{K}} \right\rfloor $, ${N_D} = \left\lfloor {\frac{{{N_T}}}{K}} \right\rfloor $, $r_0=0.2$, $|r_{EI}|=0.8$, $\nu_S=\nu_D=\nu_R =0.05$, $\mu_D=\mu_R =0.05$.} \label{fig7}
\end{figure}

Fig. \ref{fig7} considers the SE-EE tradeoff of different schemes. The large-scale fading coefficients of channels are set to
\begin{equation*}\label{6-2}
    \begin{aligned}
       \{\beta_{SR,1},\cdots,\beta_{SR,K}\} = \left\{0.818,0.052, 1.01, 0.026, 0.016,\right. \\
        \left.0.803, 0.051, 0.383, 2.85, 0.448\right\} \\
       \{\beta_{RD,1},\cdots,\beta_{RD,K}\} = \left\{1.187,0.011, 0.724, 2.11, 0.580,\right. \\
        \left. 0.012, 0.147, 0.085, 0.434, 0.458\right\}
    \end{aligned}
\end{equation*}
which is a realization generated with the model in \cite{IEEEhowto:NgoJSAC14}. The EEs of HIA-MM-FDR and ZF-MM-FDR are optimized by solving the problem in Remark \ref{rm2} with JDPO algorithm and [12, Algorithm 1], respectively. It is observed that HIA-MM-FDR achieves better SE-EE tradeoff when compared with ZF-MM-FDR. The gain is mainly due to the optimization of $A_R$ and $A_T$. This reason is that, if we fix $A_R=N_R$ and $A_T=N_T$, the JDPO algorithm for EE optimization is similar to [12, Algorithm 1]. The only difference is that the power for each steam at the relay, i.e., $E_{R,k}$, is optimized in JDPO algorithm and [12, Algorithm 1] optimizes only the total power of relay.

\section{conclusions}

This paper considers the transceiver design of MM-FDR with hardware impairments. A low complexity HIA scheme is proposed to mitigate the distortion noises by exploiting the statistical knowledge of channels and antenna arrays at sources and destinations. A joint degree of freedom and power optimization algorithm is presented to further optimize the SE of HIA-MM-FDR. The analytic results demonstrate that the proposed scheme can mitigate the ``ceiling effect" appears in traditional MM-FDR protocol, if the numbers of antennas at sources and destinations can scale with that at the relay. Moreover, simulation results show that the HIA-MM-FDR outperforms MM-FDR with traditional transceiver scheme.

\section*{Appendix}

\subsection*{A. Useful Lemmas Related to the Channel Estimates}

\begin{lemma}\label{lm1}
Let ${\bf \hat {\underline{h}}}_{SR,k}$ and ${\bf \hat {\underline{h}}}_{RD,k}$ be the LMMSE estimates of ${\bf {\underline{h}}}_{SR,k}$ and ${\bf {\underline{h}}}_{RD,k}$, respectively. Define $\delta _{k}^{SR} = \mathbb E\left[ {{{\left| {\Delta {\bf{\underline{h}}}_{SR,k}^H{{{\bf{\hat {\underline{h}}}}}_{SR,k}}} \right|}^2}} \right]$ and $\delta _{k}^{RD} = \mathbb E\left[ {{{\left| {\Delta {\bf{\underline{h}}}_{RD,k}^H{{{\bf{\hat {\underline{h}}}}}_{RD,k}}} \right|}^2}} \right]$. In the large-($N_R,N_T$) regime, we have
\begin{equation}\label{X1}
\begin{aligned}
& \delta _{k}^{SR} =
 \frac{\nu _{S,k}}{\tau}{\left( {{\rm{Tr}}\left( {{{{\bf{\hat {\underline{C}}}}}_{SR,k}}} \right)} \right)^2}{\mathbf{{\dot{p}}}}_{S,k}^H{\text{diag}}\left( {{\mathbf{{\dot{p}}}}_{S,k}^H{{\mathbf{{\dot{p}}}}_{S,k}}} \right){{\mathbf{{\dot{p}}}}_{S,k}}\\
& + \frac{{{\mu _R}}}{\tau }\left( {\frac{{\rm{1}}}{{{E_T}}} + {\frac{1}{{{\mu _R}{E_T}}}} + {\beta _{SR,k}}{\rm{Tr}}\left( {{{{\bf{\tilde C}}}_{SR,k}}{{\bf{\Omega }}_{S,k}}} \right)} \right) \\
& \times {\rm{Tr}}\left( {{\bf{\underline{C}}}_{SR,k}^3{\bf{\Gamma }}_{SR,k}^2} \right)
\end{aligned}
\end{equation}
The expression of $\delta _{kj}^{RD}$ can be obtained by replacing ${{{\bf{{\underline{C}}}}}_{SR,k}}$, ${{{\bf{{\underline{\hat C}}}}}_{SR,k}}$, ${{{\bf{{\Gamma}}}}_{SR,k}}$, $\nu_{S,k}$ and ${\bf {\dot{p}}}_{S,k}$ in (\ref{X1}) with ${{{\bf{{\underline{C}}}}}_{RD,k}}$, ${{{\bf{{\underline{\hat C}}}}}_{RD,k}}$, ${{{\bf{{\Gamma}}}}_{RD,k}}$, $\nu_{D,k}$ and ${\bf {\dot{p}}}_{D,k}$, respectively.
\end{lemma}
\begin{proof}
We show the proof for $\delta _{k}^{SR}$ and the derivation for $\delta _{k}^{RD}$ is similar. According to the uncorrlation between $\Delta {\mathbf{\underline{h}}}_{SR,k}^H$ and ${{{\mathbf{\hat {\underline{h}}}}}_{SR,k}}$, we have $\delta _{k}^{SR} = \mathbb E[ {{{| {\Delta {\mathbf{\underline{h}}}_{SR,k}^H{{{\mathbf{\hat {\underline{h}}}}}_{SR,k}}} |}^2}} ] = \mathbb E[ {{{| {{\mathbf{{\underline{h}}}}_{SR,k}^H{{{\mathbf{\hat {\underline{h}}}}}_{SR,k}}} |}^2}} ] - {( {{\text{Tr}}( {{{{\mathbf{\hat {\underline{C}}}}}_{SR,k}}} )} )^2}$. With (\ref{4-7}) and (\ref{4-8}), $\mathbb E[ {{{| {{\mathbf{{\underline{h}}}}_{SR,k}^H{{{\mathbf{\hat {\underline{h}}}}}_{SR,k}}} |}^2}} ]$ can be written as
\begin{equation}\label{X3}
\begin{aligned}
  & \mathbb E\left[ {{{\left| {{\mathbf{\underline{h}}}_{SR,k}^H{{{\mathbf{\hat {\underline{h}}}}}_{SR,k}}} \right|}^2}} \right] = \mathbb E\left[ {{{\left| {{\mathbf{{\underline{h}}}}_{SR,k}^H{{\mathbf{{\underline{C}}}}_{SR,k}}{{\mathbf{\Gamma }}_{SR,k}}{{\mathbf{{\underline{h}}}}_{SR,k}}} \right|}^2}} \right] \\
  & + {\left( {\frac{1}{{\tau {E_T}}}} \right)^2} \mathbb E\left[ {{{\left| {{\mathbf{{\underline{h}}}}_{SR,k}^H{{\mathbf{{\underline{C}}}}_{SR,k}}{{\mathbf{\Gamma }}_{SR,k}}{\mathbf{{\dot{P}}}}_R^H{{\mathbf{H}}_{SR,k}}{\mathbf{T}}_{SR,k}^{}\phi _{}^*} \right|}^2}} \right] \\
  &  + {\left( {\frac{1}{{\tau {E_T}}}} \right)^2} \mathbb E\left[ {{{\left| {{\mathbf{{\underline{h}}}}_{SR,k}^H{{\mathbf{{\underline{C}}}}_{SR,k}}{{\mathbf{\Gamma }}_{SR,k}}{\mathbf{{\dot{P}}}}_R^H{{\mathbf{R}}_{SR,k}}\phi _{}^*} \right|}^2}} \right] \\
  & + {\left( {\frac{1}{{\tau {E_T}}}} \right)^2}\mathbb E\left[ {{{\left| {{\mathbf{{\underline{h}}}}_{SR,k}^H{{\mathbf{\underline{C}}}_{SR,k}}{{\mathbf{\Gamma }}_{SR,k}}{\mathbf{{\dot{P}}}}_R^H{{\mathbf{N}}_{SR,k}}\phi _{}^*} \right|}^2}} \right] \\
\end{aligned}
\end{equation}
In the large-($N_R,N_T$) regime, the result in [17, Lemma 4 (iv)] shows that $ \mathbb E\left[ {{{\left| {{\mathbf{{\underline{h}}}}_{SR,k}^H{{\mathbf{{\underline{C}}}}_{SR,k}}{{\mathbf{\Gamma }}_{SR,k}}{{\mathbf{{\underline{h}}}}_{SR,k}}} \right|}^2}} \right] = {\left( {{\text{Tr}}\left( {{{{\mathbf{\hat {\underline{C}}}}}_{SR,k}}} \right)} \right)^2}$. Replacing ${\mathbf{\underline{h}}}_{SR,k}$ with ${\mathbf{{\dot{P}}}}_R^H{{\mathbf{H}}_{SR,k}}{{\mathbf{{\dot{p}}}}_{S,k}}$ and using the following large $N_R$ approximation (obtained based on [17, Lemma 4 (ii)])
\begin{equation}\label{X4}
\begin{aligned}
&  \frac{1}{{{N_R}}}{{\bf{H}}_{SR,k}}{{\bf{{\dot{P}}}}_R}{{\bf{\underline{C}}}_{SR,k}}{{\bf{\Gamma }}_{SR,k}}{\bf{{\dot{P}}}}_R^H{{\bf{H}}_{SR,k}}\\
& = \frac{1}{{{N_R}}}{\beta _{SR,k}}{\bf{\tilde C}}_{SR,k}^{1/2}{\bf{X}}_{SR,k}^H{\bf{C}}_{SR,k}^{1/2}{{\bf{{\dot{P}}}}_R}{{\bf{\underline{C}}}_{SR,k}} \\
& \times {{\bf{\Gamma }}_{SR,k}}{\bf{{\dot{P}}}}_R^H{\bf{C}}_{SR,k}^{1/2}{{\bf{X}}_{SR,k}}{\bf{\tilde C}}_{SR,k}^{1/2}\\
& = \frac{1}{{{N_R}}}{\beta _{SR,k}}{\rm{Tr}}\left( {{{\bf{\underline{C}}}_{SR,k}}{{\bf{\Gamma }}_{SR,k}}{\bf{{\dot{P}}}}_R^H{{\bf{C}}_{SR,k}}{{\bf{{\dot{P}}}}_R}} \right){{{\bf{\tilde C}}}_{SR,k}}
\end{aligned}
\end{equation}
It can be shown that
\begin{equation}\label{X5}
\begin{aligned}
  & \mathbb E\left[ {{{\left| {{\mathbf{{\underline{h}}}}_{SR,k}^H{{\mathbf{{\underline{C}}}}_{SR,k}}{{\mathbf{\Gamma }}_{SR,k}}{\mathbf{{\dot{P}}}}_R^H{{\mathbf{H}}_{SR,k}}{\mathbf{T}}_{SR,k}^{}\phi _{}^*} \right|}^2}} \right] \\
  & = {\nu _{S,k}}\tau E_T^2\frac{{{\bf{{\dot{p}}}}_{S,k}^H{\bf{\tilde C}}_{SR,k}^{}{\rm{diag}}\left( {{\bf{{\dot{p}}}}_{S,k}^H{{\bf{{\dot{p}}}}_{S,k}}} \right){\bf{\tilde C}}_{SR,k}^{}{{\bf{{\dot{p}}}}_{S,k}}}}{{\left( {{\rm{Tr}}\left( {{{{\bf{\hat {\underline{C}}}}}_{SR,k}}} \right)} \right)^{-2}}{{{\left( {{\lambda _1}\left( {{{{\bf{\tilde C}}}_{SR,k}}} \right)} \right)}^2}}}
\end{aligned}
\end{equation}
Using the similar approach on the remaining terms in (\ref{X3}), Lemma \ref{lm1} is obtained.
\end{proof}

\begin{lemma}\label{lm2}
Let ${{\bf{\underline{C}}}_{i,k}}$ and ${{\bf{\underline{\hat C}}}_{i,k}}$ ($i\in \{SR,RD\}$) be the covariance matrices of effective channel ${\bf {\underline{h}}}_{i,k}$ and its estimates ${\bf {\underline{\hat h}}}_{i,k}$ given by Theorem \ref{t3}. If ${\rm Tr}({{\bf{\underline{C}}}_{SR,k}})={\cal O}(N_R)$ and ${\rm Tr}({{\bf{\underline{C}}}_{RD,k}})={\cal O}(N_T)$, we have ${\rm Tr}({{\bf{\underline{\hat C}}}_{SR,k}})={\cal O}(N_R)$ and ${\rm Tr}({{\bf{\underline{\hat C}}}_{RD,k}})={\cal O}(N_T)$. Moreover, $\mathop {\lim }\limits_{{N_R} \to \infty } \frac{{{A_R}}}{{{N_R}}} > 0$ and $\mathop {\lim }\limits_{{N_T} \to \infty } \frac{{{A_T}}}{{{N_T}}} > 0$ (i.e., $A_R$ and $A_T$ scale with ${\cal O}(N_R)$ and ${\cal O}(N_T)$, respectively).
\end{lemma}

\begin{proof}
We present the proof for ${\rm Tr}({{\bf{\underline{\hat C}}}_{SR,k}})$ and $A_R$. Using definition of ${{{\mathbf{\hat {\underline{C}}}}}_{SR,k}}$ in Theorem \ref{t3} and eigenvalue decomposition ${{\mathbf{\underline{C}}}_{SR,k}} = {\mathbf{\underline{U}}}{\underline{\Sigma} _{SR,k}}{{\mathbf{\underline{U}}}^H}$, we have
\begin{equation}\label{X6}
\begin{aligned}
{\rm{Tr}}\left( {{{{\bf{\hat {\underline{C}}}}}_{SR,k}}} \right) & = {\rm{Tr}}\left( {{{\bf{{\underline{C}}}}_{SR,k}}{{\bf{\Gamma }}_{SR,k}}{{\bf{{\underline{C}}}}_{SR,k}}} \right)\\
& ={\rm{  Tr}}\left( {{\underline{\Sigma} _{SR,k}}{{\left( {{\omega _1}{\underline{\Sigma} _{SR,k}} + {\omega _2}{\bf I}_{A_R}} \right)}^{ - 1}}{\underline{\Sigma} _{SR,k}}} \right)\\
& = \sum\limits_{l = 1}^{{A_R}} {{{\left( {{\lambda _l}\left( {{{\bf{{\underline{C}}}}_{SR,k}}} \right)} \right)}^2}{{\left( {{\omega _1}{\lambda _l}\left( {{{\bf{{\underline{C}}}}_{SR,k}}} \right) + {\omega _2}} \right)}^{ - 1}}}
\end{aligned}
\end{equation}
where $\omega_1$ and $\omega_2$ are positive constants independent of $A_R$. From (\ref{X6}), ${\rm{Tr}}( {{{{\bf{\hat {\underline{C}}}}}_{SR,k}}} )$ is bounded as
\begin{equation}\label{X7}
\sum\limits_{l = 1}^{{A_R}} {\frac{{{{\left( {{\lambda _l}\left( {{{\bf{\underline{C}}}_{SR,k}}} \right)} \right)}^2}}}{{{\omega _1}{\lambda _1}\left( {{{\bf{\underline{C}}}_{SR,k}}} \right) + {\omega _2}}}}  \le {\rm{Tr}}\left( {{{{\bf{\hat {\underline{C}}}}}_{SR,k}}} \right) \le \frac{1}{{{\omega _1}}}{\rm{Tr}}\left( {{{\bf{\underline{C}}}_{SR,k}}} \right)
\end{equation}
In (\ref{X7}), the upper bound is an ${\cal O}(N_R)$ term. With the definition of ${{\bf{\underline{C}}}_{SR,k}}$ in the Theorem \ref{t3}, we have
\begin{equation}\label{X8}
\begin{aligned}
  {\lambda _1}\left( {{{\mathbf{\underline{C}}}_{SR,k}}} \right) & = {{{\beta _{SR,k}}{\lambda _1}\left( {{{{\mathbf{\tilde C}}}_{SR,k}}} \right)}}{\mathbf{u}}_1^H\left( {{{\mathbf{\underline{C}}}_{SR,k}}} \right)\\
  & \times {\mathbf{{\dot{P}}}}_R^H{{\mathbf{C}}_{SR,k}}{{\mathbf{{\dot{P}}}}_R}{{\mathbf{u}}_1}\left( {{{\mathbf{\underline{C}}}_{SR,k}}} \right) \\
  & \le {\beta _{SR,k}}{\lambda _1}\left( {{{{\mathbf{\tilde C}}}_{SR,k}}} \right){\lambda _1}\left( {{{\mathbf{C}}_{SR,k}}} \right) = {\cal O}\left( 1 \right)
  \end{aligned}
\end{equation}
The last step is based on assumption \textbf{A1}. Since ${\text{Tr}}\left( {{{\mathbf{\underline{C}}}_{SR,k}}} \right)={\cal O}(N_R)$, we can conclude that $\lambda _1\left( {{{\bf{\underline{C}}}_{SR,k}}} \right)={\cal O}(1)$ (otherwise ${\text{Tr}}\left( {{{\mathbf{\underline{C}}}_{SR,k}}} \right)\le A_R {\lambda _1}\left( {{{\mathbf{\underline{C}}}_{SR,k}}} \right) <{\cal O}(N_R)$). Thus, the lower bound in (\ref{X7}) also scales with ${\cal O}(N_R)$. Moreover, based on the definition of ${{\bf{\underline{C}}}_{SR,k}}$ in the Theorem \ref{t3}, we have
\begin{equation}\label{X9}
\begin{aligned}
  {\mathop{\rm Tr}\nolimits} \left( {{{\bf{\underline{C}}}_{SR,k}}} \right) & = {\beta _{SR,k}}{\lambda _1}\left( {{{{\bf{\tilde C}}}_{SR,k}}} \right)\sum\limits_{l = 1}^{{A_R}} {{\bf{{\dot{p}}}}_{R,l}^H} {{\bf{C}}_{SR,k}}{{\bf{{\dot{p}}}}_{R,l}} \\
  & \le {A_R}{\beta _{SR,k}}{\lambda _1}\left( {{{{\bf{\tilde C}}}_{SR,k}}} \right){\lambda _1}\left( {{{\bf{C}}_{SR,k}}} \right)
  \end{aligned}
\end{equation}
where ${{\bf{{\dot{p}}}}_{R,l}}$ is the $l$th column of ${\bf {\dot{P}}}_R$. From (\ref{X9}), $A_R$ must scale ${\cal O}(N_R)$. Otherwise, ${\text{Tr}}\left( {{{\mathbf{\underline{C}}}_{SR,k}}} \right)<{\cal O}(N_R)$ based on assumption \textbf{A1}.

\end{proof}

\subsection*{B. Proof of Theorem \ref{t1}}
To facilitate analysis, we first derive the covariance matrices of distortion noises ${\bf t}_S[u]$ and ${\bf r}_R[u]$. Based on the model (\ref{2-3}), we have ${{\bf\Theta } _S^{\textsf{T}}}= {\text{diag}}\left( {{\nu _{S,1}}{E_{S,1}}, \cdots ,{\nu _{S,K}}{E_{S,K}}} \right)$. Moreover, define ${{\mathbf{x}}_S}\left[ u \right] = {\left[ {{\mathbf{x}}_{S,1}^T\left[ u \right], \cdots ,{\mathbf{x}}_{S,K}^T\left[ u \right]} \right]^T}$, ${\mathbf{\Theta }}_R^\textsf{R}$ can be expressed based on model (\ref{2-4}) as
\begin{equation}\label{A1}
\begin{aligned}
  {\mathbf{\Theta }}_R^\textsf{R} & = {\mu _R}{\text{diag}}\left( {\mathbb{E}\left[ {{{\mathbf{H}}_{SR}}{{\mathbf{\Omega }}_S}{\mathbf{H}}_{SR}^H} \right]} \right) \\
  & + {\mu _R}{\text{diag}}\left( {\mathbb{E}\left[ {{{\mathbf{H}}_{EI}}{{\mathbf{\Omega }}_R}{\mathbf{H}}_{EI}^H} \right]} \right) + {\mu _R}{{\mathbf{I}}_{{N_R}}}
\end{aligned}
\end{equation}
where
\begin{equation}\label{A1-1}
\begin{array}{ll}
{{\bf{\Omega }}_S} = \mathbb{E} \left[ {\left( {{\bf{x}}_S\left[ u \right] + {{\bf{t}}_S}\left[ u \right]} \right){{\left( {{\bf{s}}\left[ u \right] + {{\bf{t}}_S}\left[ u \right]} \right)}^H}} \right] \\
\hspace{1.7em} = {\rm{ diag}}\left( {\left( {1 + {\nu _{S,1}}} \right){E_{S,1}}, \cdots ,\left( {1 + {\nu _{S,K}}} \right){E_{S,K}}} \right)\\
{{\bf{\Omega }}_R} = \mathbb{E}\left[ {\left( {{{\bf{x}}_R}\left[ u \right] + {{\bf{t}}_R}\left[ u \right]} \right){{\left( {{{\bf{x}}_R}\left[ u \right] + {{\bf{t}}_R}\left[ u \right]} \right)}^H}} \right] \\
\hspace{1.7em} = \mathbb{E}\left[ {{{\mathbf{x}}_R}\left[ u \right]{\mathbf{x}}_R^H\left[ u \right]} \right] + {\mathbf{\Theta }}_R^\textsf{T}
\end{array}
\end{equation}
For further analysis, we approximate $\mathbb{E}\left[ {{{\mathbf{x}}_R}\left[ u \right]{\mathbf{x}}_R^H\left[ u \right]} \right]$ in (\ref{A1-1}) with ${\rm diag}(\mathbb{E}\left[ {{{\mathbf{x}}_R}\left[ u \right]{\mathbf{x}}_R^H\left[ u \right]} \right] )= \frac{1}{\nu_R} {\mathbf{\Theta }}_R^\textsf{T}$. Note that the approximation results in a new upper bound on $\textsf{R}_{SR,k}^{{\text{Upper}}}$ (The new bound will also be referred to as $\textsf{R}_{SR,k}^{{\text{Upper}}}$ for convenience). Straight-forward computations yield to
\begin{equation}\label{A2}
\begin{aligned}
\mathbb E\left[ {{{\mathbf{H}}_{SR}}{{\mathbf{\Omega }}_S}{\mathbf{H}}_{SR}^H} \right] & = \sum_{l = 1}^K {{\nu _{S,l}}{E_{S,l}}\mathbb E\left[ {{{\mathbf{h}}_{SR,l}}{\mathbf{h}}_{SR,l}^H} \right]}  \\
& = \sum_{l = 1}^K  {{(1+\nu _{S,l})}{E_{S,l}}\beta_{SR,l}{{\mathbf{C}}_{SR,l}}}
\end{aligned}
\end{equation}
Following the spatial correlation model of ${{{\mathbf{H}}_{EI}}}$ in section III and the approximation under (\ref{A1-1}), the expression of ${\mathbb E\left[ {{{\mathbf{H}}_{EI}}{{\mathbf{\Omega }}_R}{\mathbf{H}}_{EI}^H} \right]}$ can be obtained as in (\ref{A3}), where the third step follows from the formula ${\text{vec}}\left( {{\mathbf{ABc}}} \right) = \left( {{{\mathbf{c}}^T} \otimes {\mathbf{A}}} \right){\text{vec}}\left( {\mathbf{B}} \right)$ and the fourth step is due to the fact that the elements of ${\bf X}_{EI}$ are i.i.d with distribution ${\cal CN}(0,1)$. Inserting (\ref{A2}), (\ref{A3}) into (\ref{A1}), the expression of ${\mathbf{\Theta }}_R^\textsf{R}$ in (\ref{3-3}) is obtained.

\begin{figure*}
\normalsize
\setcounter{mytempeqncnt}{\value{equation}}
\setcounter{equation}{54}
\begin{equation}\label{A3}
\begin{aligned}
   \mathbb E\left[ {{{\mathbf{H}}_{EI}}{{\mathbf{\Omega }}_R}{\mathbf{H}}_{EI}^H} \right] & \approx \left(1+\frac{1}{\nu_R}\right){\beta _{EI}} \mathbb E\left[ {{\mathbf{C}}_{EI}^{1/2}{{\mathbf{X}}_{EI}}{\mathbf{\tilde C}}_{EI}^{1/2}{{\mathbf{\Theta }}_R^\textsf{T}}{\mathbf{\tilde C}}_{EI}^{1/2}{\mathbf{X}}_{EI}^H{\mathbf{C}}_{EI}^{1/2}} \right] \hfill \\
  & = \left(1+\frac{1}{\nu_R}\right){\beta _{EI}}\sum\limits_{l = 1}^{{N_T}} {\mathbb{E}\left[ {{\mathbf{C}}_{EI}^{1/2}{{\mathbf{X}}_{EI}}{\mathbf{\tilde C}}_{EI}^{1/2}{{\left( {{\mathbf{\Theta }}_R^\textsf{T}} \right)}^{1/2}}{{\mathbf{e}}_l}{\mathbf{e}}_l^T{{\left( {{\mathbf{\Theta }}_R^\textsf{T}} \right)}^{1/2}}{\mathbf{\tilde C}}_{EI}^{1/2}{\mathbf{X}}_{EI}^H{\mathbf{C}}_{EI}^{1/2}} \right]}  \hfill \\
  & = \left(1+\frac{1}{\nu_R}\right){\beta _{EI}}\sum\limits_{l = 1}^{{N_T}} {\mathbb{E}\left[ {{{\left\| {{{\left( {{\mathbf{\tilde C}}_{EI}^{1/2}{{\left( {{\mathbf{\Theta }}_R^\textsf{T}} \right)}^{1/2}}{{\mathbf{e}}_l}} \right)}^T} \otimes {\mathbf{C}}_{EI}^{1/2}{\text{vec}}\left( {{{\mathbf{X}}_{EI}}} \right)} \right\|}^2}} \right]}  \hfill \\
  & = \left(1+\frac{1}{\nu_R}\right){\beta _{EI}}\sum\limits_{l = 1}^{{N_T}} {{{\left\| {{{\left( {{\mathbf{\tilde C}}_{EI}^{1/2}{{\left( {{\mathbf{\Theta }}_R^\textsf{T}} \right)}^{1/2}}{{\mathbf{e}}_l}} \right)}^T} \otimes {\mathbf{C}}_{EI}^{1/2}} \right\|}^2}}\hfill = \left(1+\frac{1}{\nu_R}\right){\beta _{EI}}{\text{Tr}}\left( {{{{\mathbf{\tilde C}}}_{EI}}{\mathbf{\Theta }}_R^\textsf{T}} \right){\mathbf{C}}_{EI}  \hfill \\
\end{aligned}
\end{equation}
\setcounter{equation}{\value{mytempeqncnt}}
\hrule
\end{figure*}

Then we derive the asymptotic bound in Theorem \ref{t1}. By replacing ${{\mathbf{H}}_{EI}}{{\bf\Theta } _R^{\textsf{T}}}{\mathbf{H}}_{EI}^H$ in ${\bf Q}_k$ with ${\rm diag}\left( {{\mathbf{H}}_{EI}}{{\bf\Theta } _R^{\textsf{T}}}{\mathbf{H}}_{EI}^H \right)$, we approximate ${\bf Q}_k$ as ${\mathbf{Q}}_k = {{\mathbf{H}}_{SR}}{{\bf\Theta } _S^{\textsf{T}}}{\mathbf{H}}_{SR}^H - {\nu _{S,k}}{E_{S,k}}{{\mathbf{h}}_{SR,k}}{\mathbf{h}}_{SR,k}^H
+ {\rm diag}\left( {{\mathbf{H}}_{EI}}{{\bf\Theta } _R^{\textsf{T}}}{\mathbf{H}}_{EI}^H \right) + {\mathbf{\Theta }}_R^\textsf{R} + {{\mathbf{I}}_{{N_R}}}$. The approximation results in a new upper bound on $\textsf{R}_{SR,k}^{{\text{Upper}}}$, since it reduces the power of residual EI after combining by ${\bf w}_{R,k}$ (The new bound will also be referred to as $\textsf{R}_{SR,k}^{{\text{Upper}}}$ for convenience). Based on [17, Lemma 4 (ii)], we can replace ${\text{diag}}\left( {{{\mathbf{H}}_{EI}}{\mathbf{\Theta }}_R^\textsf{T}{\mathbf{H}}_{EI}^H} \right)$ with its deterministic equivalence in the large-$N_T$ regime $\frac{1}{N_T}{\text{diag}}\left( {{{\mathbf{H}}_{EI}}{\mathbf{\Theta }}_R^\textsf{T}{\mathbf{H}}_{EI}^H} \right) \mathop  = \frac{1}{N_T} \beta_{EI}{\text{Tr}}( {{{{\mathbf{\tilde C}}}_{EI}}{\mathbf{\Theta }}_R^\textsf{T}} ){{\mathbf{I}}_{{N_R}}}$. Finally, the expression of $\textsf{R}_{SR,k}^{{\text{Upper}}}$ in Theorem \ref{t1} can be obtained by first applying [17, Lemma 4 (ii)] on ${{\mathbf{h}}_{SR}^H{\mathbf{Q}}_k^{ - 1}{{\mathbf{h}}_{SR,k}}}$, and then using [17, Theorem 1].

\subsection*{C. Proof of Theorem \ref{t2}}

Since the scenario of interest is the large-$(N_R,N_T)$ regime, we need only consider the effect of ${\bf w}_{T,k}$ on $ \textsf{R}_{SR,k}^{{\text{Upper}}}$ (given by Theorem \ref{t1}), more precisely, the term ${{\rm{Tr}}( {{{{\bf{\tilde C}}}_{EI}}{\bf{\Theta }}_R^\textsf{T}} )}$. With the assumption \textbf{A2}, we have ${\rm diag}({{{\bf{\tilde C}}}_{EI}})={\bf I}_{N_T}$. Thus, using the model (\ref{2-3}), we have ${{\rm{Tr}}( {{{{\bf{\tilde C}}}_{EI}}{\bf{\Theta }}_R^\textsf{T}} )} = {{\text{Tr}}( {{\mathbf{\Theta }}_R^\textsf{T}} )}
={\nu _R}{\text{Tr}}\left( {{\text{diag}}\left( {\mathbb E\left[ {{{\mathbf{x}}_R[u]}{\mathbf{x}}_R^H[u]} \right]} \right)} \right) = {\nu _R}\sum_{l=1}^K E_{R,l}$, which is independent of ${\bf w}_{T,k}$.

Then we derive the upper bound in Theorem \ref{t2}. Using the model (\ref{2-4}), the power of received distortion at the relay can be derived as $\mathbb E\left[ {{{\left\| {{{\mathbf{r}}_{D,k}}\left[ u \right]} \right\|}^2}} \right] = {\mu _{D,k}}\left( {{N_T}{\beta _{RD,k}}{E_{R,k}} + {\nu _R}{N_T}{\beta _{RD,k}}\sum_{l = 1}^K {{E_{R,l}}}  + 1} \right)$. Moreover, in the large $N_T$ regime, the term ${{\mathbf{h}}_{RD,k}^H{\mathbf{\Theta }}_R^\textsf{T}{{\mathbf{h}}_{RD,k}}}$ in (\ref{3-0-2}) approaches to ${{\mathbf{h}}_{RD,k}^H{\mathbf{\Theta }}_R^\textsf{T}{{\mathbf{h}}_{RD,k}}} = {\beta _{RD,k}}{\text{Tr}}\left( {{{\mathbf{C}}_{RD,k}}{\mathbf{\Theta }}_R^\textsf{T}} \right) = {\nu _R}{\beta _{RD,k}}\sum_{l=1}^K E_{R,l}$ [17, Lemma 4 (ii)]. Therefore, it is sufficient to design ${\bf w}_{T,k}$ to maximize the numerator of (\ref{3-0-2}), which is exactly the eigen BF scheme. The resultant upper bound on achievable rate can be obtained by applying [17, Lemma 4 (ii)] and assumption \textbf{A2} on the numerator of (\ref{3-0-2}).

\subsection*{D. Proof of Theorem \ref{t3}}

With the expression of LMMSE estimator \cite{IEEEhowto:KotechaTSP04}, we have ${{{\bf{\hat {\underline{h}}}}}_{SR,k}} = \mathbb{E}[ {{{\bf{{\underline{h}}}}_{SR,k}}{\bf{ \tilde{z}}}_{SR,k}^H} ] {( {\mathbb{E}[ {{{{\bf{\tilde{ z}}}}_{SR,k}}{\bf{ \tilde{z}}}_{SR,k}^H} ]} )^{ - 1}}{{{\bf{ \tilde{z}}}}_{SR,k}}$. Using the independence between ${{\bf{{\underline{h}}}}_{SR,k}}$ and distortion noises, we have
\setcounter{equation}{55}
\begin{equation}\label{C1}
\begin{aligned}
\mathbb E\left[ {{{\bf{\underline{h}}}_{SR,k}}{\bf{\tilde{z}}}_{SR,k}^H} \right] &= \mathbb E\left[ {{\bf{{\dot{P}}}}_R^H{{\bf{H}}_{SR,k}}{{\bf{{\dot{p}}}}_{S,k}}{\bf{{\dot{p}}}}_{S,k}^H{\bf{H}}_{SR,k}^H{{\bf{{\dot{P}}}}_R}} \right]\\
&= {\beta _{SR,k}}{\lambda _1}\left( {{{{\bf{\tilde C}}}_{SR,k}}} \right){\bf{{\dot{P}}}}_R^H{{\bf{C}}_{SR,k}}{{\bf{{\dot{P}}}}_R} \\
 & \buildrel \Delta \over = {{\bf{\underline{C}}}_{SR,k}}
\end{aligned}
\end{equation}
where the second step follows from a similar derivation with that in (\ref{A3}). Moreover, ${\mathbb{E}\left[ {{{{\bf{ \tilde{z}}}}_{SR,k}}{\bf{ \tilde{z}}}_{SR,k}^H} \right]}$ can be expressed as
\begin{equation}\label{C2}
\begin{aligned}
&  \mathbb E   \left[ {{{{\bf{\tilde z}}}_{SR,k}}{\bf{\tilde z}}_{SR,k}^H} \right]= {{\bf{\underline{C}}}_{SR,k}}\\
& +{\left( {\frac{1}{{\tau {E_T}}}} \right)^2}\mathbb{E}\left[ {{\mathbf{{\dot{P}}}}_R^H{{\mathbf{H}}_{SR,k}}{{\mathbf{T}}_{SR,k}}{\phi ^*}{\phi ^T}{\mathbf{T}}_{SR,k}^H{\mathbf{H}}_{SR,k}^H{{\mathbf{{\dot{P}}}}_R}} \right]\\
&  + {\left( {\frac{1}{{\tau {E_T}}}} \right)^2}\mathbb{E}\left[ {{\mathbf{{\dot{P}}}}_R^H{{\mathbf{R}}_{SR,k}}{\phi ^*}{\phi ^T}{\mathbf{R}}_{SR,k}^H{{\mathbf{{\dot{P}}}}_R}} \right] \\
& + {\left( {\frac{1}{{\tau {E_T}}}} \right)^2}\mathbb{E}\left[ {{\mathbf{{\dot{P}}}}_R^H{{\mathbf{N}}_{SR,k}}{\phi ^*}{\phi ^T}{\mathbf{N}}_{SR,k}^H{{\mathbf{{\dot{P}}}}_R}} \right]
\end{aligned}
\end{equation}
According to the independence between ${\bf H}_{SR,k}$ and ${\bf T}_{SR,k}$, the second term of right-hand side of (\ref{C2}) can be rewritten as
\begin{equation*}\label{C3}
 \begin{aligned}
& \mathbb{E}\left[ {{\mathbf{{\dot{P}}}}_R^H{{\mathbf{H}}_{SR,k}}{{\mathbf{T}}_{SR,k}}{\phi ^*}{\phi ^T}{\mathbf{T}}_{SR,k}^H{\mathbf{H}}_{SR,k}^H{{\mathbf{{\dot{P}}}}_R}} \right] \\
& ={\nu _{S,k}}\tau E_T^2\mathbb{E}\left[ {{\mathbf{{\dot{P}}}}_R^H{{\mathbf{H}}_{SR,k}}{\text{diag}}\left( {{{\mathbf{{\dot{p}}}}_{S,k}}{\mathbf{{\dot{p}}}}_{S,k}^H} \right){\mathbf{H}}_{SR,k}^H{{\mathbf{{\dot{P}}}}_R}} \right]\\
&  = {\nu _{S,k}}\tau E_T^2 \beta_{SR,k} {\rm{Tr}}\left( {{{{\bf{\tilde C}}}_{SR,k}}{\rm{diag}}\left( {{{\bf{{\dot{p}}}}_{S,k}}{\bf{{\dot{p}}}}_{S,k}^H} \right)} \right){\bf{{\dot{P}}}}_R^H{{\bf{C}}_{SR,k}}{{\bf{{\dot{P}}}}_R}
\end{aligned}
\end{equation*}
Using the similar approach on the remaining terms of (\ref{C2}), and substituting the resultant expression of (\ref{C2}) and (\ref{C1}) into the expression of LMMSE estimator, the result in Theorem \ref{t3} is obtained. Using (\ref{4-7})-(\ref{4-8}), the second order statistics of $ {\bf {\underline{\hat{h}}}}_{{SR},k} $ and $\Delta {\bf {\underline{h}}}_{{SR},k}$ in Theorem \ref{t3} can be easily verified.

\subsection*{E. Proof of Theorem \ref{t4}}

We show the proof for $\mathbb E[{\textsf{EI}}_k]$ and $\mathbb E[\textsf{D}_{R,k}^\textsf{T}]$ due to the space limitation. The proof for other terms in Theorem \ref{t4} is similar.

Note that ${\left[ {{\bf{\hat {\underline{H}}}}_{i}^H{{{\bf{\hat {\underline{H}}}}}_{i}}} \right]_{l,j}} = {\bf{\hat {\underline{h}}}}_{i,l}^H{{{\bf{\hat {\underline{h}}}}}_{i,j}} = {\mathbf{x}}_{i,l}^H{\mathbf{\hat {\underline{C}}}}_{i,l}^{1/2}{\mathbf{\hat {\underline{C}}}}_{i,j}^{1/2}{{\mathbf{x}}_{i,j}}$ ($i\in \{SR,RD\}$). Based on Lemma \ref{lm2}, the dimensions of ${{\mathbf{x}}_{SR,j}}$ and ${{\mathbf{x}}_{RD,j}}$ (i.e., $A_R$ and $A_T$) approach to infinity as $\min\{N_R,N_T\} \to \infty$. Thus, according to [17, Lemma 4 (ii) and (iii)], in the large-($N_R,N_T$) regime, we can deduce
\begin{equation}\label{D1}
\begin{aligned}
& \frac{1}{{{N_R}}}{\bf{\hat {\underline{H}}}}_{SR}^H{{{\bf{\hat {\underline{H}}}}}_{SR}} \\
& \hspace{3em} \to \frac{1}{{{N_R}}}{\rm{diag}}\left( {{\rm{Tr}}\left( {{{{\bf{\hat {\underline{C}}}}}_{SR,1}}} \right), \cdots ,{\rm{Tr}}\left( {{{{\bf{\hat {\underline{C}}}}}_{SR,K}}} \right)} \right)\\
& \frac{1}{{{N_T}}}{\bf{\hat {\underline{H}}}}_{RD}^H{{{\bf{\hat {\underline{H}}}}}_{RD}}\\
& \hspace{3em} \to \frac{1}{{{N_T}}}{\rm{diag}}\left( {{\rm{Tr}}\left( {{{{\bf{\hat {\underline{C}}}}}_{RD,1}}} \right), \cdots ,{\rm{Tr}}\left( {{{{\bf{\hat {\underline{C}}}}}_{RD,K}}} \right)} \right)
\end{aligned}
\end{equation}

\emph{1) Derivation of $\mathbb E[\textsf{EI}_k]$:}

\emph{(a) Asymptotic Expression:} Based on the model (\ref{2-3}), we have ${\mathbb E\left[ {{{\mathbf{t}}_R}\left[ u \right]{\mathbf{t}}_R^H\left[ u \right]} \right]} = {\nu _R}{\rm{diag}}\left( {\mathbb E\left[ {{{\bf{W}}_T}{\Lambda _R}{\bf{W}}_T^H} \right]} \right)$. Since ${{\bf{H}}_{EI}}$ and ${{\bf{\underline{\hat{H}}}}_{RD}}$ are independent, $\textsf{EI}_k$ can be rewritten as $\mathbb E \left[\textsf{EI}_k\right] = \mathbb E [{\bf{w}}_{R,k}^H{{\bf{H}}_{EI}}{{\bf{\Omega }}_R}{\bf{H}}_{EI}^H{{\bf{w}}_{R,k}} ]
 = \mathbb E[ {{{\| {{\bf{e}}_k^T{{( {{\bf{\hat {\underline{H}}}}_{SR}^H{{{\bf{\hat {\underline{H}}}}}_{SR}}} )}^{ - 1}}{{{\bf{\hat {\underline{H}}}}}_{SR}}{{\bf{{\dot{P}}}}_R^H}{{\bf{H}}_{EI}}{\bf{\Omega }}_R^{1/2}} \|^2}}} ]$, where ${{\bf{\Omega }}_R}$ is defined as ${{\bf{\Omega }}_R} = \mathbb{E}[ {\left( {{{\bf{x}}_R}\left[ u \right] + {{\bf{t}}_R}\left[ u \right]} \right){{\left( {{{\bf{x}}_R}\left[ u \right] + {{\bf{t}}_R}\left[ u \right]} \right)}^H}} ] = \mathbb E\left[ {{{\bf{W}}_T}{\Lambda _R}{\bf{W}}_T^H} \right] + {\nu _R}{\rm{diag}}\left( {\mathbb E\left[ {{{\bf{W}}_T}{\Lambda _R}{\bf{W}}_T^H} \right]} \right)$. Inserting (\ref{4-12-1}) into the expression of ${{\bf{\Omega }}_R}$ and using (\ref{D1}), we can obtain
\begin{equation}\label{D6}
  {{\bf{\Omega }}_R} = \sum\limits_{l = 1}^K {E_{R,l}}\frac {{{\bf{{\dot{P}}}}_T}{{{\bf{\hat {\underline{C}}}}}_{RD,l}}{\bf{{\dot{P}}}}_T^H + {\nu _R}{\rm{diag}}\left( {{{\bf{{\dot{P}}}}_T}{{{\bf{\hat {\underline{C}}}}}_{RD,l}}{\bf{{\dot{P}}}}_T^H} \right)}{{ {{\rm{Tr}}\left( {{{{\bf{\hat {\underline{C}}}}}_{RD,l}}} \right)} }}
\end{equation}
In the large-$(N_R,N_T)$ regime, using (\ref{D1}), we have
\begin{equation*}\label{D7}
\begin{aligned}
\mathbb E \left[\textsf{EI}_k\right] &  = \frac{\mathbb E\left[ {{{{\bf{\hat {\underline{h}}}}}_{SR,k}}{{\bf{{\dot{P}}}}_R^H}{{\bf{H}}_{EI}}{{\bf{\Omega }}_R}{\bf{H}}_{EI}^H{{\bf{{\dot{P}}}}_R}{\bf{\hat {\underline{h}}}}_{SR,k}^H} \right]}{\left( {{\rm{Tr}}\left( {{{{\bf{\hat {\underline{C}}}}}_{SR,k}}} \right)} \right)^{ 2}}  \\
& = \frac{\mathbb E\left[ {{\rm{Tr}}\left( {{\bf{\hat {\underline{C}}}}_{SR,k}^{1/2}{{\bf{{\dot{P}}}}_R^H}{{\bf{H}}_{EI}}{{\bf{\Omega }}_R}{\bf{H}}_{EI}^H{{\bf{{\dot{P}}}}_R}{\bf{\hat {\underline{C}}}}_{SR,k}^{1/2}} \right)} \right]}{\left( {{\rm{Tr}}\left( {{{{\bf{\hat {\underline{C}}}}}_{SR,k}}} \right)} \right)^{  2}}
\end{aligned}
\end{equation*}
The second step is based on the property ${\rm Tr}(\bf AB) ={\rm Tr}(\bf BA)$ and independence between ${{{\bf{\hat {\underline{h}}}}}_{SR,k}}$ and ${{\bf{H}}_{EI}}$. Finally, using a similar derivation with (\ref{A3}), $\mathbb E \left[\textsf{EI}_k\right]$ in Theorem \ref{t4} is obtained.

\emph{(b) Scaling Behavior:} Substituting (\ref{D6}) into the expression of $\mathbb E \left[\textsf{EI}_k\right]$ in Theorem \ref{t4} and neglecting the terms that have no effect on the scaling behavior, we have
\begin{equation*}
\begin{aligned}
\mathbb E\left[ {\textsf{EI}{_k}} \right] = {\beta _{EI}}{\left( {{\text{Tr}}\left( {{{{\mathbf{\hat {\underline{C}}}}}_{SR,k}}} \right)} \right)^{ - 2}}\sum\limits_{l = 1}^K {{E_{R,l}}{{\left( {{\text{Tr}}\left( {{{{\mathbf{\hat {\underline{C}}}}}_{RD,l}}} \right)} \right)}^{ - 1}}}   \\
   \times {\text{Tr}}\left( {{{{\mathbf{\hat {\underline{C}}}}}_{RD,l}}{\mathbf{{\dot{P}}}}_T^H{{{\mathbf{\tilde C}}}_{EI}}{{\mathbf{{\dot{P}}}}_T}} \right){\text{Tr}}\left( {{{{\mathbf{\hat {\underline{C}}}}}_{SR,k}}{\mathbf{{\dot{P}}}}_R^H{{\mathbf{C}}_{EI}}{{\mathbf{{\dot{P}}}}_R}} \right) \hfill \\
   \le {\beta _{EI}}\sum\limits_{l = 1}^K \frac{{{E_{R,l}}{\lambda _{{N_R} - {A_R} + 1}}\left( {{{\mathbf{C}}_{EI}}} \right)} {\lambda _{{N_T} - {A_T} + 1}}\left( {{{{\mathbf{\tilde C}}}_{EI}}} \right)}{ {{\text{Tr}}\left( {{{{\mathbf{\hat {\underline{C}}}}}_{SR,k}}} \right)} }
\end{aligned}
\end{equation*}
where the second step follows from the definition of ${\bf {\dot{P}}}_R$ and ${\bf {\dot{P}}}_T$. Based on Lemma \ref{lm2}, ${( {{\text{Tr}}( {{{{\mathbf{\hat {\underline{C}}}}}_{SR,k}}} )} )^{ - 1}}={\cal O}(N_R^{-1})$. Thus we can conclude that $\mathbb E \left[\textsf{EI}_k\right] \le {\cal O}({\lambda _{{N_R} - {A_R} + 1}}\left( {{{\mathbf{C}}_{EI}}} \right) {\lambda _{{N_T} - {A_T} + 1}}( {{{{\mathbf{\tilde C}}}_{EI}}} )KN_R^{ - 1})$.

\emph{2) Derivation of $\mathbb E[\textsf{D}_{R,k}^\textsf{T}]$:}

\emph{(a) Asymptotic Expression:} Based on the model (\ref{2-3}), the covariance matrix of ${\bf t}_{S,j}$ can be derived as ${\mathbb E\left[ {{{\mathbf{t}}_{S,j}}\left[ u \right]{\mathbf{t}}_{S,j}^H\left[ u \right]} \right]} ={\nu _{S,j}}{E_{S,j}}{\rm{diag}}\left( {{{\bf{{\dot{p}}}}_{S,j}}{\bf{{\dot{p}}}}_{S,j}^H} \right)$. Thus, in the large-$(N_R,N_T)$ regime, using (\ref{D1}) $\mathbb E[\textsf{D}_{R,k}^\textsf{T}]$ can be expressed as
\begin{equation}\label{D9}
\begin{aligned}
  & \mathbb E[\textsf{D}_{R,k}^\textsf{T}] = {\left( {{\rm{Tr}}\left( {{{{\bf{\hat {\underline{C}}}}}_{SR,k}}} \right)} \right)^{ - 2}}{\nu _{S,j}}{E_{S,k}}\\
  & \times\underbrace {\mathbb E\left[ {{\bf{\hat {\underline{h}}}}_{SR,k}^H{\bf{{\dot{P}}}}_R^H{{\bf{H}}_{SR,k}}{\rm{diag}}\left( {{{\bf{{\dot{p}}}}_{S,k}}{\bf{{\dot{p}}}}_{S,k}^H} \right){\bf{H}}_{SR,k}^H{{\bf{{\dot{P}}}}_R}{{{\bf{\hat {\underline{h}}}}}_{SR,k}}} \right]}_{{g_{kk}}} \\
   &+ {\left( {{\rm{Tr}}\left( {{{{\bf{\hat {\underline{C}}}}}_{SR,k}}} \right)} \right)^{ - 2}}\sum\limits_{j = 1,j \ne k}^K {\nu _{S,j}}{E_{S,j}} \\
  & \times\underbrace {\mathbb E\left[ {{\bf{\hat {\underline{h}}}}_{SR,k}^H{\bf{{\dot{P}}}}_R^H{{\bf{H}}_{SR,j}}{\rm{diag}}\left( {{{\bf{{\dot{p}}}}_{S,j}}{\bf{{\dot{p}}}}_{S,j}^H} \right){\bf{H}}_{SR,j}^H{{\bf{{\dot{P}}}}_R}{{{\bf{\hat {\underline{h}}}}}_{SR,k}}} \right]}_{{g_{kj}}}
\end{aligned}
\end{equation}
According to the independence between ${{{{\bf{\hat {\underline{h}}}}}_{SR,k}}}$ and ${{{\bf{H}}_{SR,j}}}$ when $k\ne j$, ${{g_{kj}}}$ can be derived as
\begin{equation}\label{D10}
g_{kj} = {\beta _{SR,j}}{\rm{Tr}}\left( {{{{\bf{\hat {\underline{C}}}}}_{SR,k}}{\bf{{\dot{P}}}}_R^H{{\bf{C}}_{SR,j}}{{\bf{{\dot{P}}}}_R}} \right)
\end{equation}
Moreover, substituting (\ref{4-7}) and (\ref{4-8}) in Theorem \ref{t3} into (\ref{D9}) and using the large-$N_R$ approximation in (\ref{X5}), it is straightforward but tedious to show that $g_{kk}$ can be derived as
\begin{equation}\label{D12}
  \begin{aligned}
    & g_{kk} = {\left( {{\rm{Tr}}\left( {{{{\bf{\underline{\hat C}}}}_{SR,k}}} \right)} \right)^2}{\bf{{\dot{p}}}}_{S,k}^H{\rm{diag}}\left( {{{\bf{{\dot{p}}}}_{S,k}}{\bf{{\dot{p}}}}_{S,k}^H} \right){{\bf{{\dot{p}}}}_{S,k}} \\
    & + \frac{{{\nu _{S,k}}}}{\tau }{\left( {{\rm{Tr}}\left( {{{{\bf{\underline{\hat C}}}}_{SR,k}}} \right)} \right)^2} \frac{{{\rm{Tr}}\left( {{{\left( {{\rm{diag}}\left( {{{\bf{{\dot{p}}}}_{S,k}}{\bf{{\dot{p}}}}_{S,k}^H} \right){{{\bf{\tilde C}}}_{SR,k}}} \right)}^2}} \right)}}{{{{\left( {{\lambda _1}\left( {{{{\bf{\tilde C}}}_{SR,k}}} \right)} \right)}^2}}} \\
    & + \frac{{{\mu _R}}}{\tau }\frac{{\left( {\frac{1}{{{E_T}}}{\rm{ + }}{\beta _{SR,k}}{\text{Tr}}\left( {{{{\mathbf{\tilde C}}}_{SR,k}}{{\mathbf{\Omega }}_{S,k}}} \right)} \right){\rm{Tr}}\left( {{\bf{\underline{C}}}_{SR,k}^3{\bf{\Gamma }}_{SR,k}^2} \right)}}{{{\lambda _1}\left( {{{{\bf{\tilde C}}}_{SR,k}}} \right)}} \\
    & + \frac{{1}}{{\tau {E_T}}}\frac{{{\rm{Tr}}\left( {{\bf{\underline{C}}}_{SR,k}^3{\bf{\Gamma }}_{SR,k}^2} \right)}}{{{\lambda _1}\left( {{{{\bf{\tilde C}}}_{SR,k}}} \right)}}
  \end{aligned}
\end{equation}
where ${\bf\Omega}_{S,k} = {{{\bf{{\dot{p}}}}_{S,k}}{\bf{{\dot{p}}}}_{S,k}^H + {\nu _{S,k}}{\rm{diag}}\left( {{{\bf{{\dot{p}}}}_{S,k}}{\bf{{\dot{p}}}}_{S,k}^H} \right)} $. Inserting (\ref{D10}) and (\ref{D12}) into (\ref{D9}), the expression of $\mathbb E[\textsf{D}_{R,k}^\textsf{T}]$ is obtained.

\emph{(b) Scaling Behavior:} Using the similar approach as that for $\mathbb E[\textsf{EI}_k]$, it can be shown that, as $N_S$ is fixed, the first and second terms on the right-hand side of (\ref{D9}) scale as ${\cal O}(1)$ and ${\cal O}(KN_R^{-1})$, respectively. Thus, $\mathbb E[\textsf{D}_{R,k}^\textsf{T}] ={\cal O}(1)$ since $K \le N_R$. As $N_S$ scales with $N_R$, it is shown in \cite{IEEEhowto:AdhikaryIT13} that the eigenvectors of ${\bf \tilde C}_{EI}$ form a unitary DFT matrix in the large $N_S$ regime. In this case, we have ${\rm diag}({\bf {\dot{p}}}_{S,k}{\bf {\dot{p}}}_{S,k}^H)=1/N_S$. Applying this property and Lemma \ref{lm2} on (\ref{D12}), it can be verified that the first term on the right-hand side of (\ref{D9}) scales as ${\cal O}(1/N_S)$ as $N_S$ scales with $N_R$. This completes the proof.

%

\begin{IEEEbiography}[{\includegraphics[width=1in,height=1.25in,clip,keepaspectratio]{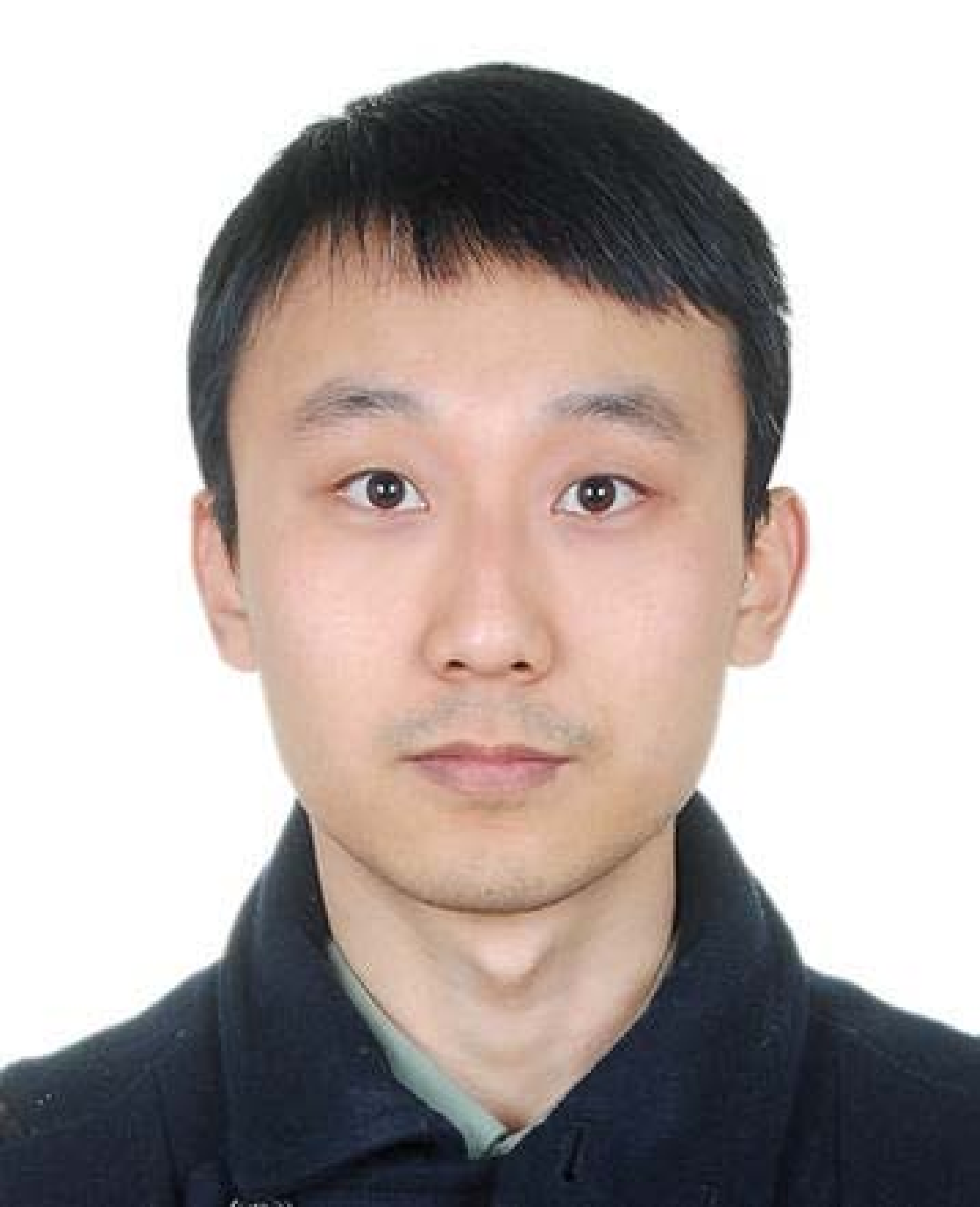}}]{Xiaochen Xia}
received his B.E. degree in electronic science and technology from Tianjin University (TJU) in 2010, and M.S. degree in communication and information system from PLA University of Science and Technology (PLAUST), in 2013. He is currently working toward the Ph.D. degree in institution of communications engineering, PLAUST. His research interests include relaying network, full-duplex communication, Network coding, MIMO techniques. He received the 2013 excellent master degree dissertation award of Jiangsu Province, China.
\end{IEEEbiography}

\begin{IEEEbiography}[{\includegraphics[width=1in,height=1.25in,clip,keepaspectratio]{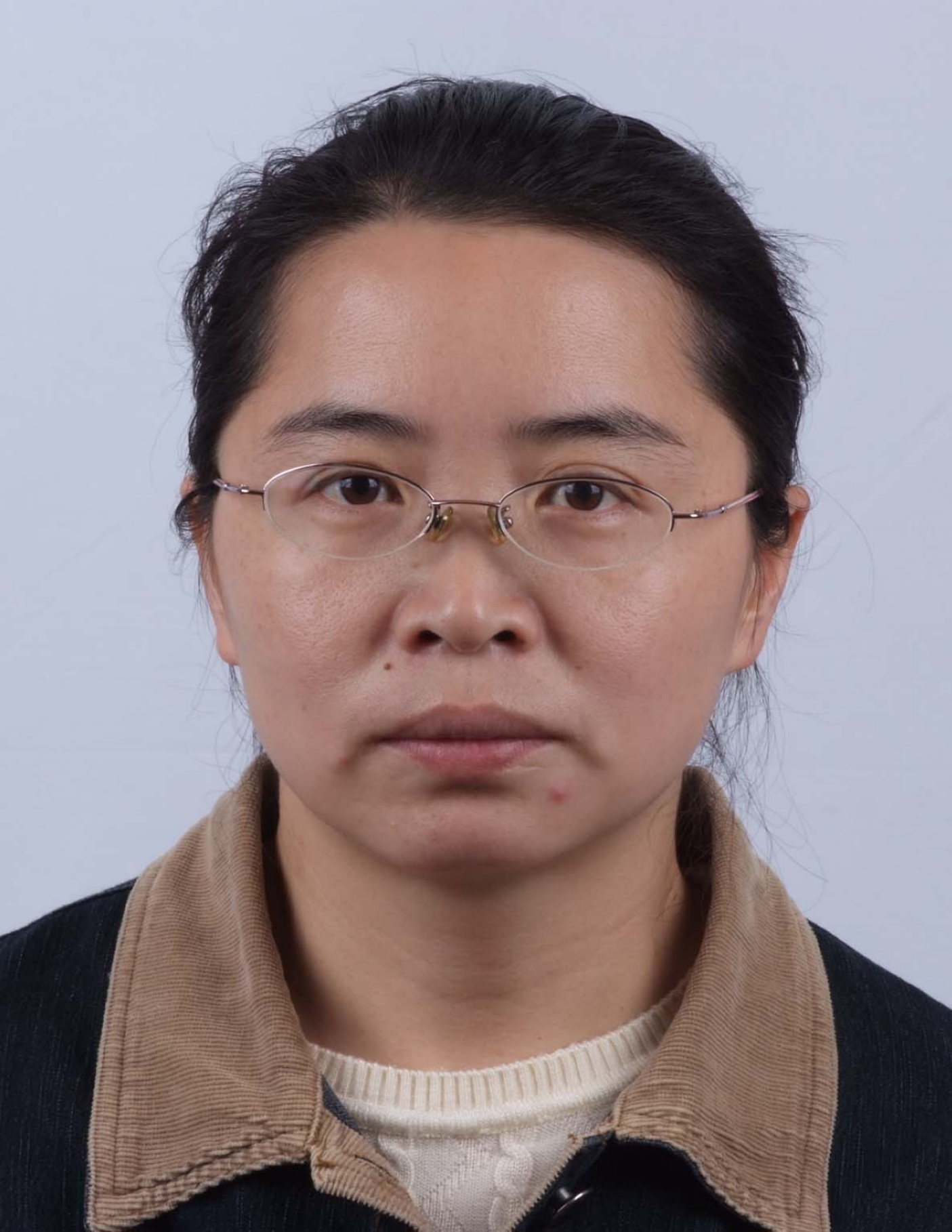}}]{Dongmei Zhang}
received the B.S. degree in institute of communications engineering, Nanjing, China in 1994, and the M.S. degree in communications and information system from the institute of communications engineering, Nanjing, China in 1997. She is currently an associate professor with the PLA University of Science and Technology, Nanjing, China. Her research interests include broadband wireless communications, cognitive radio, wireless resource allocation.
\end{IEEEbiography}

\begin{IEEEbiography}[{\includegraphics[width=1in,height=1.25in,clip,keepaspectratio]{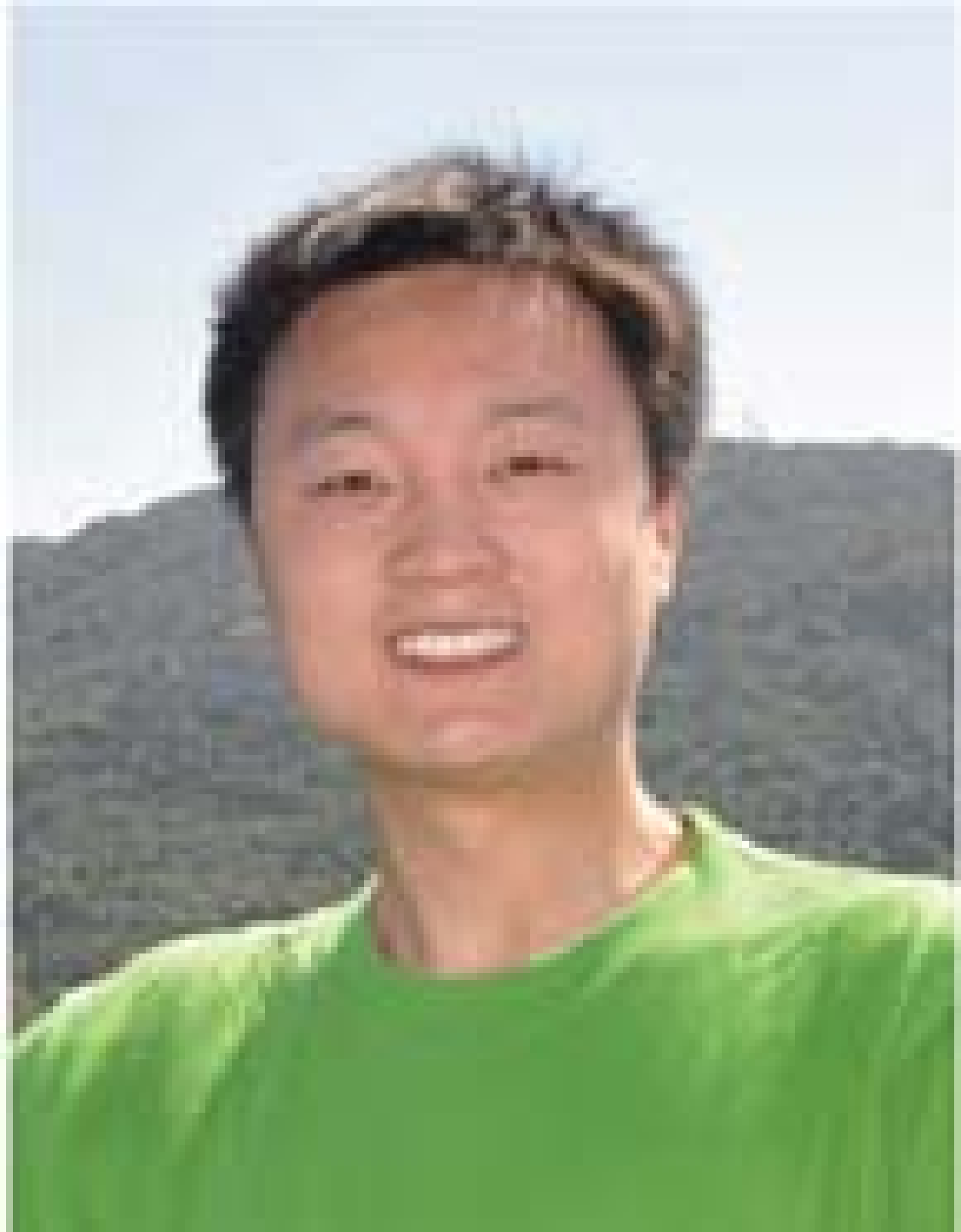}}]{Kui Xu}
received the B.S. degree and Ph.D. degree from the PLA University of Science and Technology (PLAUST), Nanjing, China in 2004 and 2009. He is currently a lecturer in the PLAUST. His research interests include broadband wireless communications, signal processing for communications and wireless communication networks. He is the author of about 50 papers in refereed journals and conference proceedings and holds 5 patents in China. Dr. Xu currently served on the technical program committee of the IEEE WCSP 2014. He received the URSI Young Scientists Award in 2014 and the 2010 ten excellent doctor degree dissertation award of PLAUST.
\end{IEEEbiography}

\begin{IEEEbiography}[{\includegraphics[width=1in,height=1.25in,clip,keepaspectratio]{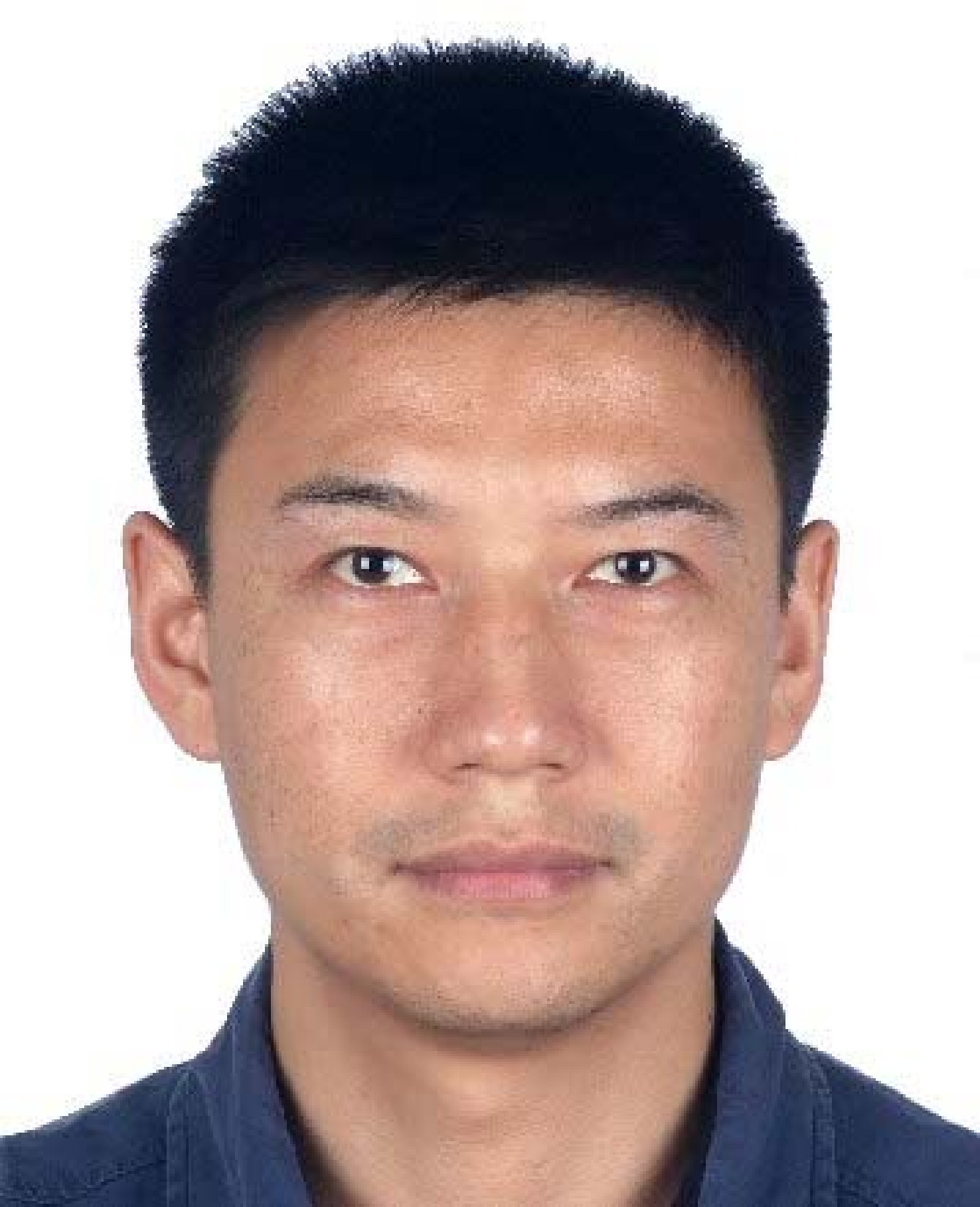}}]{Wenfeng Ma}
received the B.S. degree in microwave communication from the institute of communications engineering, Nanjing, China in 1995, and the M.S. degree in communications and information system from the institute of communications engineering, Nanjing, China in 1998, and the Ph.D. degree in communications and information system from the PLA University of Science and Technology, Nanjing, China, in 2002. He is currently an associate professor with the PLA University of Science and Technology, Nanjing, China. His research interests include wireless communication networks, broadband wireless communications.
\end{IEEEbiography}

\begin{IEEEbiography}[{\includegraphics[width=1in,height=1.25in,clip,keepaspectratio]{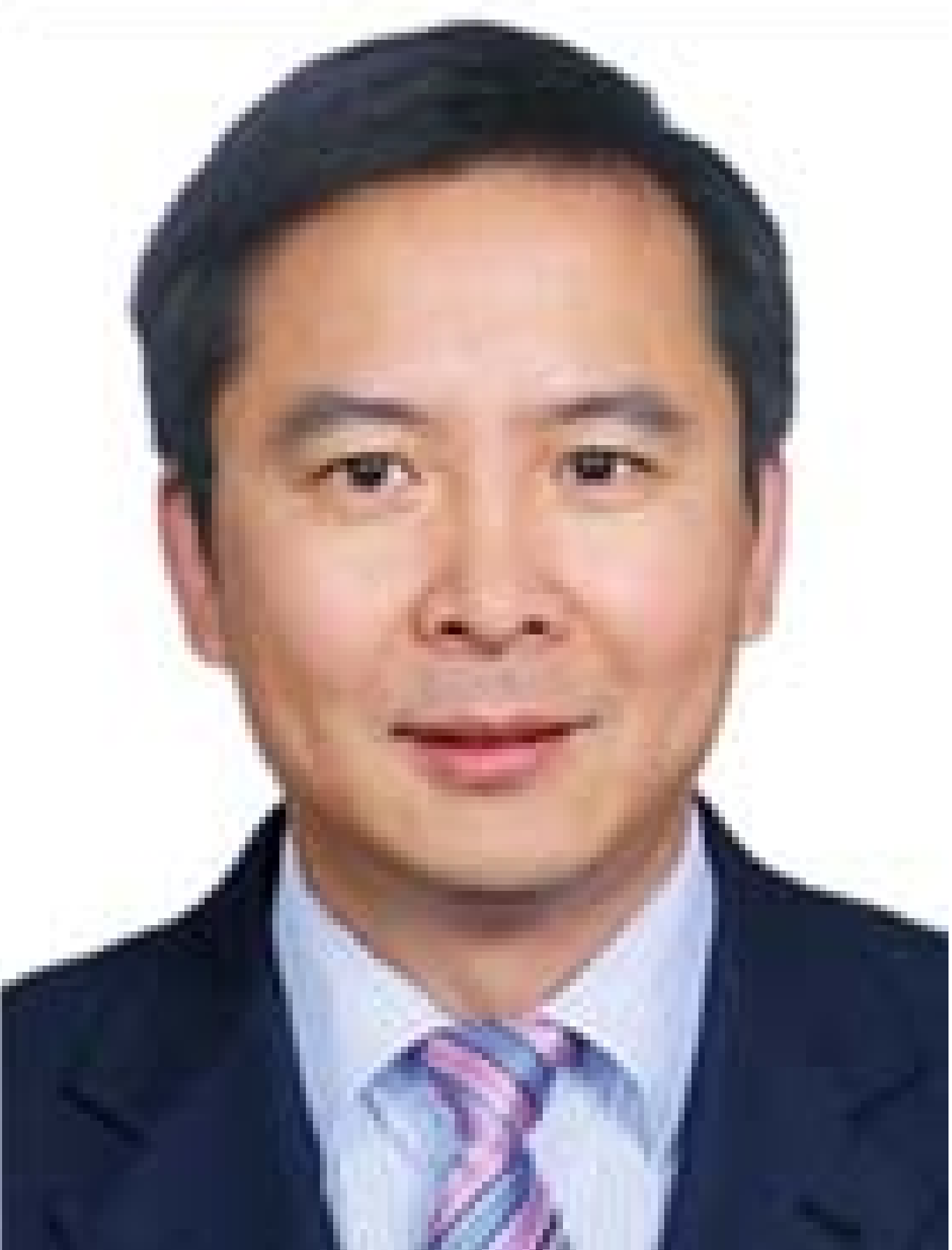}}]{Youyun Xu}
was graduated from Shanghai Jiao Tong University with Ph.D. Degree in Information and Communication Engineering in 1999. He is currently a professor with Nanjing Institute of Communication Engineering, China. He is also a part-time professor with the Institute of Wireless Communication Technology of Shanghai Jiao Tong University (SJTU), China. Professor Xu has more than 20 years professional experience of teaching and researching in communication theory and engineering. Now, his research interests are focusing on New Generation Wireless Mobile Communication System (IMT-Advanced and Related), Advanced Channel Coding and Modulation Techniques, Multi-user Information Theory and Radio Resource Management, Wireless Sensor Networks, Cognitive Radio Networks, etc.
\end{IEEEbiography}


\begin{thebibliography}{1}

\bibitem{IEEEhowto:Marz2010TWC}
T. L. Marzetta, ``Noncooperative cellular wireless with unlimited numbers of base station antennas,'' \emph{IEEE Trans. Wireless Commun.}, vol. 9, no. 11, pp. 3590-3600, Nov. 2010.

\bibitem{IEEEhowto:Ngo2013TC}
H. Q. Ngo, E. G. Larsson, and T. L. Marzetta, ``Energy and spectral efficiency of very large multiuser MIMO systems,'' \emph{IEEE Trans. Commun.}, vol. 61, no. 4, pp. 1436-1449, 2013.

\bibitem{IEEEhowto:BharSigcomm13}
D. Bharadia, E. McMilin, and S. Katti, ``Full duplex radios,'' \emph{ACM Sigcomm 2013}, Hong Kong, China, Aug. 2013.

\bibitem{IEEEhowto:EverettTWC13}
E. Everett, A. Sahai, and A. Sabharwal, ``Passive self-interference suppression for full-duplex infrastructure nodes,'' \emph{IEEE Trans. Wireless Commun.}, vol. 24, no. 2, pp. 680-694, 2014.

\bibitem{IEEEhowto:SuraweeraTW14}
H. A. Suraweera, I. Krikidis, G. Zheng, C. Yuen and P. J. Smith, ``Low-complexity end-to-end performance optimization in MIMO full-duplex relay systems,'' \emph{IEEE Trans. Wireless Commun.}, vol. 13, pp. 913-927, Feb. 2014.

\bibitem{IEEEhowto:ShangAusCTW14}
C. Y. A. Shang, P. J. Smith, G. K. Woodward and H. A. Suraweera, ``Linear transceivers for full duplex MIMO relays,'' in Proc. Australian Communications Theory Workshop (AusCTW 2014), Sydney, Australia, Feb. 2014, pp. 1722.

\bibitem{IEEEhowto:NguyenTSP13}
D. Nguyen, L.-N. Tran, P. Pirinen, and M. Latva-aho, ``Precoding for full duplex multiuser MIMO systems: spectral and energy efficiency maximization,'' \emph{IEEE Trans. Signal Process.}, vol. 61, no. 16, pp. 4038-4050, Aug. 2013.

\bibitem{IEEEhowto:SabharwalJSAC14}
A. Sabharwal, P. Schniter, D. Guo, D. W. Bliss, S. Rangarajan, and R. Wichman, ``In-band full-duplex wireless: challenges and opportunities,'' \emph{IEEE J. Sel. Areas Commun.}, vol. 32, no. 9, pp. 1637-1652, Sep. 2014.


\bibitem{IEEEhowto:CuiTWC14}
H. Cui, L. Song, and B. Jiao, ``Multi-pair two-way amplify-and-forward relaying with very large number of relay antennas,'' \emph{IEEE Trans. Wireless Commun.}, vol. 13, no. 5, pp. 2636-2645, 2014.

\bibitem{IEEEhowto:SuraICC13}
H. A. Suraweera, N. Hien Quoc, T. Q. Duong, and et al, ``Multi-pair amplify-and-forward relaying with very large antenna arrays,'' in \emph{IEEE ICC 2013}, pp. 4635-4640, Jun. 2013.

\bibitem{IEEEhowto:NgoICC13}
H. Q. Ngo, H. A. Suraweerat, E. G. Larsson, ``Multipair massive MIMO full-duplex relaying with MRC/MRT processing,'' \emph{IEEE ICC 2013},
pp. 4807-4813, Jun. 2013.

\bibitem{IEEEhowto:NgoJSAC14}
H. Q. Ngo, H. A. Suraweera, M. Matthaiou, and E. G. Larsson, ``Multipair full-duplex relaying with massive arrays and linear processing,'' \emph{IEEE J. Sel. Areas Commun.}, vol. 32, no. 9, pp. 1721-1737, Sep. 2014.

\bibitem{IEEEhowto:XXCWCNC14}
X. Xia, W. Xie, D. Zhang, and et al, ``Multi-pair full-duplex amplify-and-forward relaying with very large antenna arrays'' \emph{IEEE WCNC 2015}, New Orleans, USA, Mar. 2015.

\bibitem{IEEEhowto:LarssonCM14}
E. G. Larsson, O. Edfors, F. Tufvesson, and, and et al, ``Massive MIMO for next generation wireless systems'' \emph{IEEE Commu. Mag.}, vol. 52, no. 2, pp. 186-195, Feb. 2014.

\bibitem{IEEEhowto:DayJSAC12}
B. P. Day, A. R. Margetts, D. W. Bliss, and et al, ``Full-duplex MIMO relaying: Achievable rates under limited dynamic range,'' \emph{IEEE J. Sel. Areas Commun.}, vol. 30, pp. 1541-1553, 2012.

\bibitem{IEEEhowto:ZhengTWC13}
G. Zheng, I. Krikidis, and B. Ottersten, ``Full-duplex cooperative cognitive radio with transmit imperfections,'' \emph{IEEE Trans. Wireless Commun.}, vol. 12, pp. 2498-2511, 2013.

\bibitem{IEEEhowto:HoydisJSAC13}
J. Hoydis, S. Brink, M. Debbah, ``Massive MIMO in UL/DL of cellular networks: how many antennas do we need?'', \emph{IEEE J. Sel.
Areas Commun.}, vol. 31, pp. 160-171, Feb. 2013.

\bibitem{IEEEhowto:ZhangJSAC13}
J. Zhang, C. K. Wen, S. Jin, X. Gao, K. Wong, ``On capacity of large-scale MIMO multiple access channels with distributed sets of correlated antennas'' \emph{IEEE J. Sel. Areas Commun.}, vol. 31, pp. 133-148, Feb. 2013.

\bibitem{IEEEhowto:AdhikaryIT13}
A. Adhikary, J. Nam, J.-Y. Ahn, and G. Caire, ``Joint spatial division and multiplexing: the large-scale array regime,'' \emph{IEEE Trans. Inf.
Theory}, vol. 59, no. 10, pp. 6441-6463, Aug. 2013.

\bibitem{IEEEhowto:YinJSSP14}
H. Yin, D. Gesbert, L. Cottatellucci, ``Dealing with interference in distributed large-scale MIMO systems: a statistical
approach,'' \emph{IEEE J. Sel. Topics Signal Process.}, vol. 8, no. 5, pp. 942-953, Aug. 2014.

\bibitem{IEEEhowto:MohammedTC13}
S. Mohammed and E. Larsson, ``Per-antenna constant envelope precoding for large multi-user MIMO systems,'' \emph{IEEE Trans. Commun.}, vol. 61, no. 3, pp. 1059-1071, 2013.

\bibitem{IEEEhowto:StuderJSAC13}
C. Studer and E. Larsson, ``PAR-aware large-scale multi-user MIMO OFDM downlink,'' \emph{IEEE J. Sel. Areas Commun.}, vol. 31, no. 2, pp.
303-313, 2013.

\bibitem{IEEEhowto:KotechaTSP04}
J. H. Kotecha, and A. M. Sayeed, ``Transmit signal design for optimal estimation of correlated MIMO channels,'' \emph{IEEE Trans. Signal Process.}, vol. 52, no. 2, pp. 546-557, 2004.

\bibitem{IEEEhowto:SantellaVT98}
G. Santella and F. Mazzenga, ``A hybrid analytical-simulation procedure for performance evaluation in M-QAM-OFDM schemes in presence of
nonlinear distortions,'' \emph{IEEE Trans. Veh. Tech.}, vol. 47, pp. 142-151, Feb. 1998.

\bibitem{IEEEhowto:SuzukiJSAC08}
H. Suzuki, T. Tran, I. B. Collings, and et al, ``Transmitter noise effect on the performance of a MIMO-OFDM
hardware implementation achieving improved coverage,'' \emph{IEEE J. Sel. Areas Commun.}, vol. 26, pp. 867-876, Aug. 2008.

\bibitem{IEEEhowto:NamgoongTWC05}
W. Namgoong, ``Modeling and analysis of nonlinearities and mismatches in AC-coupled direct-conversion receiver,'' \emph{IEEE Trans. Wireless Commun.}, vol. 4, pp. 163-173, Jan. 2005.

\bibitem{IEEEhowto:Silverstein95}
J. W. Silverstein and Z. D. Bai, ``On the empirical distribution of eigenvalues of a class of large dimensional random matrices,'' \emph{Journal
of Multivariate Analysis}, vol. 54, no. 2, pp. 175-192, 1995.

\bibitem{IEEEhowto:WagnerIT12}
S. Wagner, R. Couillet, M. Debbah, and D. T. M. Slock, ``Large system analysis of linear precoding in correlated MISO broadcast channels
under limited feedback,'' \emph{IEEE Trans. Inf. Theory}, vol. 58, no. 7, pp. 4509-4537, Jul. 2012.

\bibitem{IEEEhowto:TelatarETT99}
E. Telatar, ``Capacity of multi-antenna Gaussian channels,'' \emph{European Trans. Telecom.}, vol. 10, no. 6, pp. 585-595, 1999.

\bibitem{IEEEhowto:BjCL13}
E. Bj$\ddot{\rm o}$rnson, P. Zetterberg, M. Bengtsson, and B. Ottersten, ``Capacity limits and multiplexing gains of MIMO channels with transceiver
impairments,'' \emph{IEEE Commun. Lett.}, vol. 17, no. 1, pp. 91-94, 2013.

\bibitem{IEEEhowto:ShomoronyIT14}
I. Shomorony and A. S. Avestimehr, ``Degrees of Freedom of Two-Hop Wireless Networks: Everyone Gets the Entire Cake,'' \emph{IEEE Trans. Inf. Theory}, vol. 60, no. 5, pp. 2417-2431, Mar. 2014.

\bibitem{IEEEhowto:ChengISIT14}
Z. Cheng, and N. Devroye, ``The degrees of freedom of the K-pair-user full-duplex two-way interference channel with a MIMO relay,'' \emph{in IEEE Int. Symp. Inf. Theory (ISIT)}, 2014, pp. 2714-2718.

\bibitem{IEEEhowto:QianIT11}
L. Qian, K. H. Li, and K. C. Teh, ``Achieving Optimal Diversity-Multiplexing Tradeoff for Full-Duplex MIMO Multihop Relay Networks,'' \emph{IEEE Trans. Inf. Theory}, vol. 57, no. 1, pp. 303-316, Jan. 2011.


\bibitem{IEEEhowto:Hassibi}
B. Hassibi and B. M. Hochwald, ``How much training is needed in multiple-antenna wireless links?'' \emph{IEEE Trans. Inf. Theory}, vol. 49, no. 4, pp. 951-963, Apr. 2003.


\bibitem{IEEEhowto:WeeraddanaVT11}
P. C. Weeraddana, M. Codreanu, M. Latva-aho, and A. Ephremides, ``Resource allocation for cross-layer utility maximization in wireless networks,'' \emph{IEEE Trans. Veh. Tech.}, vol. 60, pp. 2790-2809, 2011.

\bibitem{IEEEhowto:BoydOE2007}
S. Boyd, S.-J. Kim, L. Vandenberghe, and A. Hassibi, ``A tutorial on geometric programming,'' \emph{Optimization and Engineering}, vol. 8, no. 1, pp. 67-127, Mar. 2007.

\bibitem{IEEEhowto:BoydCO04}
S. Boyd and L. Vandenberghe, \emph{Convex Optimization. Cambridge}, U.K.: Cambridge Univ. Press, 2004.


\bibitem{IEEEhowto:LoykaCL01}
S. Loyka, ``Channel capacity of MIMO architecture using the exponential correlation matrix,'' \emph{IEEE Commun. Lett.}, vol. 5, no. 9, pp. 369-371, 2001.

\bibitem{IEEEhowto:HammarwallTSP09}
E. Bj$\ddot{\rm o}$rnson, D. Hammarwall, and B. Ottersten, ``Exploiting quantized
channel norm feedback through conditional statistics in arbitrarily correlated MIMO systems,'' \emph{IEEE Trans. Signal Process.}, vol. 57, no. 10, pp. 4027-4041, 2009.



\end{thebibliography}
\end{document}